\documentclass[11pt]{article}
\usepackage[utf8]{inputenc}
\usepackage{enumitem}
\usepackage[skins,xparse,breakable]{tcolorbox}
\usepackage{xcolor}
\usepackage{amsmath,amsfonts,amssymb}
\definecolor{linkblue}{HTML}{001487}
\usepackage[colorlinks=true,allcolors=linkblue]{hyperref}
\usepackage{caption}
\usepackage{subcaption}
\DeclareMathOperator{\img}{img}
\usepackage{bbm}
\usepackage[numbers]{natbib}
\usepackage{dsfont}
\usepackage{amsthm}
\usepackage{fullpage}

\newtheorem{theorem}{Theorem}[section]
\newtheorem{corollary}{Corollary}[theorem]
\newtheorem{lemma}[theorem]{Lemma}
\newtheorem{proposition}[theorem]{Proposition}

\theoremstyle{definition}
\newtheorem{definition}{Definition}

\newtheorem{fact}{Fact}

\newtheorem{assumption}[theorem]{Assumption}
\usepackage{thmtools} 
\usepackage{thm-restate}
\usepackage[margin=0.95in]{geometry}
\usepackage{authblk}
\usepackage{mathrsfs}  
\usepackage{mathtools}  
\usepackage{physics}  
\usepackage{tikz}
\usetikzlibrary{trees}
\usetikzlibrary{calc}
\usetikzlibrary{positioning}

\DeclareMathOperator{\poly}{\mathsf{poly}}
\DeclareMathOperator{\polylog}{\mathsf{polylog}}

\DeclareMathOperator{\mathrmsmallR}{\mathcal{R}_{\mathrm{small}}}
\DeclareMathOperator{\mathrmlargeR}{\mathcal{R}_{\mathrm{large}}}

\newcommand{\hi}{{\rm high}}
\newcommand{\lo}{{\rm low}}
\newcommand{\Mnew}{M_{\mathrm{new}}}

\usepackage{braket}

\newcommand{\lem}[1]{\hyperref[lem:#1]{Lemma~\ref*{lem:#1}}}
\newcommand{\defn}[1]{\hyperref[defn:#1]{Definition~\ref*{defn:#1}}}
\newcommand{\thmref}[1]{\hyperref[thmref:#1]{Theorem~\ref*{thmref:#1}}}
\newcommand{\factref}[1]{\hyperref[factref:#1]{Fact~\ref*{factref:#1}}}

\newcommand{\n}{\mathsf{m_2}}

\newcommand{\E}{\mathbb{E}}

\newcommand{\psiground}{\ket{\psi_{\mathrm{ground}}}}

\newcommand{\N}{\ensuremath{\mathds{N}}}

\newcommand{\bits}{\ensuremath{\{0, 1\}}}

\usepackage{epsfig}
\newcommand{\mysymbol}[1]{{\mbox{\raisebox{-0.3em}{\epsfysize=1.2em\epsfbox{#1}}}}}

\newcommand{\leftend}{\mysymbol{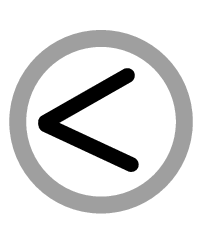}}
\newcommand{\rightend}{\mysymbol{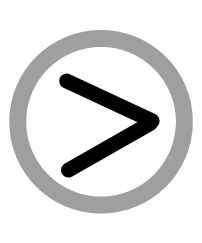}}
\newcommand{\variable}{\mysymbol{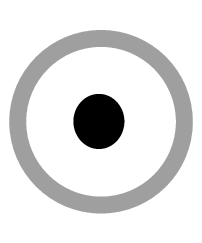}}

\newcommand{\arrR}{\mysymbol{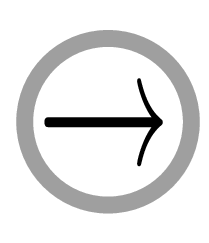}}
\newcommand{\arrRzero}{\mysymbol{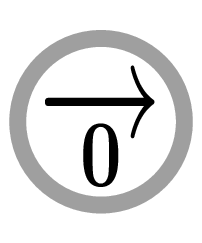}}

\newcommand{\arrL}{\mysymbol{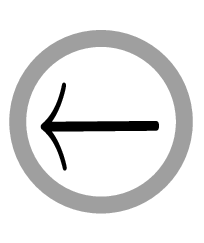}}
\newcommand{\arrLzero}{\mysymbol{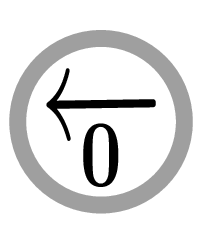}}

\newcommand{\blankR}{\mysymbol{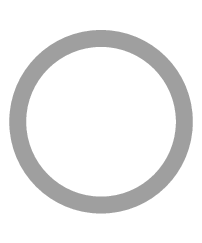}}
\newcommand{\blankL}{\mysymbol{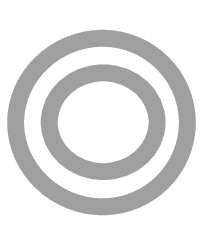}}

\newcommand{\twocellshor}[2]{\begin{array}{|@{}c@{}|@{}c@{}|} \hline  #1 & #2 \\ \hline  \end{array}}

\newcommand{\fourcells}[4]{\begin{array}{|@{}c@{}|@{}c@{}|} \hline  #1 & #3 \\ \hline #2 & #4 \\ \hline   \end{array}}

\newcommand{\sixcells}[6]{\begin{array}{|@{}c@{}| @{}c@{}|} \hline
          #1 & #2 \\ \hline
          #3 & #4 \\ \hline
          #5 & #6 \\ \hline
          \end{array}}

\newcommand{\sixcellsL}[4]{
  \begin{array}{|@{}c@{}|@{}c@{}|@{}c@{}|@{}c@{}|} 
    \hline
    \leftend & \begin{array}{@{}c@{}} #1 \\ \hline #2 \end{array} & \begin{array}{@{}c@{}} #3 \\ \hline #4 \end{array} \\ 
    \hline  
  \end{array}
}
 
\newcommand{\sixcellsR}[4]{
  \begin{array}{|@{}c@{}|@{}c@{}|@{}c@{}|@{}c@{}|} 
    \hline
    \begin{array}{@{}c@{}} #1 \\ \hline #2 \end{array} & \begin{array}{@{}c@{}} #3 \\ \hline #4 \end{array} & \rightend \\ 
    \hline  
  \end{array}
}

\newcommand{\sixcellsLR}[6]{
  \begin{tabular}{|@{}c@{}| @{}c@{}| @{}c@{}| @{}c@{}| @{}c@{}| @{}c@{}|} 
    \hline
    \leftend &
    \begin{tabular}{@{}c@{}} #1 \\ \hline #4 \end{tabular} &
    \begin{tabular}{@{}c@{}} #2 \\ \hline #5 \end{tabular} &
    \begin{tabular}{@{}c@{}} #3 \\ \hline #6 \end{tabular} &
    \rightend \\
    \hline  
  \end{tabular}
}

\newcommand{\ninecellsLR}[9]{
  \begin{tabular}{|@{}c@{}| @{}c@{}| @{}c@{}| @{}c@{}| @{}c@{}| @{}c@{}|} 
    \hline
    \leftend &
    \begin{tabular}{@{}c@{}} #1 \\ \hline #4 \\ \hline #7 \end{tabular} &
    \begin{tabular}{@{}c@{}} #2 \\ \hline #5 \\ \hline #8\end{tabular} &
    \begin{tabular}{@{}c@{}} #3 \\ \hline #6 \\ \hline #9\end{tabular} &
    \rightend \\
    \hline
  \end{tabular}
}

\newcommand{\sixcellsvertical}[6]{
  \begin{array}{|@{}c@{}|}
    \hline
    \begin{array}{@{}c@{}} #1 \\ \hline #2 \end{array} \\ \hline
    \begin{array}{@{}c@{}} #3 \\ \hline #4 \end{array} \\ \hline
    \begin{array}{@{}c@{}} #5 \\ \hline #6 \end{array} \\ \hline
  \end{array}
}

\title{\adam{brainstorming titles}\\On the hardness of learning ground state entanglement of geometrically local Hamiltonians\\
1D local gapless phases are classically hard to learn}

\title{On the hardness of learning ground state entanglement of geometrically local Hamiltonians}

\author[1]{Adam Bouland\footnote{abouland@stanford.edu}}
\author[1]{Chenyi Zhang\footnote{chenyiz@stanford.edu}}
\author[1]{Zixin Zhou\footnote{jackzhou@stanford.edu}}
\affil[1]{Department of Computer Science, Stanford University}

\date{}
\begin{document}

\maketitle
\vspace{-3em}
\begin{abstract}
Characterizing the entanglement structure of ground states of local Hamiltonians is a fundamental problem in quantum information. In this work we study the computational complexity of this problem, given the Hamiltonian as input. Our main result is that to show it is cryptographically hard to determine if the ground state of a \emph{geometrically local}, polynomially gapped Hamiltonian on qudits ($d=O(1)$) has near-area law vs near-volume law entanglement. This improves prior work of Bouland et al. (arXiv:2311.12017) showing this for non-geometrically local Hamiltonians. In particular we show this problem is roughly factoring-hard in 1D, and LWE-hard in 2D. Our proof works by constructing a novel form of public-key pseudo-entanglement which is highly space-efficient, and combining this with a modification of Gottesman and Irani’s quantum Turing machine to Hamiltonian construction. Our work suggests that the problem of learning so-called “gapless” quantum phases of matter might be intractable.
\end{abstract}

\section{Introduction}
Characterizing the entanglement structure of ground states of local Hamiltonians is a fundamental problem in quantum information.
Many works have studied when the structure of local Hamiltonians forces the ground state to have simple entanglement structures such as an area law.
For example, an area law for gapped 1D local Hamiltonians has been rigorously established \cite{hastings2007area}, and much progress has been made towards a 2D area law as well e.g.~\cite{anshu2022area}.
This is closely connected to Hamiltonian complexity, as areas laws can enable efficient algorithms for learning and describing ground states, such as through matrix product states and PEPS,  e.g.~\cite{landau2015polynomial}.
It is also closely connected to questions in condensed matter physics, where different ground state entanglement structures can be hallmarks of different quantum phases of matter, such as topological phases \cite{sachdev2023quantum}.

In this work we study the complexity of learning the ground state entanglement structure of a local Hamiltonian, given the Hamiltonian as input.
This is known as the ``Learning Ground State Entanglement Structure (LGSES)'' problem \cite{bouland2024public}.
Informally the LGSES problem captures the complexity of the job of a condensed matter physicist -- as oftentimes one writes down a Hamiltonian and applies highly nontrivial techniques (Bethe ansatze, quantum Monte Carlo algorithms, etc.) to deduce properties of its low energy states. 
The difficulty of LGSES depends significantly on the assumptions made about the structure of the Hamiltonian.
For example, if one assumes a constant spectral gap, then in many cases the area law forces the ground state to have a particularly simple structure, rendering the LGSES problem trivial.
It has also recently been shown that with geometric locality and constant spectral gaps, one can efficiently learn  properties of ground states, such as expectation values of local observables \cite{huang2022provably} if one is given labelled examples from members of a broader family\footnote{Their result also requires a certain smoothness condition on the Hamiltonians considered, see \cite{huang2022provably}.}.
On the other hand, if the Hamiltonian has an inverse polynomial spectral gap, the problem can become much harder.
For example, Bouland et al.~\cite{bouland2024public} recently showed it is LWE-hard to learn if the ground state of a general local Hamiltonian has near 1D area law or near volume law entanglement.
However, the Hamiltonians in their construction were not geometrically local, and hence not directly relevant to physics.
With geometric locality, they were only able to obtain much weaker hardness results about detecting different variants of sub-volume law entanglement.
A natural question is if it is hard to distinguish near area law vs.~volume law entanglement in the \emph{geometrically local} setting.

\textbf{Our results:} In this work we show that LGSES is cryptographically hard even with \emph{geometrically local} Hamiltonians, and even for determining if the ground state is near-area law vs.~near volume law:
We show this for 1D and 2D-local Hamiltonians on qudits where the local dimension $d=O(1)$:

\begin{theorem}[Informal]
    It is hard for a classical computer to determine if a 1D local Hamiltonian on qudits has a ground state with 1D near area law vs.~near volume law entanglement.
\end{theorem}
\begin{theorem}[Informal]
    It is hard for a quantum computer to determine if a 2D local Hamiltonian on qudits has a ground state with 2D near area law vs.~near volume law entanglement.
 \end{theorem}
Here, near area law means that if $\rho$ is the reduced density matrix on a subset of qubits, then $S(\rho)\leq |A|\text{polylog}(n)$ where $A$ is the area of the cut.
The Hamiltonians have inverse polynomial spectral gaps.
The first result is shown under the Decisional Composite Residuosity Assumption (DCRA), which is a standard classical hardness assumption based on discrete log/factoring. In particular we assume DCRA does not have a $2^{n^\epsilon}$ time algorithm for some value of $\epsilon>0$.
The second result is shown assuming subexponential hardness of LWE.

\subsection{Proof sketch}
We will describe the 1D proof, and later describe how to lift to 2D. Our proof uses on a new construction of public-key pseudo-entanglement.
Here the idea is that one can construct ensembles of states for which it is difficult to estimate their entanglement, even given the description of the quantum circuit used to prepare the state.
Such constructions imply hardness results for LGSES by passing the circuit ensembles through circuit to Hamiltonian constructions \cite{gheorghiu2020estimating,bouland2024public}.

The starting point of our work is to try to take the public-key pseudo-entanglement construction of \cite{bouland2024public} and pass them through a geometrically local circuit to Hamiltonian construction, in particular the 1D translation invariant construction of Gotteman and Irani \cite{gottesman2009quantum}.
The Gottesman-Irani construction, however, is formulated in a quantum Turing machine (QTM) model of computation, rather than the quantum circuit model.
That is, give a QTM as input, it describes a Hamiltonian whose ground state on $n$ qubits encodes the computation performed by that QTM on input $1^n$.

Therefore, to apply this approach we first need to formulate a QTM variant of pseudo-entanglement, which causes several issues.
First, we need to hard-code a cryptographic key into the construction.
This is because the the public key enabling the preparation of the pseudo-entangled states must be provided as input to the QTM.
This breaks the translation invariance of \cite{gottesman2009quantum}. Second, it requires showing the cryptographic primitives used in \cite{bouland2024public} can be implemented in the QTM model, and not just the circuit model.
Both of these changes require tedious modifications of the Gottesman-Irani construction, but are surmountable with sufficient attention to detail.

However, the larger issue that we face in combining the constructions is the \emph{space complexity} of our cryptographic primitives.
In our setting our goal is not to show $\textsf{QMA}$-hardness of the ground state energy, but rather to show cryptographic hardness of learning the ground state entanglement.
Therefore the key thing we wish to preserve in our circuit to QTM to Hamiltonian construction is the entanglement structure of the ground state.
In turns out that for \cite{gottesman2009quantum} to preserve entanglement structure, we need our pseudo-entanglement construction to run not only in polynomial time. but also in \emph{near linear} space.
This is because the space complexity of the QTM roughly equals the number of qudits in the chain, so a blowup say to $O(n^2)$ space effectively ``dilutes'' the entanglement in the ground state, as informally these $O(n)$ e-bits of entanglement are then ``spread out'' amongst $O(n^2)$ qudits, meaning that states with high entanglement in the circuit are no longer highly entangled.
We therefore can't afford more than a logarithmic blowup in the space complexity of the algorithm.

This is immediately a problem because the key sizes used in \cite{bouland2024public} are much larger than superlinear -- the total key size is actually $O(n^2 \log n)$. This key size is coming from the fact that public-key cryptographic primitives built from LWE typically have key sizes of $O(n^2)$ or larger -- this is because the public key is usually a dense matrix $A\in \mathbb{F}_q^{n\times n}$. 
Therefore it seems extremely difficult to build public-key pseudo-entanglement with linear key size from LWE.
There is also a second source of high space complexity in \cite{bouland2024public}, coming from the use of $n$ different pairwise-independent hash functions, which would itself contribute a quadratic sized key even if a different crypto primitive was slotted in place of LWE.

To address this we create a new construction of pseudo-entanglement which has a much smaller public key size while still hiding a big entanglement gap.
First, we relax the problem to consider only classical hardness.
This enables cryptographic public-key constructions with near-linear key sizes in the input size -- particularly a cryptographic object known as a lossy function \cite{peikert2008lossy} which allows one to ``hash down'' entanglement across a particular cut of the qubits \cite{bouland2024public}. 
This alone isn't enough to fix the problem, because the prior construction had $O(n)$ independent cryptographic keys plus $O(n)$ hash function keys.
This arises because the prior construction did a 1D ``sweep'' of the line, using lossy functions to ``tamp down'' entanglement across each cut of the line in the corresponding binary phase state.
To fix this, we make a new construction which tamps down entanglement in a tree-like structure.
First we cut down entanglement across the middle cut between the left $n/2$ and right $n/2$ qubits; then we cut down across the next cuts of size $n/4$, etc.~(see Figure \ref{fig:hashing-tree-structure}).
This still involves $O(n\log n)$ applications of lossy functions, so naively would not reduce the key size.
The crucial insight is that, first, the keys can be \emph{the same} across each level of the tree, and second, the key size at each level can be a \emph{decreasing function} of $n$ (as the security parameter coming from algorithmic tests of entanglement also shrinks with $n$).
This allows us to obtain a total key size of $O(n\log n)$ across all of these lossy functions.
Showing this works requires developing a new proof that tamping down entanglement across one cut doesn't ``blow up'' your entanglement across another, which in \cite{bouland2024public} had followed from a simple property of 1D sweeps.
It also requires further changes to the construction, as the classical lossy functions actually blow up the input size slightly, but crucially this blowup is tiny in the low-deoth tree, so does not result in much ``dilution'' of the entanglement. This allows us to remove the pairwise-independent hash functions used in prior work, which had served to avoid this blowup in input size.
To lift this to 2D with LWE, to show that here the quadratic blowup in key size is acceptable -- we use a larger key, and then make a ground state which is a superposition of the 1D states striped horizontally vs vertically.

\vspace{-0.5em}
\subsection{Discussion and open problems}
\vspace{-0.5em}

Our work is the first to show it can be hard to learn the ground state entanglement structure of geometrically local Hamiltonians, even promised near-maximal gaps in entanglement structure.

A number of open problems remain.
First, can we show hardness of LGSES for translationally invariant Hamiltonians?
Such a result would parallel the development $\textsf{QMA}$-hardness in Hamiltonian complexity \cite{kempe20033local,oliveira2005complexity,aharonov2009power,gottesman2009quantum}.
However, we believe this would require developing fundamentally new techniques.
This is because in 1D translation-invariant constructions, the QTMs are not allowed to have a binary input -- rather their input is written in unary, which allows one to show $\textsf{QMA}_\textsf{EXP}$ hardness.
While a QTM can generate an instance of our public-key pseudo-entanglement from $1^n$, the algorithm will generate not only the public key (i.e.~the description of how to prepare the state), but also the private key which allows one to distinguish the two ensembles.
In our work, the ground state can be efficiently constructed from the Hamiltonian, so this means one can efficiently extract the private key from the ground state, breaking the security.
As a result, one needs to go beyond pseudo-entanglement + circuit/QTM to Hamiltonian constructions if one wishes to show hardness of translationally invariant systems.
In a similar spirit, there is the also the question of if we can reduce our local dimension, which has been done in the case of $\textsf{QMA}$, e.g. \cite{hallgren2013local,bausch2017complexity}, and if we can show the 1D case is hard based on post-quantum cryptography as well.

Second, there is a question about how to interpret our result through the lens of condensed matter theory.
As different entanglement structures can be hallmarks of different quantum phases of matter, our results could be interpreted as a statement about hardness of learning phases of matter.
Our results are shown with inverse polynomial spectral gaps, so could be saying that it is hard to learn ``gapless'' phases of matter (where gap vanishes as system size grows), which stands in contrast to gapped phases which might be easier to learn \cite{huang2022provably}. Our result also complements recent work of Schuster Haferkamp and Huang \cite{schuster2024random}, a corollary of which was that it can be hard to learn topological phases of matter.
However we note our result is proven in a stronger public-key input model, as our input includes the recipe for preparing the state, whereas their corollary is in a private-key input model where one is only given copies of the state itself.

\subsection{Outline of the paper}
We introduce the notations and relevant existing results in Section~\ref{sec:preliminaries}. In Section~\ref{sec:lossy-function-constructions}, we present two lossy function constructions based on DCRA and LWE respectively, which we use in Section~\ref{sec:public-key-pseudoentangled-states} give two public-key pseudoentangled states constructions. The quantum Turing machines for these two constructions are given in Section~\ref{sec:TM-public-key-pseudoentanglement}. In Section~\ref{sec:QTM-to-Hamiltonian}, we present a general framework of encoding the running history of a QTM into a local Hamiltonian. Finally in Section~\ref{sec:everything-together}, we combine all these techniques together and prove the hardness of LGSES for both 1D Hamiltonians and 2D Hamiltonians.

\section{Preliminaries}\label{sec:preliminaries}

\subsection{Notation} \label{sec:notation}
We write $[n]$ for the set $\{1, \dots, n\}$ and write $[a,b]$ for the set $\{a,a+1,\ldots,b-1,b\}$. We denote the concatenation of strings by $x \parallel y$. For a set of indices $I \subseteq [n]$ and bitstrings $x \in \{0,1\}^{|I|}$, $y \in \{0,1\}^{n-|I|}$ we denote by $z = x \parallel_I y$ the string $z$ that equals $x$ in indices in $I$ and $y$ on indices in $[n]\setminus I$. For a quantum state $\ket{\psi}$ on $n$ qubits and a subset $I\in[n]$, we denote the reduced states on the subsystem of qubits in $I$ by $\psi_I$ and $\rho_{I}$, respectively. We denote the entanglement entropy between systems $I$ and $[n]\setminus I$ is defined as $S(\psi_I)$.
Note that this is invariant under swapping $I$ and $[n]\setminus I$ since for a pure state $\ket{\psi}$, $S(\psi_I) = S(\psi_{[n]\setminus I})$. We denote $\text{negl}(n)\coloneqq o\left(1/p(n)\right)\text{ for all polynomials } p(n)$.

\subsection{Independent hash functions}

\begin{definition}[$r$-wise independent function family]
A function family $H = \{h_k: [N] \to [M]\}_{k \in \mathcal{K}}$, indexed by a set of keys $\mathcal{K}$, is said to be $r$-wise independent if, for any distinct $x_1, \dots, x_r \in [N]$, the random variables $h_k(x_1), \dots, h_k(x_r)$ (where $k \in \mathcal{K}$ is chosen uniformly at random) are independent and uniformly distributed.
\end{definition}
The existence of $r$-wise independent functions has been well-established in the literature; see, for example, \cite[Construction 3.32, Corollary 3.34]{vadhan2012pseudorandomness}.

\begin{lemma} \label{lem:r-wise_exist}
For any $n, m, r \in \N$, there exists an $r$-wise independent function family $H_n = \{h_k: \{0, 1\}^n \to \{0, 1\}^m\}_{k \in \mathcal{K}}$ such that each key $k \in \mathcal{K}$ has length $r  \cdot \max(n, m)$, and given any $k \in \mathcal{K}$, the function $h_k$ can be computed in time $\poly(n, m, r, \log q)$, and space $O(r \cdot \max(n, m))$. 
\end{lemma}

\subsection{Quantum Turing machine}
Introduced in~\cite{deutsch1985quantum} and further developed in~\cite{bernstein1997quantum}, quantum Turing machines (QTMs) are quantum analogues of the classical Turing machines where their states can be quantum superpositions, and the transition function is a unitary operation. Further, \cite{yao1993quantum} showed that QTMs are polynomially equivalent to uniformly generated quantum circuits, which are families of quantum circuits whose descriptions can be efficiently produced by classical Turing machines. Formally, a QTM is defined as a 5-tuple \((Q, \Sigma, \delta, q_0, q_{\mathrm{acc}})\),

\begin{itemize}
    \item \(Q\) is a finite set of states.
    \item \(\Sigma\) is a finite alphabet of tape symbols, which includes a special blank symbol \(\#\).
    \item \(\delta: Q \times \Sigma \to \mathbb{C}^{Q \times \Sigma \times \{L, R\}}\) is the quantum transition function, where \(L\) and \(R\) represent the left and right movements of the tape head. We define $\alpha_{\delta}\colon Q\times\Sigma\times Q\times\Sigma\times\{L,R\}$ to be the amplitude of the transition function $\delta$, i.e.~,
    \begin{align}\label{eqn:amplitude-of-delta}
    \delta\ket{q,a}=\sum_{a,b,q_L,q}\alpha_{\delta}(q,a,q_L,b,L)\ket{q_L,b,L}+\sum_{a,b,q_R,q}\alpha_{\delta}(q,a,q_R,b,R)\ket{q_R,b,R}.
    \end{align}
    \item \(q_0 \in Q\) is the initial state.
    \item \(Q_{\text{acc}} \subseteq Q\) is the set of accepting states.
\end{itemize}

It is convenient to think of a single-tape quantum Turing machine as consisting of multiple tracks. The alphabet $\Sigma$ of a multitrack Turing machine with $k$ tracks is the product of $k$ alphabets, one for each track: $\Sigma = \Sigma_1 \times \Sigma_2 \times \cdots \times \Sigma_k$. Additionally, the blank symbol in $\Sigma$ is represented as $(\#, \#, \ldots, \#)$. This notation allows for a clear separation of different data on the tape and simplifies the definition of certain Turing machine operations that involve only data on specific tracks. The space complexity of a QTM is defined as the number of cells it ever touches during the computation.

\subsection{Entanglement measures and properties}
We recall the standard definitions of quantum entropies and adopt the convention that $0 \log 0 = 0$ throughout.

\begin{definition}[von Neumann entropy]
The von Neumann entropy of a quantum state $\rho$ is defined as 
\begin{align*}
S(\rho) = - \tr{\rho \log \rho} \,.
\end{align*}
\end{definition}
For pure states on systems $AB$, the entanglement between $A$ and $B$ is quantified using the so-called entanglement entropy, which is simply the von Neumann entropy of the reduced state on either subsystem.
\begin{definition}
For a \emph{pure} state $\ket{\psi}_{AB}$, the entanglement entropy between systems $A$ and $B$ is defined as $S(\psi_A)$.
Note that this is invariant under swapping $A$ and $B$ since for a pure state $\ket{\psi}_{AB}$, $S(\psi_A) = S(\psi_B)$.
\end{definition}

\subsubsection{Continuity properties}

\begin{definition}[Binary entropy function]
The binary entropy function is defined as
\begin{equation*}
h(x) = - x \log x - (1-x) \log (1-x),
\end{equation*}
for $x \in [0,1]$.
\end{definition}
\begin{lemma}[Continuity of the conditional von Neumann entropy \cite{winter2016}]
\label{lem:continuity_vonneumannentropy}
Let $\rho_{AB}$ and $\sigma_{AB}$ be the density matrix of two $n$-qubit quantum states respectively, each partitioned into subsystems $A$ and $B$, and let
\begin{equation*}
\frac{1}{2} ||\rho_{AB} - \sigma_{AB}||_1 \leq \epsilon. 
\end{equation*}
Then, 
\begin{equation*}
|S(\rho_{A}) - S(\sigma_{A})| \leq 2 \epsilon \cdot \log |A| + (1 + \epsilon)\cdot h\left(\frac{\epsilon}{1+\epsilon}\right),
\end{equation*}
where $|A|$ is the dimension of the Hilbert space for subsystem $A$ and $h(\cdot)$ is the binary entropy function.
\end{lemma}

\subsubsection{Entanglement entropy for phase states}
\begin{definition}[$T$-matrix associated with phase states, Definition 2.9 of \cite{bouland2024public}] \label{def:t-matrix}
Let $s: \{0,1\}^n \to \{0,1\}$.
For an $n$-qubit phase state 
\begin{align*}
\ket{\psi} = \frac{1}{2^{n/2}}\sum_{x} (-1)^{s(x)} \ket{x}
\end{align*} 
and a subset $I \subseteq [n]$, we define the ``$T$-matrix'' with respect to the cut $I$ as a $\{\pm 1\}^{2^{|I|} \times 2^{n - |I|}}$-matrix with entries
\begin{align*}
T_{ij} = (-1)^{s(i \parallel_I j)} \,,
\end{align*}
where $\parallel_I$ is the ``index string concatenation'' defined in Section~\ref{sec:notation}.
\end{definition}

\begin{lemma}[Lemma 2.10 of \cite{bouland2024public}] \label{lem:entanglement_from_matrix}
Let $s: \{0,1\}^n \to \{0,1\}$ and $\ket{\psi} = \sum_{x} (-1)^{s(x)} \ket{x}$.
Then for any cut $I \subseteq [n]$, the entanglement entropy of that cut is bounded by 
\begin{align*}
-\log \left( \norm{\frac{1}{2^{n}} T T^{\top}}_2 \right) \leq S(\psi_I) \leq \log \rank(T) \,,
\end{align*}
where $T$ is the $T$-matrix of $\ket{\psi}$ across cut $I$.
\end{lemma}

\subsubsection{Cutwise orthogonality}
\begin{definition}[Cutwise orthogonality, Definition 4.12 of~\cite{bouland2024public}]
For any \( n > 0 \), let \(\ket{\psi_0}\) and \(\ket{\psi_1}\) be two \(n\)-qudit states. For any bipartition \((A, B)\), we say that \(\ket{\psi_0}\) and \(\ket{\psi_1}\) are \textit{cutwise orthogonal} with respect to $(A,B)$ if the following condition holds:
\begin{align*}
\rho_0^{(A)} \rho_1^{(A)} = \rho_0^{(B)} \rho_1^{(B)} = 0,
\end{align*}
where \(\rho_0^{(A)}\) and \(\rho_0^{(B)}\) are the reduced density matrices of \(\ket{\psi_1}\) on subsystems \( A \) and \( B \), respectively, and \(\rho_1^{(A)}\) and \(\rho_1^{(B)}\) are the reduced density matrices of \(\ket{\psi_1}\) on subsystems \( A \) and \( B \), respectively.
\label{def:cutwise-orthogonality}
\end{definition}

\begin{lemma}[Lemma 4.13 of~\cite{bouland2024public}]\label{lem:cutwise-orthogonal-entanglement}
For any two $n$-qudit states $\ket{\psi_0}$ and $\ket{\psi_1}$ and a bipartition $(A,B)$ such that $\ket{\psi_0}$ and $\ket{\psi_1}$ are cutwise orthogonal with respect to $(A,B)$, the entanglement of $\frac{1}{\sqrt{2}}(\ket{\psi_0}+\ket{\psi_1})$ with respect to $(A,B)$ equals
\begin{align}
\frac{1}{2}\big(S(\phi_{0,A})+S(\phi_{1,A})\big)+\log 2.
\end{align}
\end{lemma}

\subsection{Local Hamiltonians}
\begin{definition}[$1$D and 2D (local) Hamiltonian]\label{def:local-hamiltonians}
Let $H$ be a Hermitian operator (interpreted as a Hamiltonian, giving the energy of some system). We say that $H$ is an $r$-state Hamiltonian if it acts on $r$-state qudits (i.e.~$d=r$). When $r=2$, namely, when the qudits are qubits. We say that $H$ is $k$-local if it can be written as $$H = \sum_i H_i,$$ where each $H_i$ acts non-trivially on at most $k$ qudits. Note that this term does not assume anything about the physical location of the qudits. We say that $H$ is a 1D Hamiltonian if the qudits are arranged on a 1D line and the terms $H_i$ interact only pairs of nearest neighbor qudits. We say that $H$ is a 2D Hamiltonian if the qudits are arranged on a 2D grid and the terms $H_i$ interact only pairs of nearest neighbor qudits. 
\end{definition}

\section{Construction of lossy functions}\label{sec:lossy-function-constructions}

In this section, we present two constructions of lossy functions based on DCRA and LWE, respectively. They are the key technical component in the construction of public-key pseudoentangled states in Section~\ref{sec:public-key-pseudoentangled-states}.

We begin by reviewing the construction of lossy functions by Rosen and Segev in~\cite{rosen2008efficient}. This construction is based on the Damgård–Jurik cryptosystem scheme~\cite{damgaard2001generalisation}. A class of lossy functions is defined by two indistinguishable function families: one consisting of injective functions, and the other of lossy functions, which are characterized by having a significantly smaller image size.  

\begin{definition}[Lossy functions] \label{def:lossyfunction}
    Let $n$ be the input length, $\ell$ be the lossiness parameter. A collection of $(n, \ell)$-lossy functions $\{ f_k(\cdot) : \{ 0, 1\}^n \to \{ 0, 1\}^*  \}$ indexed by key $k \in \mathcal{K}^{\mathrm{inj}} \cup \mathcal{K}^{\mathrm{lossy}}$ is defined by a pair of (possibly probabilistic) polynomial-time algorithms $(G, F)$ such that:
    \begin{enumerate}
        \item $G(1^n, \textsf{injective})$  samples a key $k \in \mathcal{K}_n^{\mathrm{inj}}$, and for $k \in \mathcal{K}_n^{\mathrm{inj}}$, $F(k, \cdot)$ deterministically computes an injective function $f_k(\cdot)$ over the domain $\{0, 1\}^n$.
        \item $G(1^n, \textsf{lossy})$  samples a key $k \in \mathcal{K}_n^{\mathrm{lossy}}$, and for $k \in \mathcal{K}_n^{\mathrm{lossy}}$, $F(k, \cdot)$ deterministically computes an function $f_k(\cdot)$ over the domain $\{0, 1\}^n$ whose image has size at most $2^{n - \ell}$.
        \item Keys from $G(1^n, \textsf{injective})$ and $G(1^n, \textsf{lossy})$ are computationally indistinguishable. Formally, for all $\poly(n)$-time quantum adversaries $\mathcal{A}$ that take as input a key $\mathcal{K}_n^{\mathrm{inj}} \cup \mathcal{K}_n^{\mathrm{lossy}}$ and output a single bit: 
\begin{align*}
\Big| \Pr_{k \leftarrow \mathcal{K}_n^{\mathrm{inj}}}{\mathcal{A}(k) = 0} - \Pr_{k \leftarrow \mathcal{K}_n^{\mathrm{lossy}}}{\mathcal{A}(k) = 0} \Big| = \mathrm{negl}(n) \,.
\end{align*}
    \end{enumerate}
\end{definition}

\subsection{DCRA based lossy function}\label{sec:DCRA-lossy-function}
In this construction, each key (including the injective ones and the lossy ones) is described by a triplet of integers $(N, s, c)$, where $N$ is an admissible RSA modulus and $c \in \mathbb{Z}_{N^{s + 1}}$.

Specifically, for input length $n$, parameter $\lambda$, and any positive integer number $s$ such that $2\lambda s > n$, the lossy function $f^{n, \lambda, s}_k (\cdot)$ are defined as follows:
\begin{enumerate}

    \item \textbf{Sampling an injective function.} To sample $k \in \mathcal{K}_n^{\mathrm{inj}}$, the algorithm randomly samples a $\lambda$-bit admissible RSA modulus $N = PQ$ such that $N^s \ge 2^n$, and a random $r \in \mathbb{Z}^{*}_{N}$. Let $c = (1+N) r^{N^s} \mod N^{s + 1}$. The key $k = (N, s, c)$.
    \item \textbf{Sampling a lossy function.} To sample $k \in \mathcal{K}_n^{\mathrm{lossy}}$, the algorithm randomly samples a $\lambda$-bit admissible RSA modulus $N = PQ$ and a random $r \in \mathbb{Z}^{*}_{N}$. Let $c = r^{N^s} \mod N^{s + 1}$. The key $k = (N, s, c)$.
    \item \textbf{Evaluation.} Given a key $k = (N, s, c)$ and input $s \in \{0, 1\}^n$, $F(k, s)$ outputs $c^{x} \mod N^{s + 1}$.
\end{enumerate}

Rosen and Segev~\cite{rosen2008efficient} show that breaking this construction is as hard as breaking the decisional composite residuosity assumption (DCRA) introduced by Paillier~\cite{paillier1999public}:
\begin{assumption}[Standard decisional composite residuosity assumption] \label{dcra_standard}
Let $N$ be a $\lambda$-bit RSA modulus, given $z \in \mathbb{Z}^*_{N^2}$, every algorithm running in time $O(\poly(\lambda))$ has advantage $\mathrm{negl}(\lambda)$ distinguishing the following two cases
\begin{itemize}
    \item $z = y^N (\mod N^2)$.
    \item $z$ is an element drawn from $ \mathbb{Z}^*_{N^2}$ uniformly at random.
\end{itemize}
\end{assumption}

To the best of our knowledge, the fastest classical algorithm for this task is the general number field sieve for factoring the RSA modulus, which has time complexity $2^{O(\lambda^{\frac13})}$ (see, e.g., \cite{lenstra1993development}). Based on this, we propose a stronger assumption, which we will later use to achieve larger entanglement gaps.

\begin{assumption}[Subexponential-time decisional composite residuosity assumption] \label{dcra_subexp}
There exists a constant $\beta >0$ such that given $N$ a $\lambda$-bit RSA modulus,  $z \in \mathbb{Z}^*_{N^2}$, no algorithm running in time $2^{O(\lambda^{\beta})}$ can distinguish the following two cases
\begin{itemize}
    \item $z = y^N (\mod N^2)$, for some $y \in \mathbb{Z}^*_{N^2}$.
    \item $z$ is an element drawn from $ \mathbb{Z}^*_{N^2}$ uniformly at random.
\end{itemize}
with advantage more than $2^{-\Omega(\lambda^{\beta})}$.
\end{assumption}

By the analysis of function image sizes in~\cite{rosen2008efficient} and two hardness assumptions, we have the following lemma.
\begin{lemma} \label{dcra-lossy}
~
\begin{enumerate}
    \item Under Assumption~\ref{dcra_standard}, for $\lambda = n^c$, for any $c > 0$, and $s = \lceil \frac{n}{2\lambda} \rceil$,  $f^{n, \lambda, s}_k (\cdot)$ is a collection of $(n,n - \lambda - 1)$-lossy functions. 
    \item Under Assumption~\ref{dcra_subexp}, there exists some $\lambda = \polylog n$, such that for $s = \lceil \frac{n}{2\lambda} \rceil$, $f^{n, \lambda, s}_k (\cdot)$ is a collection of $(n,n - \lambda - 1)$-lossy functions.
     \end{enumerate}
\end{lemma}

\subsection{LWE based Lossy Function}\label{sec:LWE-lossy-function}

In this subsection, we briefly review the LWE-based lossy function construction from~\cite{bouland2024public}, along with the basic definitions and assumptions related to LWE (see, e.g.~\cite{peikert2016decade}). It is important to note that there is a minor error regarding the range of one parameter in the lossy function construction in~\cite{bouland2024public}. Specifically, the parameter range of the LWE modulo $q$ was set inconsistently, with $q \ge 2^{m / \ell}$ and $q \le 2^{\ell - 1}$, where $\ell = \polylog m$. In our revised construction, we demonstrate that $q$ can be consistently set to any prime number that satisfies $2^{\polylog m} \ge q \ge \poly(m) $ throughout the construction. To resolve this inconsistency, we modify the construction in~\cite{bouland2024public} by replacing the standard rounding operation with a shifted rounding operation.

Let $U^{m \times n}_q$ denote the uniform distribution over $m \times n$ matrices with elements in $\mathbb{Z}_q$. Let $D^{m \times n}_{q, \sigma}$ be the distribution over $m \times n$ matrices with elements in $\mathbb{Z}_q$, where each element is independently sampled from a discretized Gaussian distribution with width $\sigma$. Additionally, we define the distribution $L^{m \times m}_{q, \lambda, \sigma}$ over $m \times m$ matrices with elements in $\mathbb{Z}_q$ as follows:

\begin{enumerate}
    \item Sample matrices $B$ and $C$ from $U^{\lambda \times m}_q$.
    \item Sample matrix $E$ from $D^{m \times m}_{q, \sigma}$.
    \item Output $B^\top \cdot C + E$.
\end{enumerate}

\begin{definition}[$\textbf{Decision-LWE}_{n,m,q,\sigma}$]
The $\textbf{Decision-LWE}_{n,m,q,\sigma}$ problem is the computational problem of distinguishing between the following two distributions:
\begin{enumerate}
    \item $(A, u) \leftarrow U_q^{m \times n} \times U_q^{m, 1}$,
    \item $(A, A \cdot s + e)$, where $A \leftarrow U_q^{m \times n}$, $s \leftarrow U_q^{n, 1}$, and $e \leftarrow D_{q,\sigma}^{m, 1}$.
\end{enumerate}
\end{definition}

The following two hardness assumptions of LWE are widely believed to be true~\cite{peikert2016decade, lindner2011better}.

\begin{assumption}[Standard LWE Assumption] \label{lwe_assumption}
Given a security parameter $\lambda$, $m = \poly(\lambda)$, any modulus $q = \poly(\lambda)$, and $\sigma \ge \frac{q}{\poly(m)} \ge 2 \sqrt{\lambda}$, no $\poly(\lambda)$-time quantum algorithm has a non-negligible advantage (in $m$) in solving the \textbf{Decision-LWE}$_{\lambda,m,q,\sigma}$ problem.
\end{assumption}

\begin{assumption}[Sub-Exponential LWE Assumption] \label{lwe_assumption_subexp}
Given a security parameter $\lambda$, $m = 2^{\lambda^\varepsilon}$, any modulus $q = \poly(\lambda)$, and $\sigma  \ge \frac{q}{\poly(m)} \ge 2 \sqrt{\lambda}$, no $2^{\lambda^\varepsilon}$-time quantum algorithm has a non-negligible advantage (in $m$) in solving the \textbf{Decision-LWE}$_{\lambda,m,q,\sigma}$ problem, for any $\varepsilon < 1$.
\end{assumption}

\begin{lemma}[Lemma 24 of~\cite{bouland2024public}] \label{lem:lwesecurity}
~
    \begin{itemize}
        \item Under Assumption~\ref{lwe_assumption}, matrices sampled from $U^{m \times m}_q$ and $L_{q, \lambda, \sigma}^{m \times m}$ are computationally indistinguishable for any quantum algorithm running in $\poly(\lambda)$ time, where $m = \poly(\lambda)$, $q = \poly(\lambda)$, and $\sigma  \geq \frac{q}{\poly(m)} \geq 2 \sqrt{\lambda}$.
        
        \item Under Assumption~\ref{lwe_assumption_subexp}, matrices sampled from $U^{m \times m}_q$ and $L_{q, \lambda, \sigma}^{m \times m}$ are computationally indistinguishable for any quantum algorithm running in $2^{\lambda^\varepsilon}$ time, where $m \leq 2^{\lambda^\varepsilon}$, $q = \poly(\lambda)$, and $\sigma  \geq \frac{q}{\poly(m)} \geq 2 \sqrt{\lambda}$, for any $\varepsilon < 1$.
    \end{itemize}
\end{lemma}

\newcommand{\roundp}[1]{\lfloor #1 \rfloor^s_p}

Now we define the shifted rounding operation for $\mathbb{Z}_q$ elements, where $q$ is a prime.
\begin{definition}[Shifted rounding operation] \label{def:shiftedround}
    Let $c = \lfloor q / p \rfloor$, we nearly evenly divide $\mathbb{Z}_q$ into $p$ consecutive bins such that the size difference between any two bins is at most $1$. We define $\roundp{x} \in \mathbb{Z}_p$ as the index of the bin in which $(x + \lfloor \frac{c}{2} \rfloor ) \, \mathrm{ mod }  \, q$ lies.
For a vector $x \in \mathbb{Z}_q^m$, $\roundp{\vec x} \in \mathbb{Z}_p^m$ is defined as the element-wise application of $\roundp{\cdot}$.
\end{definition}

For $m \in N$, $p = 2^4 + 1$, and $c = \lfloor q / p \rfloor$. For any $A \in \mathbb{Z}_q^{m \times m}$, define $h_A: \{0, 1\}^m \to \mathbb{Z}_p^m $ as follows: 
\[
h_A(x) = \roundp{A \vec x}.
\]

\begin{lemma} \label{lem:rounding-injective}
\[
\Pr_{A \leftarrow U_q^{m \times m}} \left [{h_A(\cdot) \text{ is injective}} \right] \geq 1 - O(2^{-m}) \,.
\]
\end{lemma}

\begin{proof}
We are going to upper bound the probability of $\roundp{A \vec x_1} = \roundp{A \vec x_2}$ for fixed $x_1, x_2 \in \mathbb{Z}^m_q$ and $x_1 \ne x_2$. Firstly, notice that for $\vec x \ne \vec 0$, $A \vec x$ is uniformly random over $\mathbb{Z}^m_q$. Therefore $\roundp{A \vec x}$ is almost uniformly random over $\mathbb{Z}^m_p$. So, if one of $\roundp{A \vec x_1}, \roundp{A \vec x_2}$ is zero, we have      
\begin{align*}
\Pr_{A \leftarrow U_q^{m \times m}}\left [ \roundp{A \vec x_1} = \roundp{A \vec 0} \right] \le (1/p + 1/q)^{-m} \le (p-1)^{-m} = 2^{-4m}.
\end{align*}
Secondly, if both $x_1, x_2$ are non-zero, $A \vec x_1$ and $A \vec x_2$ are independent uniform random vector over  $\mathbb{Z}^m_q$. Thus, $\roundp{A \vec x_1}$, and $\roundp{A \vec x_2}$ are independent  and (almost) uniform random vector over $\mathbb{Z}^m_p$. Consequently, 
\begin{align*}
\Pr_{A \leftarrow U_q^{m \times m}}\left [ \roundp{A \vec x_1} = \roundp{A \vec x_2} \right] \le (1/p + 1/q)^{-m} \le (p-1)^{-m} \le 2^{-4m}.
\end{align*}
Finally, by applying union bound to all $2^{2m}$ pairs of vectors, we obtain 
\begin{align*}
        \Pr_{A \leftarrow U_q^{m \times m}} \left [{h_A(\cdot) \text{ is injective}} \right] & \ge 1 - \sum_{x_1 \ne x_2} \Pr_{A \leftarrow U_q^{m \times m}}\left [ \roundp{A \vec x_1} = \roundp{A \vec x_2} \right] \\
        & = 1 - 2^{2m} 2^{-4m} \ge 1 - O(2^{-m}).
\end{align*}
\end{proof}

\begin{lemma} \label{lem:rounding-lossy}
\[
\Pr_{A \leftarrow L_{q, \lambda, \sigma}^{m \times m}} \left [| \img h_A(\cdot) | \le 2^{\lambda^2} \right] \geq 1 - O(2^{-\lambda})
\]
for $\sigma = 2\sqrt{\lambda}$, $2^{\lambda - 1} \ge q \ge 8pm^2\sqrt{\lambda}$.
\end{lemma}
\begin{proof}
    While this proof largely follows the proof of Lemma 38 of~\cite{bouland2024public}, it fixes a bug in that proof by replacing rounding operation $\lfloor \cdot \rfloor_p$ with shifted rounding operation $\roundp{\cdot}$. 

    Since $A \leftarrow L_{q, \lambda, \sigma}^{m \times m}$, $A = B^\top C + E$, where $B, C \leftarrow U^{\lambda \times m}_q$, and $E \leftarrow D^{m \times m}_{q, \sigma}$.
    Let $\beta = 2m\sigma$. We first prove that $E \text{ has no row $e_i$ with 1-norm $\norm{e_i}_1 \geq \beta$}$ with high probability. Since $\norm{e_i}_1 $ is the sum of $m$ independent Gaussians with standard deviation $\sigma$, we have
    \begin{align*}
     \Pr[E \text{ has at least one row with 1-norm $ \geq \beta$}] & \le \sum_{i = 1}^m \Pr \left [\norm{e_i}_1 \ge \beta \right] \\
        & = \Pr \left [\sum_{j=1}^m |e_{i, j}| \ge \beta \right] \\
        & \le \sum_{sign \in \{ 1, -1\}^m} \Pr \left [\sum_{j} sign_j \cdot e_{i, j} \ge \beta \right] \text{ (union bound)} \\
        & \le 2^m e^{-\frac{\beta^2}{2 m \sigma^2}}  \le 2^{-m}.
    \end{align*}
    Therefore, for the rest of this proof, we assume every row of $E$ has $1-$norm less than $\beta$.
    We say an element $y \in \mathbb{Z}_q$ is a border element if $\roundp{y} \ne \roundp{y + 1}$, and we define set $S \subseteq \mathbb{Z}_q$ to be the set of elements such that the distance to the nearest border element is less than or equal to $\beta$. By this definition, we know that $|S| = 2 p \beta$. Furthermore, since every row of $E$ has $1-$norm less than $beta$, the error matrix can only contribute at most $\beta$ to $h_A(x)$. Therefore, the $i$-th coordinate of $B^\top C \vec x$  must equal the $i$-th coordinate of $h_A(x)$ if the $i$-th coordinate of $B^\top C \vec x$ is not in $S$. For $z \in \mathbb{Z}^{\lambda}_q$, let $C_z = \{ x \in \{ 0, 1\}^m | C \vec x = \vec z \}$. Let $r(B^\top \cdot \vec z)$ be the number of coordinates of vector $B^\top \cdot \vec z$ that is in the set $S$, then $|\{ h_A(x) | x \in C_z\}| \le 2^{r(B^\top \cdot \vec z)}$. 

    Therefore, 
    \[
    | \img h_A(\cdot) | \le \sum_{z \in \mathbb{Z}^{\lambda}_q} 2^{r(B^\top \cdot \vec z)}.
    \]
    Next, for $\vec z \ne \vec 0$ we calculate the expected value of $2^{r(B^\top \cdot \vec z)}$:
    \begin{align*}
        \E_{B \sim U^{\lambda \times m}_q} \left [ 2^{r(B^\top \cdot \vec z)} \right] & = \E_{B \sim U^{\lambda \times m}_q} \left [ \prod_{i = 1}^m (1 + \mathds{1} [b^\top_i \cdot \vec z \in S]) \right] \\
        & = \prod_{i=1}^m \E_{b \sim U^{\lambda \times 1}_q} \left[ 1 +  \mathds{1} [b^\top \cdot \vec z \in S] \right] \\
        & =  \prod_{i=1}^m \left ( 1 + \Pr_{b \sim U^{\lambda \times 1}_q }[b^\top \cdot \vec z \in S] \right) \\
        & = \prod_{i=1}^m \left ( 1 + \frac{2 p \beta}{q} \right) \\
        & = \left ( 1 + \frac{2 p \beta}{q} \right)^m \le e, 
    \end{align*}
    where $b_i$ is the $i$-th column of $B$. The last two lines hold because $b^\top \cdot z$ is uniform over $\mathbb{Z}_q$ for any fixed non-zero $\vec z$, and $q \ge 2 p \beta m$.

    As for the case $\vec z = C \vec x = 0$, note that $r(B^\top \vec z) = r(\vec 0)$ for any $B$. In addition, $0$ is right in the middle of a bin by design, so $0 \not \in S$. Therefore, $\E_B \left [r(B^\top \vec z) \right] = \E_B \left [r(\vec 0) \right]= 0$.

    To conclude the proof, by Markov's inequality,
    \begin{align*}
        \Pr_{A \leftarrow L_{q, \lambda, \sigma}^{m \times m}} \left [| \img h_A(\cdot) | \le 2^{\lambda^2} \right] & > \frac{1}{2^{\lambda^2}} \sum_{z \in \mathbb{Z}^{\lambda}_q} \E_{B \sim U^{\lambda \times m}_q} \left [ 2^{r(B^\top \cdot \vec z)} \right]\\
        & \le \frac{e q^\lambda}{2^{\lambda^2}} \\
        & \le \frac{e 2^{\lambda (\lambda - 1)} }{2^{\lambda^2}} \\
        & = O(2^{-\lambda}),
    \end{align*}
    as desired.
\end{proof}

Finally, we define our LWE-based lossy functions. Note that the properties that the LWE-based lossy functions differ slightly from Definition~\ref{def:lossyfunction}, as the properties of the injective branch and the lossy branch do not always hold but rather hold with high probability.

\begin{definition}[LWE-based lossy functions]
For input length $n$, and parameter $\lambda$. Let $ p = 2^4 + 1, \sigma = 2\sqrt{\lambda}$, and $q$ be a prime number in $[8pn^2 \sqrt{\lambda}, 16pn^2\sqrt{\lambda}]$\footnote{By Bertrand's postulate (see, e.g.~\cite{sondow2009ramanujan}), there is always at least one prime number in $[n, 2n]$ for any $n \ge 1$.}. The collection of lossy functions $f^{n, p, q, \lambda, \sigma}_k(\cdot): \{ 0, 1\}^n \to \{ \mathbb{Z}^n_p\}$ is defined as follows:
    \begin{enumerate}
        \item \textbf{Sampling an injective function.} To sample $k \in \mathcal{K}_n^{\mathrm{inj}}$, the algorithm randomly samples matrix $A \leftarrow U^{n \times n}_{q}$, and let $k = A$.
        \item \textbf{Sampling a lossy function.} To sample $k \in \mathcal{K}_n^{\mathrm{lossy}}$, the algorithm randomly samples matrix $A \leftarrow L^{n \times n}_{q, \lambda, \sigma}$, and let $k = A$.
        \item \textbf{Evaluation.} Given a key $k = A \in \mathbb{Z}^{n \times n}_p$ and input $x \in \{ 0, 1\}^n$, the function $f_k(\cdot): \{ 0, 1\}^n \to \mathbb{Z}^{n}_p$ outputs $\roundp{A \vec x}$.
    \end{enumerate}
\end{definition}

Combining Lemma~\ref{lem:rounding-injective}, Lemma~\ref{lem:rounding-lossy}, and two LWE hardness assumptions, we have the following theorem.

\begin{theorem}\label{thm:lwe-lossy}
For input length $n$, and parameter $\lambda$. Let $ p = 2^4 + 1, \sigma = 2\sqrt{\lambda}$, and $q$ be a prime number in $[8pn^2 \sqrt{\lambda}, 16pn^2\sqrt{\lambda}]$. The following properties hold for lossy functions $f^{n, p, q, \lambda, \sigma}_k$:
    \begin{enumerate}
        \item Let $k \leftarrow \mathcal{K}_n^{\mathrm{inj}}$, $\Pr[f_k(\cdot) \text { is injective}] \ge 1 - O(2^{-n})$.
        \item Let $k \leftarrow \mathcal{K}_n^{\mathrm{lossy}}$, $ \Pr[ \left | \img f_k(\cdot) \right | \le 2^{\lambda^2}] \ge 1 - O(2^{-\lambda})$.
        \item Under Assumption~\ref{lwe_assumption}, for any $c > 0$, and $\lambda = \Omega(n^c)$, then no quantum algorithm runs in time $\poly(n)$ can distinguish $\mathcal{K}_n^{\mathrm{inj}}$ and $\mathcal{K}_n^{\mathrm{lossy}}$ with non-negligible probability in $n$.

        Under Assumption~\ref{lwe_assumption_subexp}, there exists some $\lambda = \polylog (n)$ such that then no quantum algorithm runs in time $\poly(n)$ can distinguish $\mathcal{K}_n^{\mathrm{inj}}$ and $\mathcal{K}_n^{\mathrm{lossy}}$ with non-negligible probability in $n$.
    \end{enumerate}
\end{theorem}

\section{Public-key pseudoentangled states}\label{sec:public-key-pseudoentangled-states}
In this section, we present two constructions of public-key pseudoentangled states: one based on DCRA lossy functions (Section~\ref{sec:DCRA-lossy-function}) and the other on LWE lossy functions (Section~\ref{sec:LWE-lossy-function}). 

\begin{definition}[Public-key pseudoentanglement, Definition 3.9 \& 3.16 of~\cite{bouland2024public}] \label{def:pk_entanglement}
A public-key pseudoentangled state ensemble with entanglement gap $(O(f(n)), \Omega(g(n)))$ across geometrically local cuts on a 1D line consists of two families of quantum states $\Psi^\lo_n = \{ \ket{\psi_{k}}\}_{k \in \mathcal{K}^\lo_n}$ and $\Psi^\hi_n = \{ \ket{\psi_{k}}\}_{k \in \mathcal{K}^\hi_n}$ indexed by key sets $\mathcal{K}^\lo_n$ and $\mathcal{K}^\hi_n$ respectively that satisfy
\begin{enumerate}
\item \label{item:def:n_qubits}
Every $\ket{\psi_{k}} \in \Psi^\lo_n \cup \Psi^\hi_n$ is an $n$-qubit state.
\item \label{item:def:key_sampling}
Every key $k \in \mathcal{K}^\lo_n \cup \mathcal{K}^\hi_n$ has length $\poly(n)$, and there exists an efficient sampling procedure that, given as input $n$ and a label ``high'' or ``low'', outputs a key $k \in \mathcal{K}^\lo_n$ or $k \in \mathcal{K}^\hi_n$, respectively.
We write $k \leftarrow \mathcal{K}^\lo_n$ and $k \leftarrow \mathcal{K}^\hi_n$ for keys sampled according to this procedure.
\item \label{item:def:eff_prep}
Given $k \in \mathcal{K}^\lo_n \cup \mathcal{K}^\hi_n$, the corresponding state $\ket{\psi_{k}}$ is efficiently preparable (without knowing whether $k \in \mathcal{K}^\lo$ or $k \in \mathcal{K}^\hi$).
Formally, there exists a uniform polynomial-time circuit family $\{C_n\}$ such that $C_n$ takes as input a key $k \in \mathcal{K}^\lo_n \cup \mathcal{K}^\hi_n$ and outputs a state negligibly close to $\ket{\psi_{k}}$.
\item \label{item:def:com_indist}
The keys from $\mathcal{K}^\lo_n$ and $\mathcal{K}^\hi_n$ are computationally indistinguishable.
Formally, for all $\poly(n)$-time quantum adversaries $\mathcal{A}$ that take as input a key $\mathcal{K}^\lo_n \cup \mathcal{K}^\hi_n$ and output a single bit: 
\begin{align*}
\Big| \Pr_{k \leftarrow \mathcal{K}_n^\lo}{\mathcal{A}(k) = 0} - \Pr_{k \leftarrow \mathcal{K}_n^\hi}{\mathcal{A}(k) = 0} \Big| = \mathrm{negl}(n) \,.
\end{align*}
\item \label{item:def:entanglement_gap_multi} For any function $b(n) = \omega(\log n)$, with overwhelming probability, states in $\Psi^\lo_n$ have entanglement entropy $O(f(n))$ and states in $\Psi^\hi_n$ have entanglement entropy $\Omega(g(\text{distance from end of line}))$ for all geometrically local cuts that are at least $b(n)$ far from the end of the line.
Formally,
\begin{align*}
\Pr_{k \leftarrow \mathcal{K}_n^\lo}\big[\forall c \in \{b(n), \ldots, n - b(n)\}: \; S((\psi_k)_{[c]}) \leq O(f(n))\big] &\geq 1 - \mathrm{negl}(n)\,,\\
\Pr_{k \leftarrow \mathcal{K}_n^\hi}[\forall c \in \{b(n), \dots, n - b(n)\}: \; S((\psi_k)_{[c]}) \geq \Omega\big(\min(g(c), g(n-c))\big)] &\geq 1 - \mathrm{negl}(n) \,.
\end{align*}
Here, $(\psi_k)_{[c]}$ is the reduced states of $\ket{\psi_k}$ on qubits $(1, \dots, c)$.
\end{enumerate}
\end{definition}
Both of our constructions are of the following form
\begin{align}\label{eqn:phase-state}
\ket{\psi}=\frac{1}{2^{n/2}}\sum_{x\in\{0,1\}^n}(-1)^{f\circ h(x)}\ket{x},
\end{align}
where we construct two families of keys $\mathcal{K}^\lo_n$ and $\mathcal{K}^\hi_n$ that makes the state either area-law entangled or volume-law entangled. Throughout this section, we assume $n$ is a power of 2.

\subsection{DCRA based public-key pseudoentangled states}\label{sec:DCRA-pseudoentangled-state}

In our DCRA based public-key pseudoentangled states construction, $f\colon\{0,1\}^{2n - \lambda}\to\{0,1\}$ is a four-wise independent function, and the function $h$ is a composite function with the following tree structure in Figure~\ref{fig:hashing-tree-structure}. For simplicity and clarity of exposition, we will omit the subscripts denoting the function keys where the context permits. Here, the parameter \(\lambda = \polylog n\) is the security parameter and we assume it is a power of $2$. \(h_{[a,b]}\) denotes a function acting on the tree node corresponding to the \(a\)-th bit through the \(b\)-th bit with suitable input and output sizes. Similarly, for \(x \in \{0,1\}^n\), we use \(x_{[a,b]}\) to denote the \((b-a+1)\)-bit string formed by the \(a\)-th bit through the \(b\)-th bit.

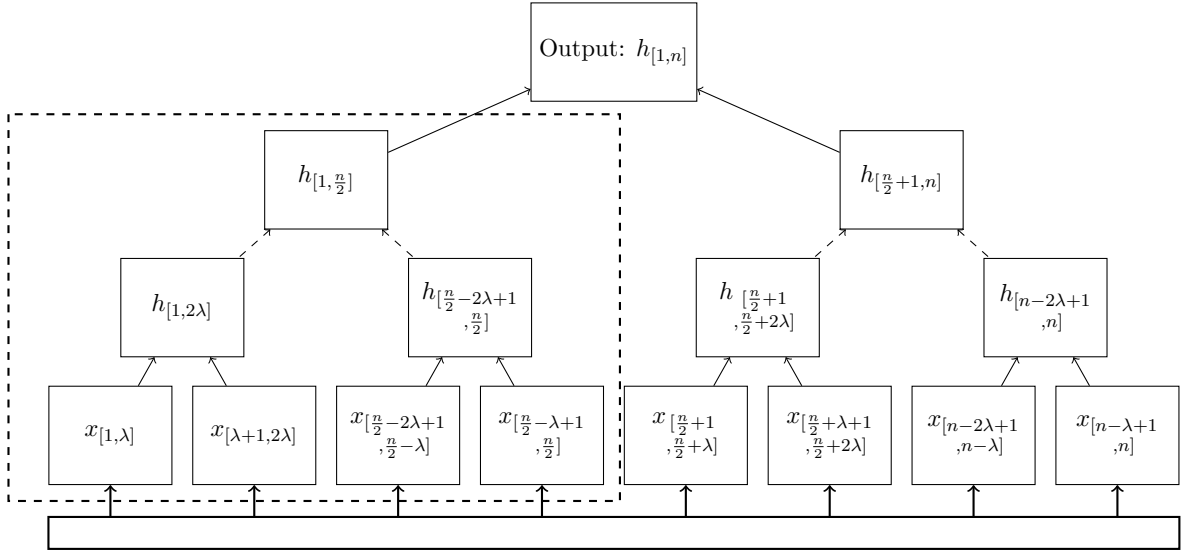
\begin{figure}[htbp]
    \centering
    \tikzset{global scale/.style={
    scale=#1,
    every node/.append style={scale=#1}
  }
}
    \begin{tikzpicture}[level distance=2cm,
  level 1/.style={sibling distance=9cm},
  level 2/.style={sibling distance=4.5cm},
  level 3/.style={sibling distance=2.25cm},
  nodes={draw, rectangle},minimum width=5em, minimum height=4em ,global scale=0.85]
  
  \node {Output: $h_{[1, n]}$}
    child {node(10) {$h_{[1,\frac{n}{2}]}$}
      child {node(1) {$h_{[1,2\lambda]}$}
        child[solid]{ node(3) {$x_{[1,\lambda]}$}
        edge from parent[<-] node[left,xshift=-2.5,draw=none]{}}
        child[solid]{ node(b2) {$x_{[\lambda+1,2\lambda]}$}
        edge from parent[<-] node[right,xshift=2.5,draw=none]{}}
      edge from parent[<-, dashed] node[left,xshift=-2.5,draw=none]{}
      }
      child {node {$h_{\substack{[\frac{n}{2}-2\lambda+1\\,\frac{n}{2}]}}$} 
        child[solid]{ node(b3) {$x_{\substack{[\frac{n}{2}-2\lambda+1 \\ ,\frac{n}{2}-\lambda]}}$}
        edge from parent[<-] node[left,xshift=-2.5,draw=none]{}}
        child[solid]{ node(b4) {$x_{\substack{[\frac{n}{2}-\lambda+1 \\ ,\frac{n}{2}]}}$}
        edge from parent[<-] node[right,xshift=2.5,draw=none]{}}
      edge from parent[<-, dashed] node[right,xshift=2.5,draw=none]{}}
      edge from parent[<-] node[left,xshift=-2.5,draw=none]{}
    }
    child {node {$h_{[\frac{n}{2}+1,n]}$}
    child {node {$h_{\substack{[\frac{n}{2}+1\\ ,\frac{n}{2}+2\lambda]}}$} 
        child[solid]{ node(b5) {$x_{\substack{[\frac{n}{2}+1\\ ,\frac{n}{2}+\lambda]}}$}
        edge from parent[<-] node[left,xshift=-2.5,draw=none]{}}
        child[solid]{ node(b6) {$x_{\substack{[\frac{n}{2}+\lambda+1\\ ,\frac{n}{2}+2\lambda]}}$}
        edge from parent[<-] node[right,xshift=2.5,draw=none]{}}
    edge from parent[<-, dashed] node[left,xshift=-2.5,draw=none]{}}
      child {node(2) {$h_{\substack{[n-2\lambda+1\\ ,n]}}$} 
        child[solid]{ node(b7) {$x_{\substack{[n-2\lambda+1\\ ,n-\lambda]}}$}
        edge from parent[<-] node[left,xshift=-2.5,draw=none]{}}
        child[solid]{ node(4) {$x_{\substack{[n-\lambda+1\\ ,n]}}$}
        edge from parent[<-] node[right,xshift=2.5,draw=none]{}}
      edge from parent[<-, dashed]}
      edge from parent[<-]
    };
    \draw[thick] ($(3.south west) + (0,-1)$) rectangle ($(4.south east) + (0,-0.5)$);
    \draw[<-, thick] (3.south) -- ($(3.south) + (0,-0.5)$);
    \draw[<-, thick] (b2.south) -- ($(b2.south) + (0,-0.5)$);
    \draw[<-, thick] (b3.south) -- ($(b3.south) + (0,-0.5)$);
    \draw[<-, thick] (b4.south) -- ($(b4.south) + (0,-0.5)$);
    \draw[<-, thick] (b5.south) -- ($(b5.south) + (0,-0.5)$);
    \draw[<-, thick] (b6.south) -- ($(b6.south) + (0,-0.5)$);
    \draw[<-, thick] (b7.south) -- ($(b7.south) + (0,-0.5)$);
    \draw[<-, thick] (4.south) -- ($(4.south) + (0,-0.5)$);
     \draw[dashed, thick] ($(10.north west) + (-4, 0.25)$) rectangle ($(b4.south east) + (0.25, -0.25)$);

\end{tikzpicture}
\caption{The tree structure of $h(x)$}
    \label{fig:hashing-tree-structure}
\end{figure}
For $m = \lambda, 3\lambda, 7\lambda, \ldots, (2n / \lambda)$, instantiate a collection of $(m, m - \lambda  - 1)$ DCRA-based lossy function family. For input length $m$, let $\mathcal{K}^{\rm lossy}_m$, and $\mathcal{K}^{\rm inj}_m$  denote the set of keys of the lossy mode and the set of keys of the injective mode, respectively. We denote $g^m_k(\cdot): \{0, 1\}^m \to \{0, 1\}^{m + \lambda}$ be the evaluation function of the lossy function with input size $m$ and key $k$. Note that the size of each key in either  $\mathcal{K}^{\rm inj}_m$ or $\mathcal{K}^{\rm lossy}_m$ is $O(m)$.  Let $F_{2n - \lambda} = \{F_k: \{0,1\}^{2n - \lambda} \to \{0,1\}\}_{k \in \mathcal{K}_{2n - \lambda}^4}$ be a 4-wise independent family as given in Lemma~\ref{lem:r-wise_exist}. We first describe the sampling procedure for the keys $\mathcal{K}^\lo_n$ and $\mathcal{K}^\hi_n$.
\begin{itemize}

\item To sample $k \leftarrow \mathcal{K}^\lo_n$,  first sample $k^{\rm fin} \leftarrow \mathcal{K}_{2n - \lambda}^4$, then for $m = \lambda, 3\lambda, 7\lambda, \ldots, (2n / \lambda)$, independently sample $k^{\rm hash}_m  \leftarrow \mathcal{K}^{\rm lossy}_m$.
Set $k = ( k^{\rm hash}_{\lambda},  k^{\rm hash}_{3\lambda}, k^{\rm hash}_{7\lambda}, \ldots, k^{\rm hash}_{2n / \lambda}, k^{\rm fin})$.

\item To sample $k \leftarrow \mathcal{K}^\hi_n$, first sample $k^{\rm fin} \leftarrow \mathcal{K}_{2n - \lambda}^4$,  then for $m = \lambda, 3\lambda, 7\lambda, \ldots, (2n / \lambda)$, independently sample $k^{\rm hash}_m  \leftarrow \mathcal{K}^{\rm inj}_m$.
Set $k = ( k^{\rm hash}_{\lambda},  k^{\rm hash}_{3\lambda}, k^{\rm hash}_{7\lambda}, \ldots, k^{\rm hash}_{2n / \lambda}, k^{\rm fin})$.
\end{itemize}

\begin{fact}\label{lem:DCRA-key-size}
The number of bits in the keys of either $ \mathcal{K}^\hi_n$ or $ \mathcal{K}^\hi_n$ is $O(n)$.
\end{fact}
\begin{proof}
The total number of key bits is given by
\[
\sum_{i = 1}^m 2^{2^i - 1} \cdot \lambda = \lambda (1 + 3 + 7 + \ldots + (2 n / \lambda)) = O(n).
\]
\end{proof}

Given key $k = ( k^{\rm hash}_{\lambda},  k^{\rm hash}_{3\lambda}, k^{\rm hash}_{7\lambda}, \ldots, k^{\rm hash}_{2n / \lambda}, k^{\rm fin})$, the function $h_{k, [a, b]}$ is defined as follows: 
\begin{itemize}
    \item If $b - a + 1 = \lambda$, $h_{k, [a, b]}$ outputs $x_{[a, b]}$.
    \item Otherwise, the input of $h_{k, [a, b]}$ is the concatenation of the outputs of its two child nodes. Suppose the input length is $m$,  $h_{k, [a, b]} := g^m_k$. Note that For any two functions $h_{k, [a_1,b_1]}$ and $h_{k, [a_2,b_2]}$ on the same layer of the tree, they are the same function. By construction of the lossy function, in either case, $h_{[a, b]}$ outputs a bit string that has $\lambda$ more bits than the input. 
\end{itemize}

Formally, we denote \( h^s_{k, [a,b]} \) as the function obtained by composing all the functions in the subtree rooted at \( h_{k, [a,b]} \):

\begin{align*}
 h^s_{k, [a,b]}(x)  = 
\begin{cases}
x_{[a, b]}, & \text{if } b - a + 1 = \lambda, \\
h_{k, [a , b]}  \left( h^s_{k, [a, \frac{a + b}{2}]}(x) \parallel h^s_{k, [\frac{a + b}{2} + 1, b]}(x) \right), & \text{otherwise.}
\end{cases}
\end{align*}

For instance,  \( h^s_{[k, 1,\frac{n}{2}]} \) corresponds to the subtree of functions contained in the dashed rectangle of Figure~\ref{fig:hashing-tree-structure}. In this notation, $h_s := h^s_{k, [1,n]}$. Furthremore, by a simple induction, it is easy to see that the output size of $h_{k, [a, b]}$ is $(2^{\ell + 1} - 1) \lambda$   given $b - a + 1 = 2^\ell \lambda$ for some integer $\ell \ge 0$. Consequently, as $n = 2^{\log (n / \lambda)} \lambda$, the output size of $h$ is exactly $(2^{\log (n / \lambda) + 1} - 1) \lambda = 2n - \lambda $.

With this notation, for $k = ( k^{\rm hash}_{\lambda},  k^{\rm hash}_{3\lambda}, k^{\rm hash}_{7\lambda}, \ldots, k^{\rm hash}_{2n / \lambda}, k^{\rm fin})$, we next define the function $s_k\colon\{0, 1\}^n \to \bits$ as 
\[
s_k(x) = f_{k^{\rm fin}} \left ( h_k(x) \right ) \,.
\]
The states $\ket{\psi_k}$ are then given by
\begin{align} \label{eqn:dcraconstruction}
\ket{\psi_k} = \sum_{x \in \{0, 1\}^n } (-1)^{s_k(x)} \ket{x} \,. 
\end{align}

\subsubsection{Low-entanglement states}\label{sec:public-key-state-low}
We construct the low entanglement states by choosing an $h$ such that each component function $h_{[a,b]}$ in Figure~\ref{fig:hashing-tree-structure} has small image size, and $f$ is sampled uniformly randomly from a $4$-wise independent function family. Throughout this part, we assume each component function $h_{[a,b]}$ in Figure~\ref{fig:hashing-tree-structure} has image size at most $2^{\lambda + 1}=2^{\polylog n}$. We use $\ket{\psi_\lo}$ to denote the resulting phase state in  Eq.~\ref{eqn:phase-state}.

\begin{fact}
The image size of $h^s_{[a,b]}$ is at most $2^{\lambda + 1}=2^{\polylog n}$.
\end{fact}

\begin{lemma}\label{lem:subtree-rank}
For each $h_{[a,b]}$ and any function $g\colon\{0,1\}^n\to \{0,1\}$, consider the $T$ matrix of the function $g\circ h^s_{[a,b]}(x)$ with respect to the cut $I=[a,b]$, we have $\rank(T)\leq 2^{\lambda + 1}$.
\end{lemma}
\begin{proof}
Observe that the effect of $h^s_{[a,b]}$ is hashing the rows of $T$ such that there are at most $2^{\lambda}$ different rows. Quantitatively, given that the image size of $h^s_{[a,b]}$ is at most $2^{\lambda + 1}=2^{\polylog n}$, there exist at most $\mathcal{K}\leq 2^\lambda$ sets $S_1,\ldots,S_{\mathcal{K}}\subseteq 2^{[a,b]}$ such that, for any $S_{\mathcal{K}}$, any $x_1,x_2\in S_{\mathcal{K}}$, and any $y\in 2^{[n]\setminus [a,b]}$, we have $h^s_{[a,b]}(x_1\|_{[a,b]}y)=h^s_{[a,b]}(x_2\|_{[a,b]}y)$. Hence, the row of $T$ corresponding to $x_1$ and the row corresponding to $x_2$ are the same for any $g$. Hence, the rank of $T$ is at most $2^{\lambda + 1}$.
\end{proof}

\begin{corollary}\label{cor:subtree-rank}
For each $h_{[a,b]}$, consider the $T$ matrix of the function $f\circ h(x)$ with respect to the cut $I=[a,b]$, we have $\rank(T)\leq 2^{\lambda + 1}$.
\end{corollary}
\begin{proof}
We denote \( h_{\overline{[a,b]}} \) as the function obtained by composing all the functions outside the subtree of \( h_{[a,b]} \). Thus, we have \( f \circ h = f \circ h_{\overline{[a,b]}} \circ h^s_{[a,b]} \). By applying Lemma~\ref{lem:subtree-rank} with \( g = f \circ h_{\overline{[a,b]}} \), we conclude the proof.
\end{proof}

\begin{lemma}\label{lem:DCRA-state-low}
For any $1\leq a<b\leq n$ and any $f\colon\{0,1\}^n\to\{0,1\}$, if each component function $h_{[a,b]}$ in Figure~\ref{fig:hashing-tree-structure} has image size at most $2^{\lambda + 1}=2^{\polylog n}$, the phase state $\ket{\psi_\lo}$ in Eq.~\ref{eqn:phase-state} satisfies $S(\psi_{[a,b]}) \leq O(\lambda\log n)$.
\end{lemma}
\begin{proof}
Given the tree structure of $h$ in Figure~\ref{fig:hashing-tree-structure}, for any $1\leq a<b\leq n$, there exists $\kappa \le 2\log n$ disjoint sets $[a_1,b_1],\ldots,[a_\kappa,b_\kappa]$ such that
\begin{align*}
[a_1,b_1]\cup\cdots\cup[a_\kappa,b_\kappa]=[a,b],
\end{align*}
where each $[a_i,b_i]$ either corresponds to a node $h_{[a,b]}$ in Figure~\ref{fig:hashing-tree-structure}, or has size $|b_i-a_i|\leq\lambda$. Then by Corollary~\ref{cor:subtree-rank}, we have $S(\psi_{[a_i,b_i]}) \leq \lambda$ for any $i\in[\kappa]$. Then by the subadditivity of entanglement entropy, we have
\begin{align*}
S(\psi_{[a,b]})\leq\sum_{i=1}^\kappa S(\psi_{[a_i,b_i]})=O(\lambda\log n).
\end{align*}
\end{proof}

\subsubsection{High-entanglement states}\label{sec:public-key-state-high}
We construct the high entanglement states by choosing an $h$ such that each component function $h_{[a,b]}$ in Figure~\ref{fig:hashing-tree-structure} is a one-to-one function, and $f$ is sampled uniformly randomly from a $4$-wise independent function family. We use $\ket{\psi_\hi}$ to denote the resulting phase state in Eq.~\ref{eqn:phase-state}.
\begin{lemma}[Lemma B.6 of \cite{aaronson2024quantum}]\label{lem:4wisefrobenius}
	Let $f$ be a function uniformly sampled from a $4$-wise independent function family $F= \{f :[2^{c}] \times [2^{n-c}] \to \{1, -1 \} \}$, the Frobenius norm $\| \frac{1}{2^{n}} AA^{\mathsf{T}} \|_F \le 2^{-n/8}$ with high probability, where $A_{ij}\coloneqq f(i, j)$, given that $c\in[8\log^2 n,n-8\log^2 n]$.
\end{lemma}

\begin{lemma}\label{lem:highentropy4wise}
    The function $f\circ h(x)$ is 4-wise independent.
\end{lemma}

\begin{lemma}\label{lem:DCRA-state-high}
For any $1\leq a<b\leq n$ and any $f\colon\{0,1\}^n\to\{0,1\}$, if each component function $h_{[a,b]}$ in Figure~\ref{fig:hashing-tree-structure} each component function $h_{[a,b]}$ in Figure~\ref{fig:hashing-tree-structure} is a one-to-one function, the phase state $\ket{\psi_{\hi}}$ in Eq.~\ref{eqn:phase-state} satisfies $S(\psi_{[a,b]}) \geq \Omega(b-a)$.
\end{lemma}
\begin{proof}
The proof proceeds by combining Lemma~\ref{lem:entanglement_from_matrix}, Lemma~\ref{lem:4wisefrobenius}, and Lemma~\ref{lem:highentropy4wise}.
\end{proof}

\subsubsection{Indistinguishability argument}
\begin{proposition}\label{prop:pk-state-indistinguishability} 
~

\begin{itemize}
\item Under the standard DCRA assumption (\ref{dcra_standard}), for any $f(n) = n^\delta$ for $\delta > 0$, the state families $\Psi^\lo_n = \{ \ket{\psi_{k}}\}_{k \in \mathcal{K}^\lo_n}$ and $\Psi^\hi_n = \{ \ket{\psi_{k}}\}_{k \in \mathcal{K}^\hi_n}$ from Eq.~\ref{eqn:dcraconstruction} form a pseudoentangled state ensemble with entanglement gap $(O(f(n)), \Omega(n))$.

\item Under the subexponential-time DCRA assumption (\ref{dcra_subexp}), there exists a function $f(n) = \polylog n$ such that the state families $\Psi^\lo_n = \{ \ket{\psi_{k}}\}_{k \in \mathcal{K}^\lo_n}$ and $\Psi^\hi_n = \{ \ket{\psi_{k}}\}_{k \in \mathcal{K}^\hi_n}$ from Eq.~\ref{eqn:dcraconstruction} form a pseudoentangled state ensemble with entanglement gap $(O(f(n)), \Omega(n))$.
\end{itemize}
\end{proposition}
\begin{proof}
The proof proceeds by combining Lemma~\ref{lem:DCRA-state-low} and Lemma~\ref{lem:DCRA-state-high}.
\end{proof}

\subsection{LWE based public-key pseudoentangled states}\label{sec:LWE-pseudoentangled-state}

Our LWE-based public-key pseudoentangled state construction closely resembles that of~\cite{bouland2024public}. The key distinction lies in the application of lossy functions: \cite{bouland2024public} samples a distinct lossy function key for each application, while we reuse the same key to reduce the total key size. Instead of requiring lossy functions with varying domain sizes as in~\cite{bouland2024public}, we use a single lossy function with input length $n$. For any input $x$ of length less than or equal to $n$, we pad it with leading zeros to create an $n$-bit string, thus maintaining compatibility across different input sizes.

The formal construction of our LWE based public-key pseudoentangled states is given as follows.

\begin{definition}[LWE based public-key pseudoentangled states] \label{def:lweconstruction}
Fix a function $ \ell = f(n)$ which will be treated as a parameter of the construction, instantiate the $n$-bit lossy function from Section~\ref{sec:LWE-lossy-function} with parameters $\lambda = \polylog n, p = 2^4 + 1, \sigma = 2\sqrt{\lambda}$, and $q$ be a prime number in $[8pn^2\sqrt{\lambda}, 16pn^2\sqrt{\lambda}]$. Let $H_n = \{h_k: \{0,1\}^n \to \{0,1\}\}_{k \in \mathcal{K}_n^4}$ be a 4-wise independent family as given in Lemma~\ref{lem:r-wise_exist}. For $m \in \{\ell, \ell+1, \ldots, n\}$ let $H'_m = \{h'_k: \mathbb{Z}_p^n \to \{0,1\}^m\}_{k \in \mathcal{K}_n^2}$ be a pairwise independent family as given in Lemma~\ref{lem:r-wise_exist}.
We first describe the sampling procedure for the keys $\mathcal{K}^\lo_n$ and $\mathcal{K}^\hi_n$.
\begin{enumerate}
\item To sample $k \leftarrow \mathcal{K}^\lo_n$, first sample $A \leftarrow L_{q, \ell, \sigma}^{n \times n}$, and sample $k^{\mathrm{fin}} \leftarrow \mathcal{K}^4_n$,  then for $m \in \{\ell, \ell+1, \ldots, n\}$, independently sample $k^{\rm hash}_m \leftarrow \mathcal{K'}^{2}_m$.
Set $k = (A, k^{\rm hash}_{\ell},  k^{\rm hash}_{\ell + 1}, \ldots, k^{\rm hash}_n, k^{\rm fin})$.
\item To sample $k \leftarrow \mathcal{K}^\hi_n$, first sample $A \leftarrow U_q^{n \times n}$, and sample $k^{\mathrm{fin}} \leftarrow \mathcal{K}^4_n$,  then for $m \in \{\ell, \ell+1, \ldots, n\}$, independently sample $k^{\rm hash}_m \leftarrow \mathcal{K'}^{2}_m$.
Set $k = (A, k^{\rm hash}_{\ell},  k^{\rm hash}_{\ell + 1}, \ldots, k^{\rm hash}_n, k^{\rm fin})$.
\end{enumerate}
For convenience, for any $A \in \mathbb{Z}_q^{n \times n}$ and $k^{\rm hash}_{m} \in H'_m$, the function $f^A_{k^{\rm hash}_{m}}: \{ 0, 1\}^m \to \{0, 1\}^m$ is defined as     
\begin{align*}
    f^A_{k^{\rm hash}_{m}}(\vec x) = h'_{k_m^{\rm hash}} \left( \roundp{A \cdot \mathrm{pad}_{n - m}(\vec x)} \right) \,,
    \end{align*}
where the function $\mathrm{pad}_{n - m}(\vec x)$ pads $n - m$ leading zeros to $\vec x$.

For  $k = (A, k^{\rm hash}_{\ell},  k^{\rm hash}_{\ell + 1}, \ldots, k^{\rm hash}_n, k^{\rm fin})$, define the \emph{labelling functions} $r^{\ell}_k,  \dots, r^n_k: \{0, 1\}^n \to \{0, 1\}^n$ recursively by 
\[
 r^m_k(x) = \begin{cases}
      x, \quad \quad  & m = n+1, \\
      r_k^{m+1}( f^A_{k^{\rm hash}_{m}}(i)\parallel j), \quad \quad & f(n) \le m \le n,\ \mathrm{MSB}_m(x) = i,\ \mathrm{LSB}_{n - m}(x) = j,
 \end{cases}
\]
where $\mathrm{MSB}_m(x)$ is the first $m$ bits of $x$, $\mathrm{LSB}_m(x)$ is the last $m$ bits of $x$, and $i \parallel j$ is the concatenation of bit strings $i$ and $j$. For simplicity, we define $r_k(x) = r_k^{\ell}(x)$.

With this notation, for $k = (A, k^{\rm hash}_{\ell},  k^{\rm hash}_{\ell + 1}, \ldots, k^{\rm hash}_n, k^{\rm fin})$, we next define the function $s_k\colon\{0, 1\}^n \to \bits$ as 
\[
s_k(x) = h_{k^{\rm fin}} \left ( r_k(x) \right ) \,.
\]
The states $\ket{\psi_k}$ are then given by
\begin{align*}
\ket{\psi_k} = \sum_{x \in \{0, 1\}^n } (-1)^{s_k(x)} \ket{x} \,.
\end{align*}

\end{definition}

\begin{fact}\label{fact:LWE-space}
    For any $k \in \mathcal{K}_n^{\mathrm{low}} \cup \mathcal{K}_n^{\mathrm{high}}$, the space complexity of computing phase functions $s_k$ is $\Theta(n^2 \log n)$ and the time complexity is $\poly(n)$.
\end{fact}

With essentially the same proof as in~\cite{bouland2024public}, we have the following proposition.
\begin{proposition}\label{prop:lwe-pseudoentangled}
    ~
\begin{itemize}
\item Under the standard LWE assumption (\ref{lwe_assumption}), for any function $f(n) = n^\delta$ for $\delta > 0$, the state families $\Psi^\lo_n = \{ \ket{\psi_{k}}\}_{k \in \mathcal{K}^\lo_n}$ and $\Psi^\hi_n = \{ \ket{\psi_{k}}\}_{k \in \mathcal{K}^\hi_n}$ from Definition~\ref{def:lweconstruction} form a pseudoentangled state ensemble with entanglement gap $(O(f(n)), \Omega(n))$.

\item Under the subexponential-time LWE assumption (\ref{lwe_assumption_subexp}), there exists a function $f(n) = \polylog n$ such that the state families $\Psi^\lo_n = \{ \ket{\psi_{k}}\}_{k \in \mathcal{K}^\lo_n}$ and $\Psi^\hi_n = \{ \ket{\psi_{k}}\}_{k \in \mathcal{K}^\hi_n}$ from Definition~\ref{def:lweconstruction} form a pseudoentangled state ensemble with entanglement gap $(O(f(n)), \Omega(n))$.

\end{itemize}
\end{proposition}

\section{Quantum Turing machine for public-key pseudoentangled states}\label{sec:TM-public-key-pseudoentanglement}
In this section, we present the construction of two quantum Turing machines (QTM) that generate the public-key pseudoentangled states based on DCRA (Section~\ref{sec:DCRA-pseudoentangled-state}) and LWE (Section~\ref{sec:LWE-pseudoentangled-state}), respectively in Section~\ref{sec:DCRA-QTM} and Section~\ref{sec:lwetm}. 

For simplicity and clarity of exposition, we will omit the subscripts denoting the function keys where the context permits.

\subsection{Quantum Turing machine for DCRA-based pseudoentangled states}\label{sec:DCRA-QTM}

In this setting, the goal is to construction a QTM whose final configuration is in the following form: 
\[
\frac{1}{\sqrt{N}} \sum_{x \in \{ 0, 1\}^n} (-1)^{s(x)} \ket{x} \ket{\textsf{Auc}},
\]
where $s(\cdot) = f \circ h $ is the phase function we defined in the previous section, and $\ket{\textsf{Auc}}$ contains all other data on the tape and does not depend on $x$. Moreover, $x$, $s(x)$, and auxiliary data \textsf{Auc} are written on distinct tracks of the machine's tape. Additionally, unless explicitly stated otherwise, the data on each track will occupy the leftmost contiguous cells, while all other cells on the track will remain empty. Since $s(x)$ only has one bit, by a standard phase kickback argument, the phase state has the same entanglement structure as the following state.
\[
\frac{1}{\sqrt{N}} \sum_{x \in \{ 0, 1\}^n} \ket{x} \ket{s(x)} \ket{\textsf{Auc}}.
\]
For simplicity, we will describe the construction of this latter state throughout this subsection.
The execution of our QTM has four phases:
\begin{enumerate}
    \item \textbf{Superposition phase.} The first phase applies single-qubit Hadamard transforms to the first $n$ positions on the input track. After this step, the QTM configuration becomes $\frac{1}{\sqrt{N}} \sum_{x \in \{ 0, 1\}^n} \ket{x} \ket{\textsf{Auc}}$. It is the only phase that we apply a non-classical operation. All other operations of this QTM are reversible classical operations.

    \item \textbf{Function computation phase.} In this phase, the input track, key track, work track and the output track are involved. Suppose the input value on the input track is $x$, the goal is to \emph{reversibly} compute the value of the phase function $s(x)$ given all keys. We will describe the details of this phase later. The high level idea is as follows: We first construct an irreversible TM that compute the phase function and then apply a space-saving irreversible-to-reversible Turing machine reduction~\cite{bennett1989time} to convert it into a reversible Turing machine.

    \item \textbf{Dilution phase.} Although the input, the key, and the output of the phase function each has $O(n)$ space complexity, the total working memory needed for the function computation phase is $\Theta(n \log n)$. This implies that the size of the quantum system, as determined by the subsequent QTM-to-Hamiltonian reduction in Section~\ref{sec:QTM-to-Hamiltonian}, is also $\Theta(n \log n)$. To ensure that the scaling of entanglement is correct for any cut of the entire system -- where the system size is $\Theta(n \log n)$ -- we must distribute the $O(n)$ input and key bits evenly across the entire $\Theta(n \log n)$-length tape.  This can be achieved by inserting spaces between the bits of the input and key, effectively diluting them over the tape. Note that we have to make sure that every step of the dilution process is reversible. We will describe the detail of this phase later. The input track, key track, and two looping tracks are involved.

    \item \textbf{Padding phase.} In this final phase, the QTM performs operations that only affect the padding track and do not alter the input, key, or output data. The primary purpose of this phase is to introduce a sufficiently large number of steps that extend the overall execution time without changing the state produced by the earlier phases. As a result, the final configuration from the earlier phases remains (almost) unchanged for the majority of the QTM's runtime.

    This padding ensures that the system remains in a near-static state during most of the QTM's operation, which is crucial for later sections where we analyze the entanglement properties of the ground state of the local Hamiltonian derived from the QTM.
\end{enumerate}

Let $\mathscr{T}^* = \poly(n)$ be a positive number such that the QTM will be the padding phase after $\mathscr{T}^*$. Then for any $t \ge \mathscr{T}^*$, the state of our QTM after $t$ steps can be represented as
\[
\frac{1}{\sqrt{N}} \sum_{x \in \{ 0, 1\}^n} \ket{x \textrm{ (diluted)}} \ket{s(x)} \ket{\textsf{Auc}(t)},
\]
where the $\ket{\textsf{Auc}(t)}$ register represent all other data on the tape. It is crucial to note that $\textsf{Auc}(t)$ only depends on time $t$ and is indpendent from $x$. This is because only the content on the input and output tracks depends on $x$, while the content on other tracks is either independent of $x$ throughout the computation or remains empty after the function computation phase.

The QTM in our construction consists of $7$ tracks: 
\begin{itemize}
    \item The input track. Initially, it contains $n$ zeros for the first $n$ cells, and empty for the rest.
    \item The key track. Initially, it contains all keys used for the phase function including keys for the $4$-wise independent function family and keys for lossy functions.
    \item The work track. Initially, all cells are empty. After the function execution phase, all cells will remain empty.
    \item The output track. Initially, all cells are empty. Suppose the input track contains value $x$ after the superposition phase, then after the function execution phase, the first cell of the output track contains the function value $s(x)$, and all other cells remains empty.
    \item The dilution pattern track. The cells with labels (the leftmost cell is labeled with $0$) $k \cdot c_n/n$ for $k = 0, 1, \ldots, n - 1$ are marked with symbol $1$, the $(c_n - 1)$-th cell has symbol $2$ and the rest cells are empty. This track is read-only.
    \item The looping track. Initially, it contains the binary encoding of number $n$. This track is only involved in the dilution phase, and at any time the content on this track only depends on the total number of steps executed so far of the QTM.
    \item The padding track. Initially, it contains the binary encoding of a large polynomial in $n$, say $n^{10}$.  This track is only involved in the padding phase, and at any time the content on this track only depends on the total number of steps executed so far of the QTM.

\end{itemize}

\subsubsection{Function computation phase} 
In this section, we describe how to construct a multi-tape reversible deterministic Turing machine that computes the phase function $h(\cdot) : \{ 0, 1\}^n \to \{ 0, 1\}^{2n - \lambda}$. We first describe a RAM (random-access machine) algorithm that computes the function and then show how to convert this algorithm to a deterministic Turing machine that only uses linear space. To compute $h(\cdot)$, besides $n$ input bits $x_1, \ldots, x_n$, we need $m = \log (n / \lambda) + 1$ keys of lossy functions $(N_1, s_1, c_1), (N_2, s_2, c_2), \ldots, (N_m, s_m, c_m) $. Specifically, for the $i$-th key,  $N_i$ is an admissible RSA modulus, $N^{s_i}_i \ge 2^{2^i - 1} \cdot \lambda$, and $c_i \in \mathbb Z^*_{N_i}$. By Lemma~\ref{lem:DCRA-key-size}, the number of total bits of keys is $O(n)$.

Let $h_{[a, b]}$ be the function obtained by composing all the functions on the subtree of segment $[a, b]$. $h_{[a, b]}$ can be computed in the following steps: 
\tcbset{colback=gray!10!white, colframe=gray!50!black, 
        boxrule=0.5mm, arc=4mm, width=\textwidth, 
        auto outer arc}

\begin{tcolorbox}[title=Non-Reversible Function Computation Turing Machine]

\begin{enumerate}
    \item Check if the length of the segment $[a, b]$ is  equal to the security parameter $\lambda$. If so, use the key $(N_1, s_1, c_1)$ to compute the function value: We represent input $x_a, \ldots, x_b$ as $x_{[a, b]} \in \mathbb Z ^*_{N_1}$, and then  return $h_{[a, b]} = c_1^{x_{[a,b]}} \mod x_{[a,b]}^{s_1 +1}$ as a binary string.
    \item Otherwise, we first compute the value (as binary strings) of  $h[a, \frac{a+b}{2}]$ and $h[\frac{a+b}{2} + 1, b]$ and concatenate them. For convenience, let $\hat{h} \coloneqq h[a, \frac{a+b}{2}] \parallel h[\frac{a+b}{2} + 1, b]$ and $k \coloneqq \log ((b - a + 1 ) / \lambda) + 1$. We use the key $(N_k, s_k, c_k)$ to compute the function value: We view the binary string $\hat{y}$ as a number in $\mathbb Z^*_{N_k}$, and return $h_{[a, b]} = c_k^{\hat{y}} \mod \hat{y}^{s_k +1}$ as a binary string.
\end{enumerate}
\end{tcolorbox}

It is easy to see that the space required to compute the hash function is $O(n)$, as the total length of keys is $O(n)$, and the temporary memory needed for the calculation of $h_{[a, b]}$ is bounded by $O(b - a)$. Therefore, the total size of the call stack is bounded by $O(n)$. 

Since we can simulate RAM (random-access machine) with a multitape deterministic Turing machine within linear space and polynomial time, we now have a single tape non-reversible Turing machine that computes the hash function with space complexity $O(n)$ and time complexity $O(\mathsf{poly} (n))$. Now we choose a time $T = \poly(n)$ such that the non-reversible function computation TM halts for any $x \in \{ 0, 1\}^n$ in $T$ steps.

To make this phase reversible, we first apply Bennett's pebble game technique~\cite{bennett1989time} to transform the non-reversible TM into a multitape reversible TM with space complexity $O(n \log n)$ and time complexity $O(\mathsf{poly} (n))$. Then, we apply Theorem 5.3 in~\cite{morita2017theory} to transform the multitape reversible TM to a single tape reversible TM with only a constant factor blowup in space and a polynomial factor blowup in time. Therefore, the final single tape reversible TM still has space complexity $O(n \log n)$ and time complexity $O(\mathsf{poly} (n))$, as desired.

\subsubsection{Dilution phase}

One caveat is that the space complexity in the function computation phase is
$O(n \log n)$, which is larger than the input size of $O(n)$. This is problematic because in the QTM-to-Hamiltonian reduction described in Section~\ref{sec:QTM-to-Hamiltonian}, the size of the quantum system is determined by the space complexity of the QTM. As a result, to ensure that the final configurations are pseudo-entangled across every 1D geometrically local cut, our last dilution step inserts blank cells between non-blank cells, distributing them evenly across the entire tape.

\begin{figure}[h]
\centering

\begin{tikzpicture}
    \def\n{5}        
    \def\logn{14}    

    \node at (7, 1.5) {Input Track (Before Dilution)};

    \foreach \i in {0,...,\logn} {
        \draw (\i, 0) rectangle (\i+1, 1); 
        \ifnum\i<\n
            \node at (\i+0.5, 0.5) {$x_{\i}$}; 
        \else
            \node at (\i+0.5, 0.5) {\scriptsize$\#$}; 
        \fi
    }

\end{tikzpicture}

\vspace{1cm}

\begin{tikzpicture}
    \def\n{4}        
    \def\logn{14}    

    \node at (7, 1.5) {Dilution Pattern Track};

    \foreach \i in {0,...,\logn} {
        \draw (\i, 0) rectangle (\i+1, 1); 
    }
    
    \foreach \i in {0,...,3} {
        \node at (3 * \i+0.5, 0.5) {$1$}; 
        \node at (3 * \i+1.5, 0.5) {\scriptsize$\#$};
        \node at (3 * \i+2.5, 0.5) {\scriptsize$\#$};
    }
    \node at (3 * \n+0.5, 0.5) {$1$}; 
    \node at (3 * \n+1.5, 0.5) {\scriptsize$\#$};
    \node at (3 * \n+2.5, 0.5) {\scriptsize$2$};

\end{tikzpicture}

\vspace{1cm}

\begin{tikzpicture}
    \def\n{4}        
    \def\logn{14}    

    \node at (7, 1.5) {Input Track (After Dilution)};

    \foreach \i in {0,...,\logn} {
        \draw (\i, 0) rectangle (\i+1, 1); 
    }
    
    \foreach \i in {0,...,\n} {
        \node at (3 * \i+0.5, 0.5) {$x_{\i}$}; 
        \node at (3 * \i+1.5, 0.5) {\scriptsize$\#$};
        \node at (3 * \i+2.5, 0.5) {\scriptsize$\#$};
    }

\end{tikzpicture}

\caption{This example demonstrates how the final dilution step works: The space complexity is $15$, and the input track has $5$ non-blank cells (aka the input length is $5$). The dilution phase inserts $2$ blank cells between non-blank ones, making them evenly distributed.}
\end{figure}

To \emph{reversibly} implement the dilution step, we introduce a new track that is initially marked with all positions that we need to send each non-blank cell of the final configuration to. By utilizing the looping lemma~\cite{bernstein1997quantum}, we only needs to describe how to reversibly move one bit to its designated position.

Suppose the QTM has already moved the $k$ rightmost bits on the original input track to their respective positions, as indicated by the dilution pattern track. The following reversible TM will move the $(n - k)$-th bit of the input to its target position:

\tcbset{colback=gray!10!white, colframe=gray!50!black, 
        boxrule=0.5mm, arc=4mm, width=\textwidth, 
        auto outer arc}

\begin{tcolorbox}[title=Reversible Single-Bit Movement Turing Machine]
    \begin{enumerate}
        \item The TM starts at the leftmost cell of the tape with an internal memory initialized to an blank symbol. It moves right until it reaches the first blank cell on the input track.
        \item It then moves left, locating the rightmost bit of the input that has not yet been moved. The TM swaps the contents of its internal memory with the content of the current position on the input track.
        \item Next, it moves right until it encounters a non-blank cell or reaches the right end of the tape, as marked by the dilution pattern track.
        \item The TM then moves left until it reaches a cell marked by the dilution pattern. It swaps the contents of its internal memory with the content of the current position on the input track (an blank symbol).
        \item Finally, it moves left until it returns to the leftmost end of the tape.
    \end{enumerate}
\end{tcolorbox}

To conclude the construction of the dilution phase, we utilize the looping lemma to repeat reversible the single-bit movement exactly $n$ times. To apply the looping lemma, we store the binary representation of $n$ (the number of input bits) on the looping track of the QTM. The time complexity of this phase is $O(n^2 \log n)$, and the space complexity is $O(n \log n)$. We use $\mathscr{T}^*$ to denote the total time complexity of the function computation phase and the dilution phase.

\subsubsection{Padding Phase}

In this phase, the goal is for the system to remain nearly static in the final configuration reached after the dilution phase for a sufficiently long period. To achieve this, the QTM simply executes an empty loop for $\mathscr{T}_{\mathrm{pad}} = \poly(n)$ steps, where $\mathscr{T}_{\mathrm{pad}} \gg \mathscr{T}^*$.

To construct a reversible empty loop that runs for exactly $\mathscr{T}_{\mathrm{pad}}$ iterations, we store the binary encoding of $\mathscr{T}_{\mathrm{pad}}$ on the padding track. The looping lemma is then applied to give us a reversible TM that loops $\mathscr{T}_{\mathrm{pad}}$ times and only touches the padding track. The time complexity of this phase is $O(\mathscr{T}_{\mathrm{pad}}) = O(\poly(n))$, and the space complexity is $O(\log \mathscr{T}_{\mathrm{pad}}) = O(\log n)$.

\subsection{Quantum Turing machine for LWE-based pseudoentangled states} \label{sec:lwetm}

In the previous subsection, we demonstrated how to construct the phase state:
\[
\frac{1}{\sqrt{N}} \sum_{x \in \{ 0, 1\}^n} (-1)^{s_k(x)} \ket{x} \ket{\textsf{Auc}}.
\]

In this subsection, we extend these constructions to 2D Hamiltonians using LWE-based pseudoentangled states. Due to the $\Theta(n^2 \log n)$ space complexity required by our LWE-based phase function construction, as in Fact~\ref{fact:LWE-space}, achieving the same entanglement gap in the 1D case, as seen with DCRA-based pseudoentangled states, is not feasible. Instead, in Section~\ref{sec:LGSES-2D}, we present a 2D local Hamiltonian construction that achieves a nearly maximum entanglement gap. For the 2D case, we require a quantum Turing machine that outputs the following state:

\[
\frac{1}{\sqrt{2}}\ket{\phi_0^*}+\frac{1}{\sqrt{2}}\ket{\phi_1^*},
\]
where $\ket{\phi_0^*}$ and $\ket{\phi_1^*}$ are the same state but with different 2D dilution patterns (i.e., $\ket{\phi_1^*}$ is obtained by permuting the qubits of $\ket{\phi_0^*}$). See Figure~\ref{fig:content_qudits_0} and Figure~\ref{fig:content_qudits_1}  for further details. 
We thank Zeph Landau for suggesting this approach to us in 2023.
Here, we describe the construction of $\ket{\phi_0^*}$, as the full case can be extended by introducing a control bit and applying a different dilution pattern. Specifically, our goal is to construct the following state:
\[
\ket{\phi_0^*} = \frac{1}{\sqrt{N}} \left( \sum_{x \in \{ 0, 1\}^{\sqrt{n}}} (-1)^{s_k(x)} \ket{x} \right)^{\otimes \sqrt{n}} \ket{\textsf{Auc}},
\]
where  $ k \in  \mathcal{K}^\lo_{\sqrt{n}} \cup \mathcal{K}^\hi_{\sqrt{n}}$, and $s_k(\cdot): \{ 0, 1\}^{\sqrt{n}} \to \{ 0, 1\}$, is the LWE based phase function with size $\sqrt{n}$. The goal is to construct a QTM with this final state which has the space complexity $O(n \log n)$, and time complexity $\poly(n)$. It is important to note that all copies of the state can share the same keys of the phase function. Therefore, although the space complexity for each copy is quadratic in $\sqrt{n}$, which is $O(n \log n)$, we can sequentially construct $\sqrt{n}$ copies of such a state with method presented in the last section while reusing the keys and space for the worktape and all other auxiliary tapes.

\section{From quantum Turing machines to local Hamiltonians}\label{sec:QTM-to-Hamiltonian}

\noindent In this section, we demonstrate how to construct a 1D Hamiltonian on a line of $n$ qudits (Definition~\ref{def:local-hamiltonians}), given a quantum Turing machine $M$ and a valid input, such that its ground state encodes the QTM's history of the first $\poly n$ steps, assuming $M$ has space complexity at most $n$ and does not halt within this period. This construction directly leads to a 2D Hamiltonian on a $\sqrt{n}\times\sqrt{n}$ grid with the same ground state, as detailed in Section~\ref{sec:LGSES-2D}.

Our approach builds upon on the previous work by Gottesman and Irani~\cite{gottesman2009quantum}, where we first construct a multi-track QTM $\Mnew$ that simulates the original QTM $M$ whose head moves in a very simple pattern: it shuttles back and forth between to end points on the tape, whose locations are specified in the input. Given this pattern, we construct a Hamiltonian whose ground state is the uniform superposition over all the configurations that appear when running $\Mnew$ by introducing energy penalty terms so that any state not conforming to this form has non-zero energy.

Compared to~\cite{gottesman2009quantum}, our construction differs in two key aspects. First, in our QTM for pseudoentangled states, there is a specific input—the cryptographic key—that we need to encode in the QTM's running history. To handle this, we introduce an additional initialization phase in $\Mnew$, where the head sweeps through the tape to verify the input position by position. This adjustment comes at the cost of losing translation invariance in our Hamiltonian. Second, we develop a novel clock construction with $\poly(n)$ distinct clock configurations using $n$ qudits, where configurations with adjacent clock values differ only locally. This improvement allows us to encode the history of a TM with $\poly(n)$ time complexity in $n$ qudits.

\subsection{Local Hamiltonians for classical reversible Turing machines\label{sec:classical-TM-Hamiltonian}}

We begin by considering a simplified case where the goal is to construct a local Hamiltonian whose ground state encodes the computation history of the first $\Theta(n^k)$ steps of a classical reversible Turing machine $M=(Q, \Sigma, \delta, q_0, Q_{acc})$ with space complexity $n$, assuming it does not halt within this time. Assume we know $n$ and $k$. As mentioned before, we first construct a new $(k+2)$-track Turing machine $\Mnew$ that simulates the given TM. The alphabet of the $\Mnew$ contains the following symbols:
\begin{itemize}
\item Universal symbols: \leftend, \rightend, \blankL, \blankR.
\item Original work tape symbols: all the symbols from the original Turning machine work tape alphabet $\Sigma$.
\item Original Turing machine state symbols: all the symbols from the original Turing machine state set $Q$.
\item Clock tape symbols: \arrR, \arrL, \arrRzero, \arrLzero, \blankL.
\end{itemize}
\begin{figure}[htbp]
\begin{center}
\begin{small}
\begin{tabular}{|@{}c@{}|@{}c@{}|@{}c@{}|}
\hline
$\leftend$ &
\begin{tabular}{@{}c@{}|@{}c@{}|clc|@{}c@{}|@{}c@{}}
\arrRzero & \blankR & $\cdots$ & ``Track 0'': ``Clock track 0'' (second hand) & $\cdots$ & \blankR &  \blankR \\
\hline
\arrR & \blankR  & $\cdots$ & Track 1: Clock track 1 (minute hand) & $\cdots$ & \blankR &  \blankR \\
\hline
\vdots & \vdots &  & \qquad\qquad\qquad\qquad\ \vdots &  & \vdots &  \vdots \\
\hline
\arrR & \blankR  & $\cdots$ & Track $k$: Clock track $k$ & $\cdots$ & \blankR &  \blankR \\
\hline
$q_0$ & \blankR   & $\cdots$ & Track $k+1$: Original TM head location and state  & $\cdots$ & \blankR   &  \blankR   \\
\hline
\# & \# & $\cdots$  & Track $k+2$: Original TM input written on the work tape & $\cdots$ & \# &  \# \\
\hline
\end{tabular}
& $\rightend$ \\
\hline
\end{tabular}
\end{small}
\end{center}
\caption{The initial configuration of the $(k+2)$-track TM}
\label{fig:initial-configuration-classical}
\end{figure}

The initial configuration $\Mnew$ is depicted in Figure~\ref{fig:initial-configuration-classical}, where, for simplicity, the head location and state are written on a special 'Track 0.' Tracks 0 to $k$ are called clock tracks and drive the simulation of the original Turing machine, whose configuration is stored on the last two tracks. The length of the tape is $n+2$, i.e., there are $n$ symbols on each track other than \leftend\ and \rightend.

\begin{definition}[Valid configuration]\label{def:valid-configuration}
We say a configuration is valid if it satisfies the following conditions. 
\begin{itemize}
\item The leftmost position is $\leftend$ and the rightmost position is $\rightend$;
\item Track 0 is of the form $\leftend \blankL^*\arrR\blankR^* \rightend$, $\leftend \blankL^*\arrL\blankR^* \rightend$, $\leftend \blankL^*\arrRzero\blankR^* \rightend$, or $\leftend \blankL^*\arrLzero\blankR^* \rightend$;
\item Each track from track 1 to $k$ is of the form $\leftend \blankL^*\arrR\blankR^* \rightend$ or $\leftend \blankL^*\arrL\blankR^* \rightend$;
\item The $(k+1)$-th track is of the form $\leftend \blankL^* q\,\blankR^* \rightend$,
\end{itemize}
where $q$ is some state of the original Turing machine. We use $\mathcal{C}$ to denote the set of valid configurations, and use $\mathcal{S}(\mathcal{C})$ to denote the subspace spanned by all configurations in $\mathcal{C}$.
\end{definition}

Starting from the initial configuration in Figure~\ref{fig:initial-configuration-classical}, this new TM will evolve under the set of transition rules before it halts. There are two phases of transition rules: the initialization phase and the computation phase, where in each phase there are some standard basis states with consecutive clock values. Compare to~\cite{gottesman2009quantum},  our construction does not require a counting phase, as we can verify whether the input is correctly encoded during the initialization phase, since our Hamiltonian need not be translationally invariant. However, we introduce an additional padding phase, during which all standard basis states correspond to the final configuration of the Turing machine, while only the clock value advances. We will introduce the transition rules in these two phases respectively in Section~\ref{sec:initialization-transition-rules} and Section~\ref{sec:computation-transition-rules}. Then in Section~\ref{sec:new-TM-properties} we discuss useful properties of this new TM, and in Section~\ref{sec:local-Hamiltonian-classical}, we introduce the Hamiltonian whose ground state is a uniform superposition over all the history configurations of this new TM.

\subsubsection{Transition rules in the initialization phase}\label{sec:initialization-transition-rules}

There are four possible types of symbols in the clock track in this phase: \arrR, \arrL, \blankL, and \blankR. Starting from the initial configuration in Figure~\ref{fig:initial-configuration-classical}, track $0$ to track $k$ will be of the form $\leftend \arrR\blankR^* \rightend$ and remain unchanged. The transition rules move the head of the TM, i.e., the arrow in track 0, in the direction it point, despite the content on other tracks.
\begin{enumerate}
\item $\arrR \blankR \rightarrow \blankL \arrR$, $\blankL \arrL \rightarrow \arrL \blankR$: the arrow moves in the direction it is pointing;
\item $\arrR \rightend \rightarrow \arrL\rightend$: the arrow changes
direction when it reaches the right endpoint.
\end{enumerate}
When the head moves from left to right as the clock value increases, it will check along the way whether the input is correctly written on the Turing machine work tape, which will be described in detail in Section~\ref{sec:local-Hamiltonian-classical}. When the head goes back to the left end, it changes both its direction and its shape 
\begin{align}
\leftend \arrL \rightarrow \leftend \arrRzero,
\end{align}
where the configuration on the left is the end of the initialization phase and the configuration on the right is the beginning of the computation phase.

The content on track 1 to track $(k+2)$ does not changed throughout the entire initialization phase.

\subsubsection{Transition rules in the computation phase}\label{sec:computation-transition-rules}

Similar to the initialization phase, the transition rules move the head, i.e., the arrow on track 0, in the direction it points, despite the content on other tracks.
\begin{enumerate}
\item $\arrRzero \blankR \rightarrow \blankL \arrRzero$, $\blankL \arrLzero \rightarrow \arrLzero \blankR$: the arrows move in the direction they are pointing;
\item $\arrRzero \rightend \rightarrow \arrLzero \rightend$, $\leftend\arrLzero\rightarrow\leftend\arrRzero$: the arrows change direction when they reach an endpoint.
\end{enumerate}

\paragraph{Transition rules of the clock tracks.}
In the computation phase, the transition of the clock tracks is triggered by the movement of the head and is independent of the last two tracks. There are four possible types of symbols on track 1 to track $k$ in this phase: \arrR, \arrL, \blankL, and \blankR. Moreover, they will be of the form $\leftend \blankL^*\arrR\blankR^* \rightend$. For brevity, we refer to the arrow on the \(j\)-th track as the \(j\)-th clock arrow. 

The transition of the first clock arrow is triggered by the movement of the head, i.e., the zeroth clock arrow. If the first clock arrow is pointing right, when the head sweeps from right to left passing the first clock arrow, it triggers the advance of the first clock arrow to the right. Similarly, if the first clock arrow is pointing left, when the head sweeps from left to right passing the first clock arrow, it triggers the advance of the first clock arrow to the left,
\begin{align}
\fourcells{\blankL}{\arrRzero}{\arrL}{\blankR} \rightarrow \fourcells{\arrLzero}{\blankL}{\blankR}{\arrR}\,,\qquad\fourcells{\arrRzero}{\blankL}{\blankR}{\arrL} \rightarrow \fourcells{\blankL}{\arrL}{\arrRzero}{\blankR}\,.
\end{align}
If the first clock arrow is on the boundary, the haed will change its direction when passing through it oppositely.
\begin{align}
\sixcellsL{\arrRzero}{\arrL}{\blankR}{\blankR}\rightarrow \sixcellsL{\blankL}{\arrR}{\arrRzero}{\blankR}\,,\qquad
\sixcellsR{\blankL}{\blankL}{\arrLzero}{\arrR}\rightarrow \sixcellsR{\arrLzero}{\blankL}{\blankR}{\arrL}\,.
\end{align}
We call one round trip of the head, i.e., the zeroth clock arrow from the left end to the right end and back to be an \textit{iteration}, which is further generalized to other clock arrows on track 1 to track $k$ in Definition~\ref{def:iteration}. As a concrete example, the configuration of track 0 and track 1 in the first five iterations of the head are shown in Figure~\ref{fig:first-five-iterations}.
\begin{figure}[htbp]
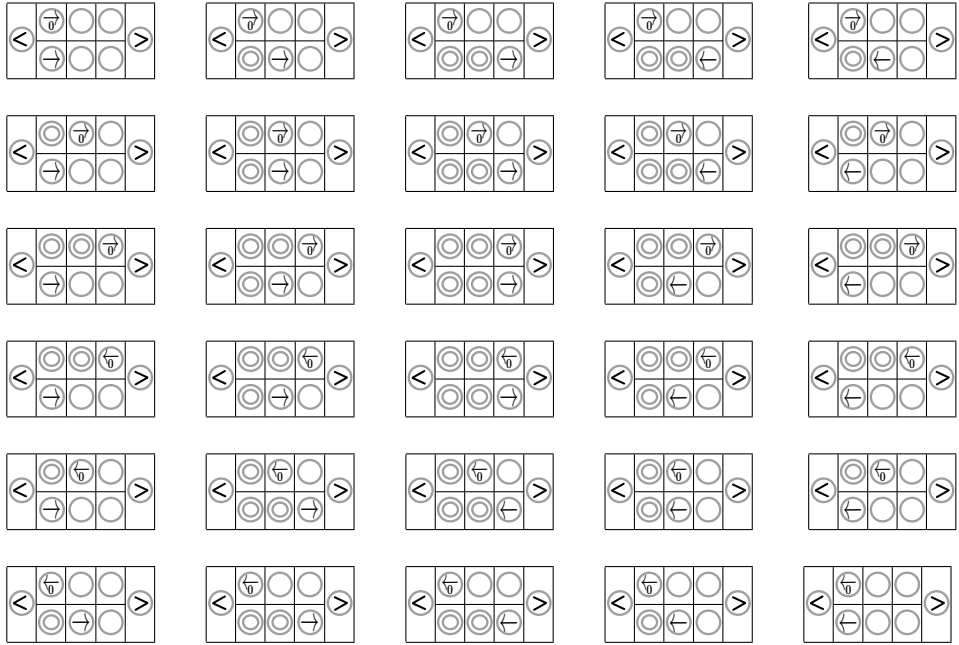

\begin{center}
\begin{tabular}{ccccc}
\sixcellsLR{\arrRzero}{\blankR}{\blankR}{\arrR}{ \blankR}{\blankR}
\qquad
&  \sixcellsLR{\arrRzero}{\blankR}{\blankR}{\blankL}{\arrR}{\blankR} 
\qquad
&  \sixcellsLR{\arrRzero}{\blankR}{\blankR}{\blankL}{\blankL}{\arrR} 
\qquad
&  \sixcellsLR{\arrRzero}{\blankR}{\blankR}{\blankL}{\blankL}{\arrL}
\qquad
&  \sixcellsLR{\arrRzero}{\blankR}{\blankR}{\blankL}{\arrL}{\blankR}\\
& \\
\sixcellsLR{\blankL}{\arrRzero}{\blankR}{\arrR}{ \blankR}{\blankR}
\qquad 
 & \sixcellsLR{\blankL}{\arrRzero}{\blankR}{\blankL}{\arrR}{\blankR} 
 \qquad
&  \sixcellsLR{\blankL}{\arrRzero}{\blankR}{\blankL}{\blankL}{\arrR} 
\qquad
&  \sixcellsLR{\blankL}{\arrRzero}{\blankR}{\blankL}{\blankL}{\arrL}
\qquad
&  \sixcellsLR{\blankL}{\arrRzero}{\blankR}{\arrL}{\blankR}{\blankR}\\
& \\
\sixcellsLR{\blankL}{\blankL}{\arrRzero}{\arrR}{ \blankR}{\blankR} 
\qquad 
&  \sixcellsLR{\blankL}{\blankL}{\arrRzero}{\blankL}{\arrR}{\blankR} 
\qquad
&  \sixcellsLR{\blankL}{\blankL}{\arrRzero}{\blankL}{\blankL}{\arrR}
\qquad
&  \sixcellsLR{\blankL}{\blankL}{\arrRzero}{\blankL}{\arrL}{\blankR}
\qquad
&  \sixcellsLR{\blankL}{\blankL}{\arrRzero}{\arrL}{\blankR}{\blankR}\\
& \\
\sixcellsLR{\blankL}{\blankL}{\arrLzero}{\arrR}{ \blankR}{\blankR} \qquad 
  & \sixcellsLR{\blankL}{\blankL}{\arrLzero}{\blankL}{\arrR}{\blankR} 
  \qquad
&  \sixcellsLR{\blankL}{\blankL}{\arrLzero}{\blankL}{\blankL}{\arrR}
\qquad
&  \sixcellsLR{\blankL}{\blankL}{\arrLzero}{\blankL}{\arrL}{\blankR}
\qquad
&  \sixcellsLR{\blankL}{\blankL}{\arrLzero}{\arrL}{\blankR}{\blankR}\\
  & \\
\sixcellsLR{\blankL}{\arrLzero}{\blankR}{\arrR}{ \blankR}{\blankR} \qquad 
  & \sixcellsLR{\blankL}{\arrLzero}{\blankR}{\blankL}{\blankL}{\arrR} \qquad
&  \sixcellsLR{\blankL}{\arrLzero}{\blankR}{\blankL}{\blankL}{\arrL}
\qquad
&  \sixcellsLR{\blankL}{\arrLzero}{\blankR}{\blankL}{\arrL}{\blankR}
\qquad
&  \sixcellsLR{\blankL}{\arrLzero}{\blankR}{\arrL}{\blankR}{\blankR}\\
  & \\
\sixcellsLR{\arrLzero}{\blankR}{\blankR}{\blankL}{\arrR}{\blankR} \qquad 
  & \sixcellsLR{\arrLzero}{\blankR}{\blankR}{\blankL}{\blankL}{\arrR} 
  \qquad
&  \sixcellsLR{\arrLzero}{\blankR}{\blankR}{\blankL}{\blankL}{\arrL}
\qquad
&  \sixcellsLR{\arrLzero}{\blankR}{\blankR}{\blankL}{\arrL}{\blankR}
\qquad
&  \sixcellsLR{\arrLzero}{\blankR}{\blankR}{\arrL}{\blankR}{\blankR}
\\
\end{tabular}
\caption{The configuration of track 0 and track 1 on a chain with $n=3$ in the first five iterations of the head. The order is from top to bottom and then from left to right.}
\label{fig:first-five-iterations}
\end{center}
\end{figure}

Similarly for any $j\geq 1$, the transition of the $(j+1)$-th arrow is triggered by the movement of the $j$-th arrow. If the $(j+1)$-th arrow is pointing right, when the $j$-th arrow sweeps from right to left passing the $(j+1)$-th arrow, it triggers the advance of the $(j+1)$-th arrow to the right. Similarly, if the $(j+1)$-th arrow is pointing left, when the $j$-th arrow sweeps from left to right passing the $(j+1)$-th arrow, it triggers the advance of the $(j+1)$-th arrow to the left. If the $(j+1)$-th arrow is on the boundary, the $j$-th arrow will change its direction when passing through it oppositely. As a concrete example, the configurations of the clock tracks with $n=3$ and $k=2$ in the fifth iteration of the head are shown in Figure~\ref{fig:the-fifth-iteration}.

\begin{figure}[htbp]
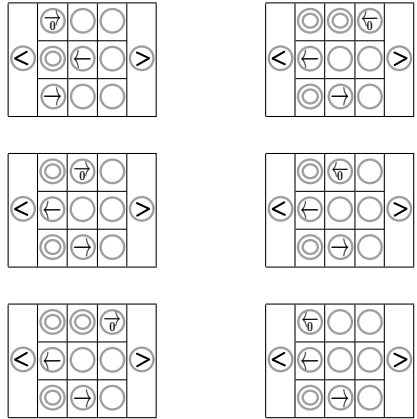

\begin{center}
\begin{tabular}{cc}
\ninecellsLR{\arrRzero}{\blankR}{\blankR}{\blankL}{\arrL}{\blankR}{\arrR}{\blankR}{\blankR} \qquad\qquad
&  \ninecellsLR{\blankL}{\blankL}{\arrLzero}{\arrL}{\blankR}{\blankR}{\blankL}{\arrR}{\blankR} \\
& \\
\ninecellsLR{\blankL}{\arrRzero}{\blankR}{\arrL}{\blankR}{\blankR}{\blankL}{\arrR}{\blankR} \qquad\qquad
 & \ninecellsLR{\blankL}{\arrLzero}{\blankR}{\arrL}{\blankR}{\blankR}{\blankL}{\arrR}{\blankR} \\
& \\
\ninecellsLR{\blankL}{\blankL}{\arrRzero}{\arrL}{\blankR}{\blankR}{\blankL}{\arrR}{\blankR} \qquad\qquad
& \ninecellsLR{\arrLzero}{\blankR}{\blankR}{\arrL}{\blankR}{\blankR}{\blankL}{\arrR}{\blankR} \\
\end{tabular}
\caption{The configuration of track 0 to track 2 on a chain with $n=3$ in the fifth iteration of the first clock arrow. The order is from top to bottom and then from left to right.}
\label{fig:the-fifth-iteration}
\end{center}
\end{figure}

The transition stops when the $k$-th arrow on track $k$ reaches the right end. As a concrete example, the configurations of track 0 to track $k=2$ in the last two iterations of the first clock arrow are shown in Figure~\ref{fig:the-last-iteration-clock}. Note that the last iteration is not a complete one because the first clock arrow will not return to the left end. 

\begin{figure}[htbp]
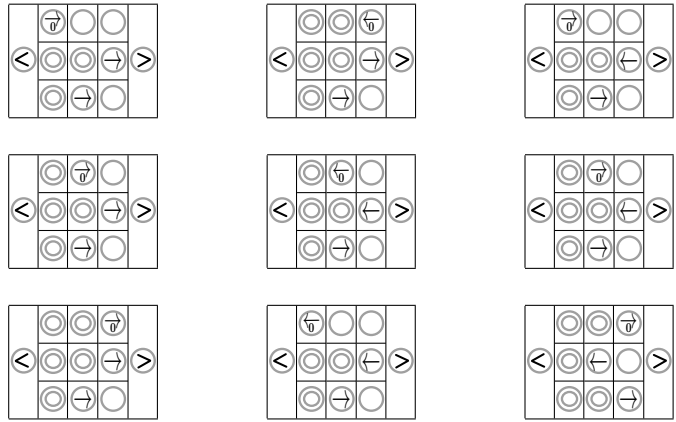

\begin{center}
\begin{tabular}{ccc}
\ninecellsLR{\arrRzero}{\blankR}{\blankR}{\blankL}{\blankL}{\arrR}{\blankL}{\arrR}{\blankR}
\qquad\qquad
&
\ninecellsLR{\blankL}{\blankL}{\arrLzero}{\blankL}{\blankL}{\arrR}{\blankL}{\arrR}{\blankR}
\qquad\qquad
&  \ninecellsLR{\arrRzero}{\blankR}{\blankR}{\blankL}{\blankL}{\arrL}{\blankL}{\arrR}{\blankR} \\
& \\
\ninecellsLR{\blankL}{\arrRzero}{\blankR}{\blankL}{\blankL}{\arrR}{\blankL}{\arrR}{\blankR}
\qquad\qquad
&
\ninecellsLR{\blankL}{\arrLzero}{\blankR}{\blankL}{\blankL}{\arrL}{\blankL}{\arrR}{\blankR}
\qquad\qquad
 & \ninecellsLR{\blankL}{\arrRzero}{\blankR}{\blankL}{\blankL}{\arrL}{\blankL}{\arrR}{\blankR} \\
& \\
\ninecellsLR{\blankL}{\blankL}{\arrRzero}{\blankL}{\blankL}{\arrR}{\blankL}{\arrR}{\blankR}
\qquad\qquad
& \ninecellsLR{\arrLzero}{\blankR}{\blankR}{\blankL}{\blankL}{\arrL}{\blankL}{\arrR}{\blankR}
\qquad\qquad
& \ninecellsLR{\blankL}{\blankL}{\arrRzero}{\blankL}{\arrL}{\blankR}{\blankL}{\blankL}{\arrR} \\
\end{tabular}
\caption{The configuration of all the three clock tracks on a chain with $n=3$ in the last two iterations of the first clock arrow. The order is from top to bottom and then from left to right.}
\label{fig:the-last-iteration-clock}
\end{center}
\end{figure}

\paragraph{Transition rules of the last two tracks.}
The last two tracks form a specific configuration of the original TM. If in the original TM its head is set to move left in the next step, the head of the new Turing machine, i.e., the arrow on track 0, will trigger this movement as it approaches the position of $q$ from the left,
\begin{align}
\sixcells{\arrR}{\blankR}{\blankL}{q}{\variable}{a} \rightarrow \sixcells{\blankL}{\arrR}{q_L}{\blankR}{\variable}{b}\,,
\end{align}
where $q_L$ is the next state of the original TM, the tape symbol is overwritten by some $b$ specified by the transition function, and $\variable$ denotes any tape symbol and remains unchanged in this transition. Similarly, if in the original TM its head is set to move right in the next step, the first clock arrow will trigger this movement as it approaches the position of $q$ from the right,
\begin{align}
\sixcells{\blankL}{\arrL}{q}{\blankR}{a}{\variable} \rightarrow \sixcells{\arrL}{\blankR}{\blankL}{q_R}{\variable}{b}\,,
\end{align}
where $q_R$ is the next state of the TM. 

Otherwise, the head just sweeps through leaving the last two tracks unchanged,
\begin{align}
\sixcells{\arrR}{\blankR}{\blankL}{\blankL}{\variable}{\variable} \rightarrow \sixcells{\blankL}{\arrR}{\blankL}{\blankL}{\variable}{\variable}\,,\quad
\sixcells{\arrR}{\blankR}{\blankR}{\blankR}{\variable}{\variable} \rightarrow \sixcells{\blankL}{\arrR}{\blankR}{\blankR}{\variable}{\variable}\,,\quad
\sixcells{\blankL}{\arrL}{\blankL}{\blankL}{\variable}{\variable} \rightarrow \sixcells{\arrL}{\blankR}{\blankL}{\blankL}{\variable}{\variable}\,,\quad
\sixcells{\blankL}{\arrL}{\blankR}{\blankR}{\variable}{\variable} \rightarrow \sixcells{\arrL}{\blankR}{\blankR}{\blankR}{\variable}{\variable}\,.
\end{align}

\subsubsection{Properties of the new Turing machine}\label{sec:new-TM-properties}

\begin{fact}\label{fact:classical-valid-configuration-garantee}
For any input $w \in \Sigma^n$ to the original TM $M$, starting from the initial configuration in Figure~\ref{fig:initial-configuration-classical} and evolving according to the transition rules, the new TM $\Mnew$ will be in a valid configuration at every time step, as defined in Definition~\ref{def:valid-configuration}.
\end{fact}

\begin{definition}\label{def:iteration}
For any $j\leq k-2$, we call one round trip of the $j$-th arrow from the left end in the shape $\leftend \arrR\blankR^* \rightend$ to the right end and back to the shape $\leftend \arrR\blankR^* \rightend$ to be an iteration of the $j$-th error.
\end{definition}

\begin{lemma}
Suppose there are $k$ clock tracks and each track has length $n\geq 2$. Then, the new Turing machine has time complexity $\mathscr{T}=\Theta(n^{k+1})$ for any input $w\in\Sigma^n$.
\end{lemma}
\begin{proof}
Since $\Mnew$ halts if and only if the clock tracks reach a specific construction, its time complexity is independent from the input $w\in\Sigma^n$. Note that $\Mnew$ goes through $2n$ steps in the initialization phase, where only the first clock arrow completes an iteration and all the other arrows do not move.

As for the computation phase, at the beginning the clock tracks are in the following configuration:
\begin{center}
\begin{small}
\begin{tabular}{|@{}c@{}|@{}c@{}|@{}c@{}|}
\hline
$\leftend$ &
\begin{tabular}{@{}c@{}|@{}c@{}|clc|@{}c@{}|@{}c@{}}
\arrR & \blankR & $\cdots$ & ``Track 0: Clock track 0 (second hand)'' & $\cdots$ & \blankR &  \blankR \\
\hline
\arrR & \blankR  & $\cdots$ & Track 1: Clock track 1 (minute hand) & $\cdots$ & \blankR &  \blankR \\
\hline
\vdots & \vdots &  & \qquad\qquad\qquad\qquad\ \vdots &  & \vdots &  \vdots \\
\hline
\arrR & \blankR  & $\cdots$ & Track $k-1$: Clock track $k-1$ & $\cdots$ & \blankR &  \blankR \\
\hline
\arrR & \blankR  & $\cdots$ & Track $k$: Clock track $k$ & $\cdots$ & \blankR &  \blankR \\
\hline
\end{tabular}
& $\rightend$ \\
\hline
\end{tabular}\,.
\end{small}
\end{center}

For any $j \geq 2$, the transition rules that applies to the $(j+1)$-th arrow can be summarized as follows.
\begin{itemize}
    \item If the arrow is not at the right end and is pointing right, one iteration of the $j$-th arrow will move it one position to the right.
    \item If the arrow is not at the left end and is pointing left, one iteration of the $j$-th arrow will move it one position to the left.
    \item If the arrow is at the right end and pointing right, one iteration of the $j$-th arrow will change its direction to left, keeping its position unchanged.
    \item If the arrow is at the left end, regardless of its direction, one iteration of the $j$-th arrow will move it one position to the right and change its direction to right if it is not already so.
\end{itemize}
Hence, to reach the following configuration of the clock tracks
\begin{center}
\begin{small}
\begin{tabular}{|@{}c@{}|@{}c@{}|@{}c@{}|}
\hline
$\leftend$ &
\begin{tabular}{@{}c@{}|@{}c@{}|clc|@{}c@{}|@{}c@{}}
\arrL & \blankR & $\cdots$ & ``Track 0'': Clock track 0 (second hand) & $\cdots$ & \blankR &  \blankR \\
\hline
\arrL & \blankR  & $\cdots$ & Track 1: Clock track 1 (minute hand) & $\cdots$ & \blankR &  \blankR \\
\hline
\vdots & \vdots &  & \qquad\qquad\qquad\qquad\ \vdots &  & \vdots &  \vdots \\
\hline
\arrL & \blankR  & $\cdots$ & Track $k-1$: Clock track $k-1$ & $\cdots$ & \blankR &  \blankR \\
\hline
\blankL & \blankL  & $\cdots$ & Track $k$: Clock track $k$ & $\cdots$ & \arrR &  \blankR \\
\hline
\end{tabular}
& $\rightend$ \\
\hline
\end{tabular}\,,
\end{small}
\end{center}
the $k$-th arrow must go through $n-2$ steps, where to trigger the $k$-th arrow to move one step, the $(k-1)$-th arrow must go through one iteration and moves for $(2n-1)$ steps on track $(k-1)$. Hence, the $(k-1)$-th arrow must go through $n-2$ iterations in total, and the $(k-2)$-th arrow must go through $(2n-1)(n-2)$ iterations in total. By recursion, the 0-th arrow must go through $(2n-5)^{k-1}(n-2)$ iterations, where in each iteration of the 0-th arrow there are $2n$ time steps. Hence, there are $2n(2n-1)^{k-1}(n-2)=\Theta(n^{k+1})$ time steps in total from the beginning of the computation phase. 
\begin{center}
\begin{small}
\begin{tabular}{|@{}c@{}|@{}c@{}|@{}c@{}|}
\hline
$\leftend$ &
\begin{tabular}{@{}c@{}|@{}c@{}|clc|@{}c@{}|@{}c@{}}
\vdots & \vdots &  & \qquad\qquad\qquad\qquad\ \vdots &  & \vdots &  \vdots \\
\hline
\blankL & \blankL  & $\cdots$ & Track $k-2$: Clock track $k-4$ & $\cdots$ & \blankL &  \arrR \\
\hline
\blankL & \blankL & $\cdots$ & Track $k-1$: Clock track $k-3$ & $\cdots$ & \arrL &  \blankR \\
\hline
\blankL & \blankL  & $\cdots$ & Track $k$: Clock track $k-2$ & $\cdots$ & \blankL &  \arrR \\
\hline
\end{tabular}
& $\rightend$ \\
\hline
\end{tabular}\,,
\end{small}
\end{center}
From this configuration to the final configuration where the $k$-th arrow reaches the right end, since the $k$-th arrow only moves one step, by similar recursion we can show that the zeroth arrow will go through at most $2(2n-1)^{k-1}$ iterations, which contribute to at most $4n(2n-1)^{k-1}=O(n^{k})$ time steps. Counting in the time steps in the initialization phase, we can conclude that the overall time complexity of the new Turing machine is $\Theta(n^{k+1})$.
\end{proof}

\subsubsection{Construction of the local Hamiltonian}\label{sec:local-Hamiltonian-classical}
We now construct the local Hamiltonian whose ground state encodes the history of $\Mnew$. This ground state will be a uniform superposition over all configurations of $\Mnew$, starting from the initial configuration  Figure~\ref{fig:initial-configuration-classical}. Each configuration is represented using $\Theta(n)$ qudits, corresponding to the $\Theta(n)$ different positions in the configuration. The local dimension is $k+3$, accounting for the number of tracks, including `track 0' which stores the head's position and state. The local Hamiltonian will consist of three sets of terms: one to enforce valid configurations, one to enforce transition rules, and one to verify the input. Throughout this subsection, for any set of valid configurations $\mathcal{C}_{\mathrm{config}}$, we use $\mathcal{S}(\mathcal{C}_{\mathrm{config}})$ to denote the subspace spanned by all the configurations $\ket{c}\in \mathcal{C}_{\mathrm{config}}$.

\paragraph{Local Hamiltonian terms enforcing valid configurations.}
We first introduce the set of local Hamiltonian terms that enforce that the configuration is valid as per Definition~\ref{def:valid-configuration}. We note that for any configuration that is not valid, it must violate at least one rule in Definition~\ref{def:valid-configuration}, and these violations can be detected locally. We assign an energy penalty to each potential violation, ensuring that the ground state must be in the subspace spanned by valid configurations.

To ensure that the leftmost position is $\leftend$ and the rightmost position is $\rightend$, we put Hamiltonian terms
\begin{align}
\sum_{\variable\neq\leftend}\ket{\variable}_1\bra{\variable}_1,\qquad
\sum_{\variable\neq\rightend}\ket{\variable}_{n+2}\bra{\variable}_{n+2}
\end{align}
separately at the leftmost qudit and the rightmost qudit respectively, with subscripts indicating the positions of the Hamiltonian terms.

To ensure that track 0 is of the form $\leftend \blankL^*\arrR\blankR^* \rightend$, $\leftend \blankL^*\arrL\blankR^* \rightend$, $\leftend \blankL^*\arrRzero\blankR^* \rightend$, or $\leftend \blankL^*\arrLzero\blankR^* \rightend$, for each pair of adjacent qudits $\nu$ and $\nu+1$ that are not on the boundary, we put the following terms
\begin{align}\label{eqn:track0-forbidden-0}
\sum_{\small{\variable}_1\in\{\arrL,\arrR,\arrLzero,\arrRzero\}}\sum_{\small{\variable}_2\in\{\arrL,\arrR,\arrLzero,\arrRzero\}}\ket{\,\twocellshor{\small{\variable}_1}{\small{\variable}_2}\,}_{\nu,\nu+1}\bra{\,\twocellshor{\small{\variable}_1}{\small{\variable}_2}\,}_{\nu,\nu+1}
\end{align}
and
\begin{align}\label{eqn:track0-forbidden-1}
&\sum_{\small{\variable}\in\{\blankR,\arrL,\arrR,\arrLzero,\arrRzero\}}
\ket{\,\twocellshor{\variable}{\blankL}\,}_{\nu,\nu+1}\bra{\,\twocellshor{\variable}{\blankL}\,}_{\nu,\nu+1}\nonumber\\
&\qquad\qquad+\sum_{\small\variable\in\{\blankL,\arrL,\arrR,\arrLzero,\arrRzero\}}
\ket{\,\twocellshor{\blankR}{\variable}\,}_{\nu,\nu+1}\bra{\,\twocellshor{\blankR}{\variable}\,}_{\nu,\nu+1}+\ket{\,\twocellshor{\blankL}{\blankR}\,}_{\nu,\nu+1}\bra{\,\twocellshor{\blankL}{\blankR}\,}_{\nu,\nu+1},
\end{align}
in track 0. On the leftmost position and its neighbor as well as on the rightmost position and its neighbor in track 0, we put the following terms respectively,
\begin{align}\label{eqn:track0-forbidden-2}
\ket{\,\twocellshor{\leftend}{\blankR}\,}_{1,2}\bra{\,\twocellshor{\leftend}{\blankR}\,}_{1,2},\qquad\ket{\,\twocellshor{\blankL}{\rightend}\,}_{n+1,n+2}\bra{\,\twocellshor{\blankL}{\rightend}\,}_{n+1,n+2}.
\end{align}
The first term in Eq.~\ref{eqn:track0-forbidden-1} enforces that, aside from $\leftend$, only $\blankL$ can appear to the left of another $\blankL$. Similarly, the second term in Eq.~\ref{eqn:track0-forbidden-2} ensures that, aside from $\rightend$, only $\blankR$ can appear to the right of another $\blankR$. The third term in Eq.~\ref{eqn:track0-forbidden-1}, together with the terms in Eq.~\ref{eqn:track0-forbidden-2}, guarantees the presence of an arrow symbol on the track. Finally, these terms, along with those in Eq.~\ref{eqn:track0-forbidden-0}, ensure that there is at most one arrow in track 0.

Similarly, for any $1\leq j\leq k$, to ensure that track $j$ is of the form $\leftend \blankL^*\arrR\blankR^* \rightend$, $\leftend \blankL^*\arrL\blankR^* \rightend$, for each pair of adjacent qudits $\nu$ and $\nu+1$ that are not on the boundary, we put the following terms
\begin{align}
\sum_{\small{\variable}_1,\small{\variable}_2\in\{\arrL,\arrR\}}\ket{\,\twocellshor{\small{\variable}_1}{\small{\variable}_2}\,}_{\nu,\nu+1}\bra{\,\twocellshor{\small{\variable}_1}{\small{\variable}_2}\,}_{\nu,\nu+1}
\end{align}
and
\begin{align}
&\sum_{\small{\variable}\in\{\blankR,\arrL,\arrR\}}
\ket{\,\twocellshor{\variable}{\blankL}\,}_{\nu,\nu+1}\bra{\,\twocellshor{\variable}{\blankL}\,}_{\nu,\nu+1}\nonumber\\
&\qquad\qquad+\sum_{\small{\variable}\in\{\blankL,\arrL,\arrR\}}
\ket{\,\twocellshor{\blankR}{\variable}\,}_{\nu,\nu+1}\bra{\,\twocellshor{\blankR}{\variable}\,}_{\nu,\nu+1}+\ket{\,\twocellshor{\blankL}{\blankR}\,}_{\nu,\nu+1}\bra{\,\twocellshor{\blankL}{\blankR}\,}_{\nu,\nu+1}.
\end{align}
in track $j$. On the leftmost position and its neighbor as well as on the rightmost position and its neighbor in track $j$, we additionally put the following terms respectively,
\begin{align}
\ket{\,\twocellshor{\leftend}{\blankR}\,}_{1,2}\bra{\,\twocellshor{\leftend}{\blankR}\,}_{1,2},\qquad\ket{\,\twocellshor{\blankL}{\rightend}\,}_{n+1,n+2}\bra{\,\twocellshor{\blankL}{\rightend}\,}_{n+1,n+2}.
\end{align}

To ensure that track $(k+1)$ is of the form $\leftend \blankL^*q\,\blankR^* \rightend$, for each pair of adjacent qudits $\nu$ and $\nu+1$ that are not on the boundary, we put the following terms
\begin{align}
\sum_{q_1,q_2\in Q}\ket{\,\twocellshor{q_1}{q_2}\,}_{\nu,\nu+1}\bra{\,\twocellshor{q_1}{q_2}\,}_{\nu,\nu+1}
\end{align}
and
\begin{align}
&\sum_{\small{\variable}\in\{\small{\blankR}\} \cup Q}
\ket{\,\twocellshor{\variable}{\blankL}\,}_{\nu,\nu+1}\bra{\,\twocellshor{\variable}{\blankL}\,}_{\nu,\nu+1}\nonumber\\
&\qquad\qquad+\sum_{\small{\variable}\in\{\small{\blankL}\}\cup Q}
\ket{\,\twocellshor{\blankR}{\variable}\,}_{\nu,\nu+1}\bra{\,\twocellshor{\blankR}{\variable}\,}_{\nu,\nu+1}+\ket{\,\twocellshor{\blankL}{\blankR}\,}_{\nu,\nu+1}\bra{\,\twocellshor{\blankL}{\blankR}\,}_{\nu,\nu+1}.
\end{align}
in track $(k+1)$. On the leftmost position and its neighbor as well as on the rightmost position and its neighbor in track $k+1$, we additionally put the following terms respectively,
\begin{align}\label{eqn:H_valid-last-term}
\ket{\,\twocellshor{\leftend}{\blankR}\,}_{1,2}\bra{\,\twocellshor{\leftend}{\blankR}\,}_{1,2},\qquad\ket{\,\twocellshor{\blankL}{\rightend}\,}_{1,2}\bra{\,\twocellshor{\blankL}{\rightend}\,}_{1,2}.
\end{align}
We define $H_{\mathrm{valid}}$ to be the sum of all these terms enforcing valid configurations.

\begin{lemma}\label{lem:H_valid-property}
The nullspace of \( H_{\mathrm{valid}} \) corresponds to the subspace \(\mathcal{S}(\mathcal{C})\), where $\mathcal{C}$ is the set of all valid configurations as defined in Definition~\ref{def:valid-configuration}.
\end{lemma}
\begin{proof}
The proof proceeds by observing that for each basis state that corresponds to a specific configuration of the $k+2$ tracks together with ``track 0'', if it is not a valid configuration as per Definition~\ref{def:valid-configuration}, it must violates at least one Hamiltonian term defined above from Eq.~\ref{eqn:track0-forbidden-0} to Eq.~\ref{eqn:H_valid-last-term}.
\end{proof}

\paragraph{Local Hamiltonian terms enforcing transition rules.}

The construction for the Hamiltonian terms enforcing transition rules follows that in Section 5 of \cite{gottesman2009quantum}, with the primary difference being our use of a distinct clock construction.

Note that each valid configuration, according to the transition rules, has at most one a preceding configuration and one subsequent configuration. More specifically, if a configuration has its clock tracks being
\begin{center}
\begin{small}
\begin{tabular}{|@{}c@{}|@{}c@{}|@{}c@{}|}
\hline
$\leftend$ &
\begin{tabular}{@{}c@{}|@{}c@{}|clc|@{}c@{}|@{}c@{}}
\arrRzero & \blankR & $\cdots$ & ``Track 0'': ``Clock track 0'' (second hand) & $\cdots$ & \blankR &  \blankR \\
\hline
\arrR & \blankR  & $\cdots$ & Track 1: Clock track 1 (minute hand) & $\cdots$ & \blankR &  \blankR \\
\hline
\vdots & \vdots &  & \qquad\qquad\qquad\qquad\ \vdots &  & \vdots &  \vdots \\
\hline
\arrR & \blankR  & $\cdots$ & Track $k$: Clock track $k$ & $\cdots$ & \blankR &  \blankR \\
\hline
\end{tabular}
& $\rightend$ \\
\hline
\end{tabular}\,,
\end{small}
\end{center}
it does not have a preceding configuration and has one subsequent configuration. Similarly, \begin{center}
\begin{small}
\begin{tabular}{|@{}c@{}|@{}c@{}|@{}c@{}|}
\hline
$\leftend$ &
\begin{tabular}{@{}c@{}|@{}c@{}|clc|@{}c@{}|@{}c@{}}
\vdots & \vdots &  & \qquad\qquad\qquad\qquad\ \vdots &  & \vdots &  \vdots \\
\hline
\blankL & \blankL  & $\cdots$ & Track $k-2$: Clock track $k-4$ & $\cdots$ & \blankL &  \arrR \\
\hline
\blankL & \blankL & $\cdots$ & Track $k-1$: Clock track $k-3$ & $\cdots$ & \arrL &  \blankR \\
\hline
\blankL & \blankL  & $\cdots$ & Track $k$: Clock track $k-2$ & $\cdots$ & \blankL &  \arrR \\
\hline
\end{tabular}
& $\rightend$ \\
\hline
\end{tabular}\,,
\end{small}
\end{center}
it does not have a subsequent configuration and has one preceding configuration. Denote $\Sigma^n$ to be the set of all possible inputs to the original Turing machine $M$. Then, the set of valid configurations $\mathcal{C}$ can be divided into a series of disjoint subsets $\mathcal{C}_{w}$'s such that
\begin{align}
\bigcup_{w\in \Sigma^n}\mathcal{C}_w=\mathcal{C},
\end{align}
where each $\mathcal{C}_w$ contains all the configurations that $\Mnew$ goes through when the input is $w$.

\begin{definition}[Input branch]\label{defn:input-branch}
We refer to each $\mathcal{C}_w$ as an \textit{input branch}, which contains all the configurations of the new TM $\Mnew$ for a specific input $w$. 
\end{definition}
Note that for the $t$-th configuration $c_w^{(t)}$, it admits the following tensor product form
\begin{align}\label{eqn:c_w^t-tensor-product}
\ket{c_w^{(t)}}=\ket{\mathsf{Clock}(t)}\otimes\ket{\mathsf{Track}(t,w)_{(k+1,k+2)}},
\end{align}
where $\ket{\mathsf{Clock}(t)}$ is the configuration of the clock tracks and is independent from $w$, and $\ket{\mathsf{Track}(t,w)_{(k+1,k+2)}}$ is the configuration of the last two tracks. Our goal here is to construct a local Hamiltonian $H_{\text{transition}}$ such that for any $w\in\Sigma^n$, $H_{\text{transition}}$ restricted to subspace $\mathcal{S}(\mathcal{C}_w)$ is proportional to

\begin{align}\label{eqn:transition-hamiltonian-restricted-classical}
\left(
\begin{array}{rrrrrrr}
\frac{1}{2} & -\frac{1}{2} &0 & & \cdots& & 0 \\ -\frac{1}{2} & 1 & -\frac{1}{2} & 0 &
\ddots & & \vdots\\ 0 & -\frac{1}{2} & 1 & -\frac{1}{2} & 0 & \ddots & \vdots\\ & \ddots & \ddots
& \ddots & \ddots & \ddots & \\ \vdots& & 0 & -\frac{1}{2} &1 & -\frac{1}{2}& 0 \\ & & & 0 & -\frac{1}{2} &1 & -\frac{1}{2} \\ 0& & \cdots& & 0&
-\frac{1}{2} & \frac{1}{2} \\
\end{array}
\right)
\end{align}
The matrix has dimensions $\mathscr{T} \times \mathscr{T}$, where $\mathscr{T} = \Theta(n^{k+1})$ represents the time complexity $\Mnew$. The rows and columns are indexed according to the order of configurations in $\mathcal{C}_w$.

\begin{lemma}\label{lem:H_transition-ground-state}
For any possible input $w\in\Sigma^n$ and any Hamiltonian $H_{\mathrm{transition}}$ that, when restricted to $\mathcal{S}(\mathcal{C}_w)$ is proportional to Eq.~\ref{eqn:transition-hamiltonian-restricted}, its ground states when restricted to $\mathcal{S}(\mathcal{C})$ can be expressed as
\begin{align}
\sum_{w\in\Sigma^n}\frac{\alpha_w}{\sqrt{\mathscr{T}}}\sum_{t=1}^{\mathscr{T}}\ket{c_w^{(t)}},
\end{align}
where $\alpha_w\in \mathbb{C}$, $\sum_{w\in\Sigma^n}|\alpha_w|^2=1$, and $\mathscr{T} = \Theta(n^{k+1})$ represents the time complexity $\Mnew$.
\end{lemma}
\begin{proof}
Note that any state $\ket{\psi} \in \mathcal{S}(\mathcal{C})$ can be written as
\[
\ket{\psi} = \sum_{w \in \Sigma^n} \alpha_w \sum_{t=1}^{\mathscr{T}} \beta_{t,w} \ket{c_w^{(t)}},
\]
where $\alpha_w, \beta_{t,w} \in \mathbb{C}$, with
\[
\sum_{w \in \Sigma^n} |\alpha_w|^2 = 1, \quad \sum_{t=1}^{\mathscr{T}} |\beta_{t,w}|^2 = 1.
\]
If $\ket{\psi}$ is a ground state of $H_{\mathrm{transition}}$, then
\[
H_{\mathrm{transition}} \ket{\psi} = \sum_{w \in \Sigma^n} \alpha_w H_{\mathrm{transition}} \left( \sum_{t=1}^{\mathscr{T}} \beta_{t,w} \ket{c_w^{(t)}} \right).
\]
Since $\mathcal{S}(\mathcal{C}_w)$ is preserved by $H_{\mathrm{transition}}$, i.e., $H_{\mathrm{transition}} (\mathcal{S}(\mathcal{C}_w)) \subseteq \mathcal{S}(\mathcal{C}_w)$, for any $\ket{\psi}$ with $H_{\mathrm{transition}} \ket{\psi} = 0$ and $\alpha_w \neq 0$, we must have $H_{\mathrm{transition}} \left( \sum_{t=1}^{\mathscr{T}} \beta_{t,w} \ket{c_w^{(t)}} \right) = 0$, which holds only if $\beta_{t,w} = 1/\sqrt{\mathscr{T}}$ for all $t, w$.
\end{proof}

The following lemma demonstrates that $H_{\mathrm{transition}}$ can be expressed as a sum of local terms.

\begin{lemma}\label{lem:H_transition-construction}
There exists a Hamiltonian $H_{\mathrm{transition}}$ with $\poly n$ local terms that, when restricted to the subspace $\mathcal{S}(\mathcal{C}_w)$, is proportional to Eq.~\ref{eqn:transition-hamiltonian-restricted-classical}.
\end{lemma}
\begin{proof}
We denote
\begin{align}
H_{w,t}\coloneqq\frac{1}{2}\big(\ket{c_w^{(t)}}\bra{c_w^{(t)}}-\ket{c_w^{(t)}}\bra{c_w^{(t+1)}}-\ket{c_w^{(t+1)}}\bra{c_w^{(t)}}+\ket{c_w^{(t+1)}}\bra{c_w^{(t+1)}}\big),\quad\forall w\in\Sigma^n.
\end{align}
Then for any $w$, the restricted Hamiltonian in Eq.~\ref{eqn:transition-hamiltonian-restricted-classical} in the subspace $\mathcal{S}(\mathcal{C}_w)$ can be expressed as
\begin{align}
    \frac{1}{2}\sum_{t=1}^{\mathscr{T}-1}H_{w,t}.
\end{align}
Note that for any $w$ and $t$, the term $H_{w,t}$ connects the $t$-th valid configuration to the $(t+1)$-th valid configuration. This is a local interaction since the difference between these two configurations is localized. Moreover, the positions where these local changes occur depend only on $t$, and are independent of $w$. Thus, there exists a local term $H_t$ that equals $H_{w,t}$ when restricted to the subspace $\mathcal{S}(\mathcal{C}_w)$. Given that there are $\poly n$ different values of $t$, we conclude that $H_{\mathrm{transition}}$ can be expressed as the sum of $\poly n$ number of local terms.
\end{proof}

\begin{lemma}[Lemma 3.5 of \cite{aharonov2005adiabatic}]
For any $w\in\Sigma^n$ and any Hamiltonian $H_{\mathrm{transition}}$ that, when restricted to $\mathcal{S}(\mathcal{C}_w)$ is proportional to Eq.~\ref{eqn:transition-hamiltonian-restricted-classical}, the spectral gap of the restriction of $H_{\mathrm{transition}}$ to $\mathcal{S}(\mathcal{C}_w)$ is at least $\Omega(\mathscr{T}^{-1})$, where $\mathscr{T}=\Theta(n^{k+1})$ is the time complexity of $\Mnew$.
\end{lemma}

\paragraph{Local Hamiltonian terms verifying the input.}
Given that the original TM has space complexity $n$, its input can be represented by an $n$-bit string $w^*=\big(\sigma_1^{*},\ldots,\sigma_n^{*}\big)\in\Sigma^n$. To verify that the Hamiltonian correctly encodes $w^{*}$, at the second position from left, we put the term
\begin{align}
\sum_{\small{\variable}\neq\small{\variable}_1}\ket{\variable}_1\bra{\variable}_1,
\end{align}
where
\begin{align}
\small{\variable}_1=\sixcellsvertical{\arrRzero}{\arrR}{\vdots}{\arrR}{q_0}{\sigma_1^*}\,.
\end{align}
Similarly, for any $1\leq j\leq n$, at the $(j+1)$-th position from left we put the term 
\begin{align}
\sum_{\small{\variable}\neq\small{\variable}_{j+1}}\ket{\variable}_{j+1}\bra{\variable}_{j+1},
\end{align}
where
\begin{align}
\variable_{j+1}=\sixcellsvertical{\arrRzero}{\blankR}{\vdots}{\blankR}{q_0}{\sigma_j^*}\,.
\end{align}
We denote the sum of these terms to be $H_{\mathrm{input}}$.
Hamiltonian terms interacts with the valid configurations in the initialization phase.

\paragraph{Ground state and spectral gap properties.} The following geometrical lemma by Kitaev is useful for proving the spectral gap properties of our Hamiltonian.
\begin{lemma}[Lemma 14.4 in \cite{kitaevbook}]\label{lem:kitaev-geometric}
Let $H_1,H_2$ be two Hamiltonians with
ground energies $a_1,a_2$, respectively.
Suppose that for both Hamiltonians the difference between the energy of the (possibly degenerate) ground space and the next highest eigenvalue is
larger than $\Lambda$,
and that the angle between the two ground spaces is $\theta$. Then
the ground energy of $H_1+H_2$ is at least
 $a_1+a_2+2\Lambda \sin^2(\theta/2)$.
\end{lemma}

\begin{theorem}\label{thm:ground-state-form}
The ground state of $H\coloneqq H_{\mathrm{valid}}+H_{\mathrm{transition}}+H_{\mathrm{input}}$ equals
\begin{align}
\frac{1}{\sqrt{\mathscr{T}}}\sum_{c\in\mathcal{C}_{w^*}}\ket{c}=\frac{1}{\sqrt{\mathscr{T}}}\sum_{t=1}^{\mathscr{T}}\ket{\mathsf{Clock}(t)}\otimes\ket{\mathsf{Track}(t,w^*)_{(k+1,k+2)}},
\end{align}
where $\mathcal{C}_{w^*}$ is the input branch corresponding to the input $w^*$, as per Definition~\ref{defn:input-branch}, and $\mathscr{T}=\Theta(n^{k+1})$ is the time complexity of $\Mnew$. Moreover, the spectral gap of $H$ is at least $\Omega(1/\poly n)$.
\end{theorem}
\begin{proof}
Note that \( H_{\mathrm{valid}} \), \( H_{\mathrm{transition}} \), and \( H_{\mathrm{input}} \) are all positive semidefinite. Consequently, the ground state of \( H \) lies in the intersection of their null spaces. By Lemma~\ref{lem:H_transition-ground-state}, any state in the intersection of the null spaces of $H_{\mathrm{valid}}$ and $H_{\mathrm{transition}}$ has the form
\begin{align}
\sum_{w\in\Sigma^n}\frac{\alpha_w}{\sqrt{\mathscr{T}}}\sum_{c\in\mathcal{C}_w}\ket{c}.
\end{align}
Moreover, note that for any $w\in\Sigma^n$ such that $w\neq w^*$, we have
\begin{align}
\Big(\sum_{c\in\mathcal{C}_w}\bra{c}\Big)H_{\mathrm{input}}\Big(\sum_{c\in\mathcal{C}_w}\ket{c}\Big)>0,
\end{align}
given that each term in $H_{\mathrm{input}}$ is positive semidefinite, and for any $j\in [n]$ with $w_j\neq w_j^*$, the $j$-th term in $H_{\mathrm{input}}$ satisfies
\begin{align}
\Big(\sum_{c\in\mathcal{C}_w}\bra{c}\Big)\bigg(\sum_{\small{\variable}\neq\small{\variable}_{j+1}}\ket{\variable}_{j+1}\bra{\variable}_{j+1}\bigg)\Big(\sum_{c\in\mathcal{C}_w}\ket{c}\Big)>0.
\end{align}
Hence, we can conclude that any state in the intersection of the null spaces of \( H_{\mathrm{valid}} \), \( H_{\mathrm{transition}} \), and \( H_{\mathrm{input}} \) has the form
\begin{align}
\frac{1}{\sqrt{\mathscr{T}}}\sum_{c\in\mathcal{C}_{w^*}}\ket{c},
\end{align}
where the form of each $\ket{c}\in\mathcal{C}_{w^*}$ is in Eq.~\ref{eqn:c_w^t-tensor-product}.

Regarding the spectral gap, note that the subspace spanned by all valid configurations, $\mathcal{S}(\mathcal{C})$, is preserved by $H$, meaning $H(\mathcal{S}(\mathcal{C})) \subseteq \mathcal{S}(\mathcal{C})$. Thus, the eigenstates of $H$ lie either in $\mathcal{S}(\mathcal{C})$ or its orthogonal complement, $\mathcal{S}(\mathcal{C})^\perp$. This allows us to analyze the spectra of $H_{\mathcal{S}(\mathcal{C})}$ and $H_{\mathcal{S}(\mathcal{C})^\perp}$ separately. Due to the term $H_{\mathrm{valid}}$ and the fact that all other terms are positive semidefinite, the ground energy of $H_{\mathcal{S}(\mathcal{C})^\perp}$ is at least 1. For $H_{\mathcal{S}(\mathcal{C})}$, observe that $\mathcal{S}(\mathcal{C})$ can be expressed as the direct sum of $|\Sigma|^n$ orthogonal subspaces ${\mathcal{S}_w}$, where each $\mathcal{S}_w$ is spanned by all valid configurations in $\mathcal{C}w$ corresponding to a specific input $w \in \Sigma^n$. Consequently, $H_{\mathcal{S}(\mathcal{C})}$ is block diagonal in the $\mathcal{S}_w$'s as shown below
\begin{align}
\mathbf{H} =
\begin{pmatrix}
H_{\mathcal{S}(\mathcal{C}_{w}^{(1)})} & 0 & 0 & \cdots  \\
0 & H_{\mathcal{S}(\mathcal{C}_{w}^{(2)})} & 0 & \cdots  \\
0 & 0 & \ddots & \cdots  \\
\vdots & \vdots & \vdots & \ddots 
\end{pmatrix}
\end{align}
Enforced by $H_{\mathrm{input}}$, the ground energy of $H_{\mathcal{S}(\mathcal{C}_{w})}$ for any $w\neq w^*$ is higher by at least $\Omega (1/\poly n)$ than that of $H_{\mathcal{S}(\mathcal{C}_{w^*})}$. As for $H_{\mathcal{S}(\mathcal{C}_{w^*})}$, we note that the angle $\theta$ between the two ground spaces of $H_{\mathrm{transition}}$ and $H_{\mathrm{input}}$ restricted to $\mathcal{S}(\mathcal{C}_{w^*})$ satisfies $\cos(\theta)\leq 1-\frac{1}{\poly n}$. Then by Lemma~\ref{lem:kitaev-geometric}, the spectral gap of $H_{\mathcal{S}(\mathcal{C}_{w^*})}$ is at least $\Omega(1/\poly n)$. Then, we can conclude that the spectral gap of $H$ is at least $\Omega(1/\poly n)$.

\end{proof}

\subsection{Local Hamiltonians for quantum Turing machines.}
In this section, we extend the Hamiltonian construction from Section~\ref{sec:classical-TM-Hamiltonian} to construct a local Hamiltonian whose ground state encodes the first $\Theta(n^k)$ steps of a quantum Turing machine $M=(Q, \Sigma, \delta, q_0, q_{\mathrm{acc}})$ with space complexity at most $n$. Similarly, we first construct a new $(k+2)$-track QTM $\Mnew$ that simulates the given QTM. The alphabet of this new QTM is the same as in Section~\ref{sec:classical-TM-Hamiltonian}, and the initial configuration is as shown in Figure~\ref{fig:initial-configuration-classical}, assuming a classical input. Starting from this initial configuration, the new QTM will evolve under a set of unitary transition rules before it halts. There are still two phases of transition rules: the initialization phase and the computation phase, with their roles similar to Section~\ref{sec:classical-TM-Hamiltonian}.

\subsubsection{Transition rules}\label{sec:transition-rules-quantum}
We focus only on the transition rules for states within the subspace $\mathcal{S}(\mathcal{C})$ spanned by all valid configurations $\mathcal{C}$. As shown later in Fact~\ref{fact:quantum-valid-configuration-garantee}, starting from the initial configuration in Figure~\ref{fig:initial-configuration-classical}, the new QTM remains in a superposition of valid configurations at every time step.

\paragraph{Transition rules of the clock tracks.}
The transition rules for the clock tracks in all phases are identical to those in Section~\ref{sec:classical-TM-Hamiltonian}. In other words, the clock track evolves independently as a classical object, regardless of whether the original Turing machine is classical or quantum.

\paragraph{Transition rules of the last two tracks}
The transition rules for the last two tracks in the initialization phase are identical to those in Section~\ref{sec:initialization-transition-rules}. As for the computation phase, similar to the classical case, the transition in the last two tracks are triggered by the movement of the head of the new Turing machine, i.e., the arrow on track 0, when it overlaps with the head symbol on track $k+1$ which marks the head location of the QTM being simulated. 

Since the transition function $\delta$ is no longer classical, the head symbol on track $k+1$ may be in a superposition of moving left and moving right. Similar to the approach in Section 5.6 of~\cite{gottesman2009quantum}, we handle the left and right head movements separately. According to \cite{bernstein1997quantum}, the state set $Q$ of the QTM $M$ can be partitioned into two subsets, $Q_L$ and $Q_R$, such that states in $Q_L$ are reached only by left moves, and those in $Q_R$ only by right moves. For each $q_R \in Q_R$, we introduce a new symbol $q_R'$ to denote an intermediate state before moving right. When the head is moving right and meets the head symbol on track $k+1$, we have the following transition rule
\begin{align}
\left|\,\sixcells{\arrR}{\blankR}{\blankL}{q}{\variable}{a}\,\right\rangle 
\rightarrow \sum_{a,b,q_L,q}\alpha_{\delta}(q,a,q_L,b,L)\left|\,\sixcells{\blankL}{\arrR}{q_L}{\blankR}{\variable}{b}\,\right\rangle
+\sum_{a,b,q_R,q}\alpha_{\delta}(q,a,q_R,b,R)\left|\,\sixcells{\blankL}{\arrR}{\blankL}{q_R'}{\variable}{b}\,\right\rangle.
\end{align}
When the head is moving left and meets the head symbol on track $k+1$, we have the following transition rule
\begin{align}
\left|\,\sixcells{\blankL}{\arrL}{\blankL}{q_R'}{\variable}{\variable}\,\right\rangle\rightarrow
\left|\,\sixcells{\arrL}{\blankR}{\blankL}{q_R}{\variable}{\variable}\,\right\rangle.
\end{align}
Otherwise, the head just sweeps through leaving the last two tracks unchanged.
\begin{align}
\left|\,\sixcells{\arrR}{\blankR}{\blankL}{\blankL}{\variable}{\variable}\,\right\rangle \rightarrow \left|\,\sixcells{\blankL}{\arrR}{\blankL}{\blankL}{\variable}{\variable}\,\right\rangle,\quad
\left|\,\sixcells{\arrR}{\blankR}{\blankR}{\blankR}{\variable}{\variable}\,\right\rangle \rightarrow \left|\,\sixcells{\blankL}{\arrR}{\blankR}{\blankR}{\variable}{\variable}\,\right\rangle,\quad
\left|\,\sixcells{\blankL}{\arrL}{\blankL}{\blankL}{\variable}{\variable}\,\right\rangle \rightarrow \left|\,\sixcells{\arrL}{\blankR}{\blankL}{\blankL}{\variable}{\variable}\,\right\rangle,\quad
\left|\,\sixcells{\blankL}{\arrL}{\blankR}{\blankR}{\variable}{\variable}\,\right\rangle \rightarrow \left|\,\sixcells{\arrL}{\blankR}{\blankR}{\blankR}{\variable}{\variable}\,\right\rangle.
\end{align}

\begin{fact}
The transition rules introduced throughout Section~\ref{sec:transition-rules-quantum} are unitary when restricted to the subspace spanned by valid configurations $\mathcal{C}$.
\end{fact}

\begin{fact}\label{fact:quantum-valid-configuration-garantee}
For any input $w \in \Sigma^n$ to the original QTM $M$, starting from the initial configuration in Figure~\ref{fig:initial-configuration-classical} and evolving according to the transition rules, the new QTM $\Mnew$ will be in a superposition over valid configurations at every time step, as defined in Definition~\ref{def:valid-configuration}.
\end{fact}

\subsubsection{Construction of the local Hamiltonian}

We now construct the local Hamiltonian whose ground state encodes the history of the new QTM. This ground state will be a uniform superposition over all configurations of the new QTM, starting from the initial configuration  Figure~\ref{fig:initial-configuration-classical}. Similarly as Section~\ref{sec:classical-TM-Hamiltonian}, each configuration is represented using $\Theta(n)$ qudits, corresponding to the $\Theta(n)$ different positions in the configuration. The local dimension is $k+3$, accounting for the number of tracks, including `track 0' which stores the head's position and state. The local Hamiltonian will consist of three sets of terms: one to enforce valid configurations, one to enforce transition rules, and one to verify the input.

\paragraph{Local Hamiltonian terms enforcing valid configurations.}
The local Hamiltonian terms enforcing valid configurations are exactly the same as the classical case presented in Section~\ref{sec:classical-TM-Hamiltonian}. We use $H_{\mathrm{valid}}$ to denote the sum of these terms. As shown in Lemma~\ref{lem:H_valid-property}, the nullspace of \( H_{\mathrm{valid}} \) corresponds to the subspace spanned by all valid configurations in \(\mathcal{C}\), as defined in Definition~\ref{def:valid-configuration}.

\paragraph{Local Hamiltonian terms enforcing transition rules.}
The local Hamiltonian terms enforcing transition rules are almost the same as the classical case presented in Section~\ref{sec:classical-TM-Hamiltonian}. The only difference is that the transition rules are quantum, and for an input $w$ at time step $t$, the configuration $\ket{c_w^{(t)}}$ is a quantum superposition of valid configurations. Nevertheless, given that the transition function in the QTM is unitary, the subspace $\mathcal{S}(\mathcal{C})$ spanned by valid can be divided into the direct sum of disjoint subspaces $\mathcal{S}(\mathcal{C}_{w})$'s such that
\begin{align}
\bigoplus_{w\in \Sigma^n}\mathcal{S}_{w}=\mathcal{S}(\mathcal{C}),
\end{align}
where each $\mathcal{S}_{w}$ satisfies
\begin{align}
\mathcal{S}_{w}=\mathrm{span}\{\ket{c_w^{(1)}},\ldots,\ket{c_w^{(\mathcal{T})}}\}.
\end{align}
Similarly, we construct a local Hamiltonian $H_{\mathrm{transition}}$ such that for any $w\in\Sigma^n$, $H_{\mathrm{transition}}$ restricted to $\mathcal{S}_w$ is proportional to 
\begin{align}\label{eqn:transition-hamiltonian-restricted}
\left(
\begin{array}{rrrrrrr}
\frac{1}{2} & -\frac{1}{2} &0 & & \cdots& & 0 \\ -\frac{1}{2} & 1 & -\frac{1}{2} & 0 &
\ddots & & \vdots\\ 0 & -\frac{1}{2} & 1 & -\frac{1}{2} & 0 & \ddots & \vdots\\ & \ddots & \ddots
& \ddots & \ddots & \ddots & \\ \vdots& & 0 & -\frac{1}{2} &1 & -\frac{1}{2}& 0 \\ & & & 0 & -\frac{1}{2} &1 & -\frac{1}{2} \\ 0& & \cdots& & 0&
-\frac{1}{2} & \frac{1}{2} \\
\end{array}
\right).
\end{align}
The matrix has dimensions $\mathscr{T} \times \mathscr{T}$, where $\mathscr{T} = \Theta(n^{k+1})$ represents the time complexity $\Mnew$. The rows and columns are indexed according to the order of configurations $\{\ket{c_w^{(1)}},\ldots,\ket{c_w^{(\mathscr{T})}}\}$.

\begin{lemma}\label{lem:H_transition-construction-quantum}
There exists a Hamiltonian $H_{\mathrm{transition}}$ with $\poly n$ local terms that, when restricted to the subspace $\mathcal{S}(\mathcal{C}_w)$, is proportional to Eq.~\ref{eqn:transition-hamiltonian-restricted}.
\end{lemma}
\begin{proof}
The proof essentially follows the proof of Lemma~\ref{lem:H_transition-construction}.
\end{proof}

\paragraph{Local Hamiltonian terms verifying the input.}
The local Hamiltonian terms enforcing valid configurations are exactly the same as the classical case presented in Section~\ref{sec:classical-TM-Hamiltonian}. We use $H_{\mathrm{input}}$ to denote the sum of these terms.

\paragraph{Ground state and spectral gap properties.}
\begin{theorem}\label{thm:ground-state-form-quantum}
The ground state of $H\coloneqq H_{\mathrm{valid}}+H_{\mathrm{transition}}+H_{\mathrm{input}}$ equals
\begin{align}
\frac{1}{\sqrt{\mathscr{T}}}\sum_{t=1}^{\mathscr{T}}\ket{c_{w^*}^{(t)}}=\frac{1}{\sqrt{\mathscr{T}}}\sum_{t=1}^{\mathscr{T}}\ket{\mathsf{Clock}(t)}\otimes\ket{\mathsf{Track}(t,w^*)_{(k+1,k+2)}},
\end{align}
where $\mathcal{C}_{w^*}$ is the input branch corresponding to the input $w^*$, as per Definition~\ref{defn:input-branch}, and $\mathscr{T}=\Theta(n^{k+1})$ is the time complexity of $\Mnew$. Moreover, the spectral gap of $H$ is at least $\Omega(1/\poly n)$.
\end{theorem}

\begin{proof}
The proof essentially follows the proof of Theorem~\ref{thm:ground-state-form}.
\end{proof}

\section{Putting everything together}\label{sec:everything-together}

In this section, we combine the techniques and the results introduced in previous sections and show that the LGSES problem is cryptographically hard for both 1D Hamiltonians (Section~\ref{sec:LGSES-1D}) and 2D Hamiltonians (Section~\ref{sec:LGSES-2D})

\subsection{1D case}\label{sec:LGSES-1D}
\noindent In this subsection, we show how to obtain two families of $1$D Hamiltonians on a $1$D line of $n+2$ qudits, one whose ground state has entanglement entropy of order $\tilde{O}(n)$ and the other whose ground state has entanglement entropy of order $\polylog n$, with respect to most cuts across the $1$D line. Our proof proceeds through the following steps.:
\begin{itemize}
\item We begin with two $n$-qudit public key pseudoentangled states across multiple cuts, following the construction outlined in Section~\ref{sec:TM-public-key-pseudoentanglement}, and consider the quantum Turing machines (QTMs) that prepare them. We choose $m=n/\polylog n$ to substitute the parameter $n$ in Section~\ref{sec:TM-public-key-pseudoentanglement} such that the resulting QTM have space complexity $n$ and time complexity $\poly m=\poly n$.
\item We then apply our quantum Turing machine to Hamiltonian construction in Section~\ref{sec:QTM-to-Hamiltonian} to create two families of 1D Hamiltonians, embedding the history of the quantum Turing machine in their ground states. Specifically, we construct a Hamiltonian $H$ whose ground state $\psiground$, defined on $n+2$ qudits placed on one line that encodes the history of the QTM.
\item We demonstrate that, due to the padding, the entanglement structure of the 1D ground state across any cut mirrors that of the state $\ket{\psi^*}$ across the same cut, where $\ket{\psi^*}$ is the quantum state on the worktape in the padding phase.
\item Due to our pseudoentanglement construction, the state $\ket{\psi^*}$ exhibits either high or low entanglement when the cut is at a distance of $\omega(\log n)$ from the grid boundary. Using a continuity argument, we show that the ground state $\psiground$ inherits this high or low entanglement property, which is shown in the following theorem.
\end{itemize}

\begin{theorem}
For every $n \in \mathbb{N}$, there exist two families $\mathcal{H}^\lo_n$ and $\mathcal{H}^\hi_n$ of local Hamiltonians on $n$ qubits arranged on a 1D line with spectral gap $\Omega(1/\poly(n))$ and efficient procedures that sample (classical descriptions of) Hamiltonians from either family (denoted $H \leftarrow \mathcal{H}^\lo_n$ and $H \leftarrow \mathcal{H}^\hi_n$) with the following properties:
\begin{enumerate}
\item Hamiltonians sampled according to $H \leftarrow \mathcal{H}^\lo_n$ and $H \leftarrow \mathcal{H}^\hi_n$ are computationally indistinguishable under Assumption~\ref{dcra_subexp}.
\item  With overwhelming probability, the ground states of Hamiltonians $H \leftarrow \mathcal{H}^\lo_n$ have near area-law entanglement and Hamiltonians $H \leftarrow \mathcal{H}^\hi_n$ have near volume-law entanglement.
Formally, this means that for geometrically local cuts of size $r = \omega(\log n)$, the ground states of the Hamiltonians have entanglement entropy $O(\polylog n)$ or $\tilde{\Omega}(\min(r, n-r))$, respectively.
\end{enumerate}
\end{theorem}
\begin{proof}
Let $\mathcal{M}^\lo_n = \{M_{\mathsf{key}}\}_{\mathsf{key} \in \mathcal{K}^\lo_n}$ and $\mathcal{M}^\hi_n = \{ M_\mathsf{key}\}_{\mathsf{key} \in \mathcal{K}^\hi_n}$ be two ensembles of quantum Turing machines preparing public key pseudoentangled states across geometrically local cuts from Section~\ref{sec:TM-public-key-pseudoentanglement}. For a $\mathsf{key} \in \mathcal{K}^\lo_n \cup \mathcal{K}^\hi_n$, we denote by $H_\mathsf{key}$ the local Hamiltonian from Section~\ref{sec:QTM-to-Hamiltonian}. We then choose $\mathcal{H}^\lo_n = \{H_\mathsf{key}\}_{\mathsf{key} \in \mathcal{K}^\lo_n}$ and $\mathcal{H}^\hi_n = \{H_\mathsf{key}\}_{\mathsf{key} \in \mathcal{K}^\hi_n}$.
It follows from our pseudoentanglement construction and the padded history state Hamiltonian that these are efficiently sampleable families of local Hamiltonians, noting that each term has a classical description (as a list of matrix entries) of size $\poly(n)$ and there are $\poly(n)$ terms, so the full Hamiltonians has an explicit classical description of size $\poly(n)$, too.
The statement about the spectral gap follows from Theorem~\ref{thm:ground-state-form-quantum}. The first property claiming that $H \leftarrow \mathcal{H}^\lo_n$ and $H \leftarrow \mathcal{H}^\hi_n$ are computationally indistinguishable follows from Proposition~\ref{prop:pk-state-indistinguishability}.

To prove the second property, first consider a Hamiltonian $H \leftarrow \mathcal{H}^\lo_n$.
Let $|\psi_{\mathrm{ground}}\rangle$ be its ground state. By Theorem~\ref{thm:ground-state-form-quantum}, we have 
\begin{align}
|\psi_{\mathrm{ground}}\rangle=\frac{1}{\mathscr{T}_{\mathrm{total}}}\sum_{t=1}^{\mathscr{T}_{\mathrm{total}}}\ket{\mathsf{Clock}(t)}\otimes\ket{\mathsf{Track}(t,w^*)_{(k+1,k+2)}},
\end{align}
where we have $\mathscr{T}_{\mathrm{total}}=\Theta(n^{k+1})$. Moreover, by Lemma~\ref{lem:DCRA-state-low}, for any $t\geq 2n\mathscr{T}^*$, we have
\begin{align}
\ket{\mathsf{Track}(t,w^*)_{(k+1,k+2)}}=\ket{\psi_{\lo}}\otimes \ket{\mathsf{Auc}(t)},
\end{align}
where $\ket{\psi_{\lo}}$  has near area-law entanglement, and the qudits consisting $\ket{\psi_{\lo}}$ is evenly distributed across the line, as described in Section~\ref{sec:TM-public-key-pseudoentanglement}. We define
\begin{align}
|\psi_{\mathrm{ground}}'\rangle
&\coloneqq\frac{1}{\sqrt{\mathscr{T}_{\mathrm{total}}-2n\mathscr{T}^*}}\sum_{t=2n\mathscr{T}^*+1}^{\mathscr{T}_{\mathrm{total}}}\ket{\mathsf{Clock}(t)}\otimes\ket{\mathsf{Track}(t,w^*)_{(k+1,k+2)}}\nonumber\\
&=\ket{\psi_{\lo}}\otimes \frac{1}{\sqrt{\mathscr{T}_{\mathrm{total}}-\mathscr{T}^*}}\sum_{t=2n\mathscr{T}^*+1}^{\mathscr{T}_{\mathrm{total}}}\ket{\mathsf{Clock}(t)}\otimes\ket{\mathsf{Auc}(t)},
\end{align}
which satisfies
\begin{align}
\||\psi_{\mathrm{ground}}'\rangle-|\psi_{\mathrm{ground}}\rangle\|\leq \sqrt{\frac{2n\mathscr{T}^*}{\mathscr{T}_{\mathrm{total}}}}\leq\frac{1}{n^2}.
\end{align}
Moreover, for any geometrically local cut $(A,B)$, we have
\begin{align}
S((\psi_{\mathrm{ground}}')_A)=O(S(\psi_{\lo})_A)=\polylog n
\end{align}
Then by Lemma~\ref{lem:continuity_vonneumannentropy}, we can derive that
\begin{align}
S((\psi_\mathrm{ground})_A)\leq S((\psi_{\mathrm{ground}}')_A)+2n\||\psi_{\mathrm{ground}}'\rangle-|\psi_{\mathrm{ground}}\rangle\|=\polylog n.
\end{align}
By applying Lemma~\ref{lem:DCRA-state-high} and a similar argument, when we consider $H \leftarrow \mathcal{H}_n^\hi$ and a bipartition $(r, n-r)$,
\begin{equation*}
\label{eq:low-everthing}
S((\psi_{\mathrm{ground}})_A) = \tilde{\Omega}(\text{min}(r, n-r)).
\end{equation*}
This completes the proof.
\end{proof}

\subsection{2D case}\label{sec:LGSES-2D}
\noindent In this subsection, we show how to obtain two families of $2$D Hamiltonians on a $2$D grid of $n$ qudits, one whose ground state has almost volume law entanglement and the other whose ground state has almost area law entanglement. Our proof proceeds through the following steps.:
\begin{itemize}
\item We begin with two $n$-qudit states that consists of $\sqrt{n}$ identical copies of $\sqrt{n}$ public key pseudoentangled states, following the construction outlined in Section~\ref{sec:lwetm}, and consider the quantum Turing machines (QTMs) that prepare them. We choose $m=\sqrt{n}/\polylog n$ to substitute the parameter $n$ in Section~\ref{sec:lwetm} such that the resulting QTM has space complexity $n$ and time complexity $\poly m=\poly n$.
\item We then apply our quantum Turing machine to Hamiltonian construction in Section~\ref{sec:QTM-to-Hamiltonian} to create two families of Hamiltonians, embedding the history of the quantum Turing machine in their ground states. Specifically, we construct a Hamiltonian $H$ whose ground state $\psiground$, defined on $n$ qudits placed on one line that encodes the history of the QTM. Then, we fill the 2D grid of $n$-qudits as a 1D ``snake'' of $\psiground$, similar to Section D.5 of~\cite{aaronson2024quantum} and Section 3.4 of~\cite{bouland2024public}, as shown in Figure~\ref{fig:snake}.
\item We demonstrate that, due to the padding, the entanglement structure of the 2D ground state across any cut mirrors that of the state $\ket{\psi^*}$ across the same cut, where $\ket{\phi^*}$ is the quantum state on the worktape in the padding phase.
\item Due to our pseudoentanglement construction, the state $\ket{\phi^*}$ exhibits either almost volume law or almost area law entanglement. Using a continuity argument, we show that the ground state $\psiground$ inherits this high or low entanglement property, which is shown in Theorem~\ref{thm:2D-main}
\end{itemize}

\begin{figure*}[htbp!]
\centering
\includegraphics[width=0.3\textwidth]{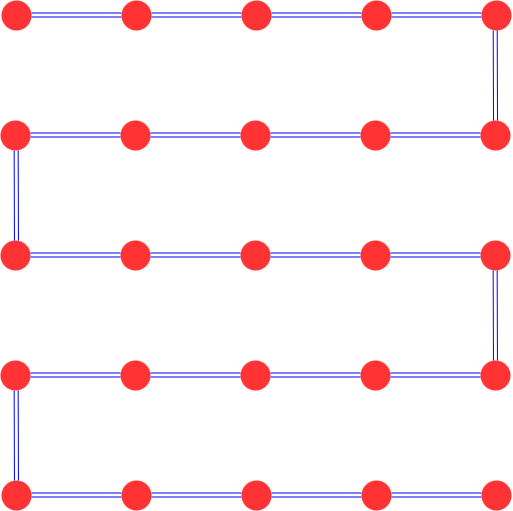}
\caption{Arranging an $n$-qubit state on a $\sqrt{n}\times \sqrt{n}$ grid (Figure 1 of~\cite{bouland2024public}).}
\label{fig:snake}
\end{figure*}

\begin{fact}
Consider the QTM defined in Section~\ref{sec:lwetm}, we consider the $n$-qudit quantum state $\ket{\phi^*}$ on the worktape in the padding phase. Put on a 2D grid of qudits of size $\sqrt{n}\times\sqrt{n}$, where each qudit is the tensor product of an indicator qubit and a content qudit, $\ket{\phi^*}$ can be decomposed into
\begin{align}
\ket{\phi^*}=\frac{1}{\sqrt{2}}\ket{\phi_0^*}+\frac{1}{\sqrt{2}}\ket{\phi_1^*},
\end{align}
where 
\begin{align}
\ket{\phi_0^*}=\ket{0}^{\otimes n}_{\mathrm{indicator\ qubits}}\otimes\ket{\text{content qudits}}_0,
\end{align}
with $\ket{\text{content qudits}}_0$ being putting $\sqrt{n}$ copies of some $\sqrt{n}$-qudit state $\ket{\psi}$ horizontally in parallel where each $\ket{\psi}$ occupies $\sqrt{n}$ qubits in a row, as in Figure~\ref{fig:content_qudits_0}.
Similarly,
\begin{align}
\ket{\phi_1^*}=\ket{1}^{\otimes n}_{\mathrm{indicator\ qubits}}\otimes\ket{\text{content qudits}}_1,
\end{align}
with $\ket{\text{content qudits}}_1$ being putting $\sqrt{n}$ copies of $\ket{\psi}$ vertically in parallel where each $\ket{\psi}$ occupies $\sqrt{n}$ qubits in a row, as in Figure~\ref{fig:content_qudits_1}.
\end{fact}

\begin{figure}[htbp]
\centering
\begin{tikzpicture}
    \node at (-3.2, -3) {$\ket{0}^{\otimes n} \otimes$};
    \foreach \i in {1, 2, 4}
    {
        \draw (-2, -\i) -- (0, -\i);
        \node at (0.5, -\i) {$| \psi \rangle$};
        \draw (1, -\i) -- (3, -\i);
    }
    
    \node at (0.5, -3) {$\vdots$};
    
    \draw (-2, -5) -- (0, -5);
    \node at (0.5, -5) {$| \psi \rangle$};
    \draw (1, -5) -- (3, -5);

    \node[scale=1] at (3.4, -3) {\scalebox{0.7}[0.73]{$\left\}\rule{0pt}{8.5em}\right.$}};
    
    \node at (6.8, -3) {$\sqrt{n}$ copies of $\ket{\psi}$ placed horizontally};
    
\end{tikzpicture}
\caption{Illustrations of $\ket{\phi_0^*}$.}
\label{fig:content_qudits_0}
\end{figure}
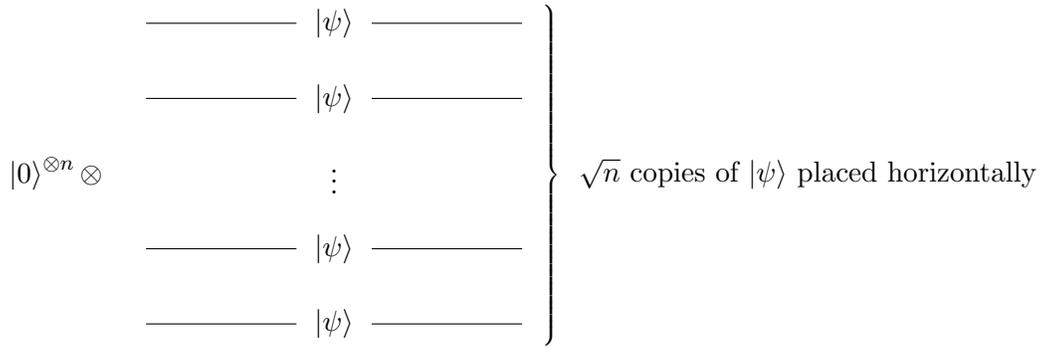

\begin{figure}[htbp]
\centering
\begin{tikzpicture}
    \node at (-0.7, -3) {$\ket{0}^{\otimes n} \otimes$};
    
    \foreach \i/\x in {1/0.5, 2/1.5}
    {
        \draw (\x, -1) -- (\x, -2.5);
        \node at (\x, -3) {$| \psi \rangle$};
        \draw (\x, -3.5) -- (\x, -5);
    }

    \node at (2.5, -3) {$\cdots$};
    \foreach \i/\x in {1/3.5, 2/4.5}
    {
        \draw (\x, -1) -- (\x, -2.5);
        \node at (\x, -3) {$| \psi \rangle$};
        \draw (\x, -3.5) -- (\x, -5);
    }
    \node[scale=1] at (5.2, -3) {\scalebox{0.7}[0.73]{$\left\}\rule{0pt}{8.5em}\right.$}};
    
    \node at (8.4, -3) {$\sqrt{n}$ copies of $\ket{\psi}$ placed vertically};
    
\end{tikzpicture}
\caption{Illustrations of $\ket{\phi_1^*}$.}
\label{fig:content_qudits_1}
\end{figure}
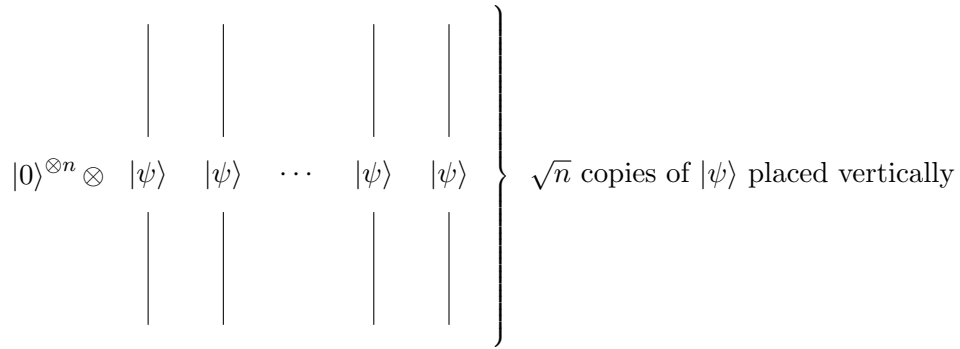

We analyze the entanglement properties of $\psiground$ by analyzing the entanglement properties of $\ket{\phi^*}$, in the case where the grid is bipartitioned into red and blue points, representing two distinct sets of qudits, with no restrictions on the form of the bipartition. Examples are shown in Figure~\ref{fig:bipartition-examples}. 

\begin{figure}[htbp]
    \centering
    \begin{subfigure}[t]{0.3\textwidth} 
        \centering
        \begin{tikzpicture}
            \foreach \x/\xindex in {0/0, 0.6/1, 1.2/2, 1.8/3, 2.4/4, 3.0/5, 3.6/6, 4.2/7}     \foreach \y in {0, 0.6, 1.2, 1.8, 2.4, 3.0, 3.6, 4.2} 
                {                    \ifnum\xindex<4
                        \fill[red] (\x,\y) circle (2pt); 
                    \else
                        \fill[blue] (\x,\y) circle (2pt); 
                    \fi
                }
        \end{tikzpicture}
        \caption{Left half red, right half blue.}
    \end{subfigure}
    \hfill
    \begin{subfigure}[t]{0.3\textwidth} 
        \centering
        \begin{tikzpicture}
            \foreach \x in {0, 0.6, 1.2, 1.8, 2.4, 3.0, 3.6, 4.2} 
                \foreach \y in {0, 0.6, 1.2, 1.8, 2.4, 3.0, 3.6, 4.2} 
                {
                    \fill[red] (\x,\y) circle (2pt); 
                }
           \fill[blue] (1.2, 0.6) circle (2pt);
            \fill[blue] (1.2, 1.2) circle (2pt); 
            \fill[blue] (0.6, 1.2) circle (2pt);
            \fill[blue] (1.8, 1.2) circle (2pt);
            \fill[blue] (1.8, 1.8) circle (2pt);
            \fill[blue] (2.4, 1.8) circle (2pt);
            \fill[blue] (3.0, 1.8) circle (2pt);
            \fill[blue] (3.0, 1.2) circle (2pt);
            \fill[blue] (3.6, 1.2) circle (2pt);
            \fill[blue] (3.6, 1.8) circle (2pt);
            \fill[blue] (4.2, 1.8) circle (2pt);
            \fill[blue] (4.2, 2.4) circle (2pt);
            \fill[blue] (3.6, 2.4) circle (2pt);
            \fill[blue] (3.0, 2.4) circle (2pt);
            \fill[blue] (2.4, 2.4) circle (2pt);
            \fill[blue] (2.4, 3.0) circle (2pt);
            \fill[blue] (3.0, 3.0) circle (2pt);
            \fill[blue] (3.6, 3.0) circle (2pt);
            \fill[blue] (3.6, 3.6) circle (2pt);
            \fill[blue] (3.0, 3.6) circle (2pt);
            \fill[blue] (2.4, 3.6) circle (2pt);
            \fill[blue] (1.8, 3.6) circle (2pt);
            \fill[blue] (1.8, 3.0) circle (2pt);
            \fill[blue] (1.8, 2.4) circle (2pt);
            \fill[blue] (1.2, 2.4) circle (2pt);
            \fill[blue] (0.6, 2.4) circle (2pt);
            \fill[blue] (0.6, 3.6) circle (2pt);
            \fill[blue] (0.6, 4.2) circle (2pt);
            \fill[blue] (1.8, 4.2) circle (2pt);
            \fill[blue] (2.4, 4.2) circle (2pt);
            \fill[blue] (1.2, 3.0) circle (2pt);
            \fill[blue] (1.2, 3.6) circle (2pt);
            \fill[blue] (2.4, 1.2) circle (2pt);
        \end{tikzpicture}
        \caption{Random pattern of blue and red dots, each forming distinct continuous clusters.}
    \end{subfigure}
    \hfill
    \begin{subfigure}[t]{0.3\textwidth} 
    \centering
    \begin{tikzpicture}
        \foreach \x in {0, 0.6, 1.2, 1.8, 2.4, 3.0, 3.6, 4.2}
            \foreach \y in {0, 0.6, 1.2, 1.8, 2.4, 3.0, 3.6, 4.2}
            {
                \fill[red] (\x,\y) circle (2pt);             }
        \fill[blue] (0.6, 0) circle (2pt);
        \fill[blue] (1.2, 0) circle (2pt);
        \fill[blue] (1.8, 0.6) circle (2pt);
        \fill[blue] (3.0, 0.6) circle (2pt);
        \fill[blue] (3.6, 0.6) circle (2pt);
        \fill[blue] (4.2, 0.6) circle (2pt);
        \fill[blue] (0, 1.2) circle (2pt);
        \fill[blue] (1.2, 1.2) circle (2pt);
        \fill[blue] (2.4, 1.2) circle (2pt);
        \fill[blue] (3.0, 1.2) circle (2pt);
        \fill[blue] (4.2, 1.2) circle (2pt);
        \fill[blue] (0.6, 1.8) circle (2pt);
        \fill[blue] (2.4, 1.8) circle (2pt);
        \fill[blue] (3.0, 1.8) circle (2pt);
        \fill[blue] (3.6, 1.8) circle (2pt);
        \fill[blue] (0, 2.4) circle (2pt);
        \fill[blue] (1.8, 2.4) circle (2pt);
        \fill[blue] (2.4, 2.4) circle (2pt);
        \fill[blue] (3.6, 2.4) circle (2pt);
        \fill[blue] (4.2, 2.4) circle (2pt);
        \fill[blue] (0, 3.0) circle (2pt);
        \fill[blue] (0.6, 3.0) circle (2pt);
        \fill[blue] (2.4, 3.0) circle (2pt);
        \fill[blue] (3.0, 3.0) circle (2pt);
        \fill[blue] (4.2, 3.0) circle (2pt);
        \fill[blue] (0.6, 3.6) circle (2pt);
        \fill[blue] (1.2, 3.6) circle (2pt);
        \fill[blue] (2.4, 3.6) circle (2pt);
        \fill[blue] (1.2, 4.2) circle (2pt);
        \fill[blue] (3.0, 4.2) circle (2pt);
        \fill[blue] (3.6, 4.2) circle (2pt);
    \end{tikzpicture}
    \caption{Random pattern of blue and red dots with no continuous guarantee.}
\end{subfigure}
    \caption{Examples of bipartitions of a 2D grid of 64 qudits.}
    \label{fig:bipartition-examples}
\end{figure}
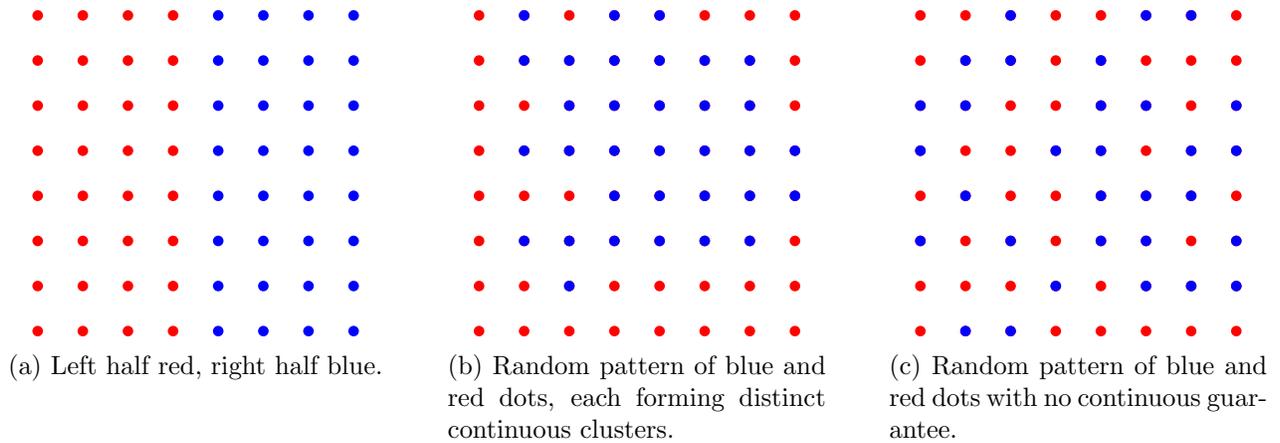
\subsubsection{High entanglement case}

\begin{fact}\label{fact-2D-cutwise-orthogonal}
$\ket{\phi_0^*}$ and $\ket{\phi_1^*}$ are cutwise orthogonal with respect to any bipartition $(A,B)$ of the $n$-qudits.
\end{fact}

\begin{lemma}\label{lem:2D-entanglement-high}
If $\ket{\psi}$ is almost volume law entangled i.e., for any bipartition $(A,B)$ of the $\sqrt{n}$-qudits we have
\begin{align}
S(\psi_{A})=S(\psi_{B})=\tilde{\Omega}(\min\{|A|,|B|\}),
\end{align}
then $\ket{\phi^*}$ is also volume law entangled, i.e., for any bipartition $(A,B)$ of the $n$-qudits, we have 
\begin{align}
S(\phi_{A}^*)=S(\phi_{B}^*)=\tilde{\Omega}(\min\{|A|,|B|\})=\frac{1}{\polylog(\min\{|A|,|B|\})}.
\end{align}
\end{lemma}
\begin{proof}
We adopt the notion in Figure~\ref{fig:bipartition-examples} where we denote qudits in $A$ by blue dots and qudits in $B$ by red dots. WLOG we assume $n$ can be divided by 4 and $|A|\leq \frac{n}{2}\leq|B|$. For every $j,k\in[\sqrt{n}]$, we use $r_j$ to denote the number of blue dots in row $j$ and use $c_k$ to denote the number of blue dots in column $k$. We define $\beta\colon [n/2]\to\mathbb{R}$ as follows
\begin{align}
\beta(k)\coloneqq \frac{\min_{|A|=k}S(\psi_A)}{k},
\end{align}
then we have
\begin{align}
\beta(\min\{|A|,|B|\})=\frac{1}{\polylog(\min\{|A|,|B|\})}.
\end{align}
Moreover, we define
\begin{align}
\mathrmsmallR\coloneqq\{i\,|\,i\in[\sqrt{n}],r_i\leq\sqrt{n}/2\},\qquad\mathrmlargeR\coloneqq\{i\,|\,i\in[\sqrt{n}],r_i>\sqrt{n}/2\},
\end{align}

then we have
\begin{align}\label{eqn:row-entanglement-sum}
S(\phi^*_{0,A})
&\geq
\sum_{j\in[\sqrt{n}]}\beta(\min\{r_j,\sqrt{n}-r_j\})\min\{r_j,\sqrt{n}-r_j\}\nonumber\\
&=\sum_{j\in\mathrmsmallR}\beta(r_j)\cdot r_j+\sum_{j\in\mathrmlargeR}\beta(\sqrt{n}-r_j)\cdot (\sqrt{n}-r_j),
\end{align}
and
\begin{align}\label{eqn:column-entanglement-sum}
S(\phi^*_{1,A})\geq\sum_{j\in[\sqrt{n}]}\beta(\min\{c_j,\sqrt{n}-c_j\})\cdot\min\{c_j,\sqrt{n}-c_j\},
\end{align}
If $|A|\geq 100\sqrt{n}$, we prove the lemma by contradiction. Assume there exists a bipartition $(A,B)$ such that $S(\phi_{A}^*)\leq \gamma\beta(n)|A|$ for $\gamma=1/1000$. By Fact~\ref{fact-2D-cutwise-orthogonal} and Lemma~\ref{lem:cutwise-orthogonal-entanglement}, we have 
\begin{align}\label{eqn:assumed-entanglement-upper-bound}
S(\phi_{0,A}^*)\leq 2\gamma\beta(|A|)\cdot|A|,\qquad S(\phi_{1,A}^*)\leq 2\gamma\beta(|A|)\cdot|A|,
\end{align}
combined with Eq.~\ref{eqn:row-entanglement-sum}, we have 
\begin{align}\label{eqn:row-colors-volume}
\sum_{j\in\mathrmsmallR}r_j\leq 2\gamma|A|,\qquad\sum_{j\in\mathrmlargeR}(\sqrt{n}-r_j)\leq 2\gamma|A|.
\end{align}

Hence, 
\begin{align}
\sum_{j\in\mathrmlargeR}r_j=|A|-\sum_{j\in\mathrmsmallR}r_j\geq(1-2\gamma)|A|,
\end{align}
which leads to
\begin{align}
|\mathrmlargeR|\geq \frac{1}{\sqrt{n}}\sum_{j\in\mathrmlargeR}r_j=(1-2\gamma)\cdot\frac{|A|}{\sqrt{n}},
\end{align}
given that there can be at most $\sqrt{n}$ dots in each row. Focus on the region formed by all the columns, and since there are at least $\frac{\sqrt{n}}{2}+1$ blue dots in each row in $\mathrmlargeR$, we know that there must be at least $\sqrt{n}/4$ columns $k$ satisfying
\begin{align}\label{eqn:contributing-columns}
    \hat{c}_k\coloneqq |\{j\in\mathrmlargeR, (j,k)=\text{blue}\}|\geq\frac{|\mathrmlargeR|}{4}\geq(1-2\gamma)\cdot\frac{|A|}{4\sqrt{n}},
\end{align}
which satisfies $10\leq\hat{c}_k\leq\sqrt{n}/2$. By Eq.~\ref{eqn:row-colors-volume}, we have
\begin{align}\label{eqn:contributing-columns-volume-upper-bound}
\sum_{k\in[n]}(c_k-\hat{c}_k)\leq 2\gamma|A|.
\end{align}
Hence, there must be at least $\sqrt{n}/8$ columns $k$ satisfying both Eq.~\ref{eqn:contributing-columns} and $c_k\leq 3\sqrt{n}/4$, for otherwise we have
\begin{align}
\sum_{k\in[n]}(c_k-\hat{c}_k)\geq \frac{\sqrt{n}}{8}\cdot\Big(\frac{3\sqrt{n}}{4}\Big)\geq\frac{n}{10},
\end{align}
which contradicts to Eq.~\ref{eqn:contributing-columns-volume-upper-bound}. Then, by Eq.~\ref{eqn:column-entanglement-sum} we have
\begin{align}
S(\phi^*_{1,A})\geq\beta(n)\sum_{j\in[n]}\min\{c_j,\sqrt{n}-c_j\}\geq \frac{\sqrt{n}\beta(n)}{16}\cdot\frac{|A|}{4\sqrt{n}}\geq\frac{\beta(n)|A|}{128},
\end{align}
where the last inequality is obtained by only considering the entanglement from columns satisfying both Eq.~\ref{eqn:contributing-columns} and $c_k\leq 3\sqrt{n}/4$, which contradicts to Eq.~\ref{eqn:assumed-entanglement-upper-bound}. Thus, we can conclude that 
\begin{align}
S(\phi_{A}^*)=\gamma\beta(n)|A|=\tilde{\Omega}(\min\{|A|,|B|\}).
\end{align}
for any bipartition $(A,B)$.

Otherwise, if $|A|<100\sqrt{n}$, we denote $r_{\max}\coloneqq \max_{j\in[\sqrt{n}]}r_j$ and denote $j_{\max}\coloneqq\mathrm{argmax}_{j\in[\sqrt{n}]}r_j$. If $r_{\max}\leq\sqrt{n}/2$, we have
\begin{align}
S(\phi_{A}^*)=\Omega(S(\phi_{0,A}^*))=\tilde{\Omega}(|A|).
\end{align}
If $r_{\max}>\sqrt{n}/2$, there must be at least $\sqrt{n}/4$ columns whose number of blue dots is at least 1 but at most $400$, which leads to
\begin{align}
S(\phi_{A}^*)=\Omega(S(\phi_{1,A}^*))=\tilde{\Omega}(\sqrt{n})=\tilde{\Omega}(|A|).
\end{align}
\end{proof}

\begin{definition}[Area of a bipartition]
For any bipartition $(A,B)$ of $n$ qubits, the area of the bipartition of $A$, denoted by $\mathsf{Area}(A)$, is defined as the number of qubits in $A$ that are adjacent to $B$.
\end{definition}

\begin{lemma}\label{lem:2D-entanglement-low}
If $\ket{\psi}$ is almost area law entangled i.e., for any geometrically local bipartition $(A,B)$ of the $\sqrt{n}$ qudits we have
\begin{align}
S(\psi_A)=S(\psi_B)= \tilde{O}(\polylog n),
\end{align}
then $\ket{\phi^*}$ is also area law entangled, i.e., for any bipartition $(A,B)$ of the $n$-qudits, we have 
\begin{align}
S(\phi_{A}^*)=S(\phi_{B}^*)=\tilde{O}(\mathsf{Area}(A))=\tilde{O}(\mathsf{Area}(B)).
\end{align}
\end{lemma}
\begin{proof}
We adopt the notion in Figure~\ref{fig:bipartition-examples} where we denote qudits in $A$ by blue dots and qudits in $B$ by red dots. WLOG we assume $n$ can be divided by 4 and $|A|\leq \frac{n}{2}\leq|B|$. For any $j\in[\sqrt{n}]$, if we sweep from left to right, the color might change for multiple times. We use $\hat{r}_j$ to denote the number of color changes in row $j$. Similarly, for any $k\in[\sqrt{n}]$, we use $\hat{c}_k$ to denote the number of color changes in column $k$. Given that each color change corresponds to one point of $A$ being on the boundary, we have
\begin{align}
\sum_{j\in[\sqrt{n}]}\hat{r}_j\leq 2\mathsf{Area}(A),\quad\sum_{k\in[\sqrt{n}]}\hat{c}_k\leq 2\mathsf{Area}(A).
\end{align}
On the other hand, by the strong subadditivity of von Neumann entropy, we have
\begin{align}
S(\phi_{0,A}^*)\leq \sum_{j\in[\sqrt{n}]}\hat{r}_j\polylog (n),\quad S(\phi_{1,A}^*)\leq \sum_{k\in[\sqrt{n}]}\hat{c}_k\polylog (n),
\end{align}
which leads to
\begin{align}
S(\phi_A^*)\leq O\big(S(\phi_{0,A}^*)+S(\phi_{1,A}^*)\big)=\tilde{O}(\mathsf{Area}(A))=\tilde{O}(\mathsf{Area}(B)).
\end{align}
\end{proof}

\begin{theorem}\label{thm:2D-main}
For every $n \in \mathbb{N}$ with $\sqrt{n}\in\mathbb{N}$, there exist two families $\mathcal{H}^\lo_n$ and $\mathcal{H}^\hi_n$ of local Hamiltonians on $n$ qubits arranged on a 2D grid of size $\sqrt{n}\times \sqrt{n}$ with spectral gap $\Omega(1/\poly(n))$ and efficient procedures that sample (classical descriptions of) Hamiltonians from either family (denoted $H \leftarrow \mathcal{H}^\lo_n$ and $H \leftarrow \mathcal{H}^\hi_n$) with the following properties:
\begin{enumerate}
\item Hamiltonians sampled according to $H \leftarrow \mathcal{H}^\lo_n$ and $H \leftarrow \mathcal{H}^\hi_n$ are computationally indistinguishable under Assumption~\ref{lwe_assumption_subexp}.
\item  With overwhelming probability, the ground states of Hamiltonians $H \leftarrow \mathcal{H}^\lo_n$ have near area-law entanglement and Hamiltonians $H \leftarrow \mathcal{H}^\hi_n$ have near volume-law entanglement.
Formally, this means that for any bipartition $(A,B)$, the ground states of the Hamiltonians have entanglement entropy $\tilde{O}(\mathsf{Area}(A))= \tilde{O}(\mathsf{Area}(B))$ or $\tilde{\Omega}(\min\{|A|,|B|\})$, respectively.
\end{enumerate}
\end{theorem}

\begin{proof}
Let $\mathcal{M}^\lo_n = \{M_{\mathsf{key}}\}_{\mathsf{key} \in \mathcal{K}^\lo_n}$ and $\mathcal{M}^\hi_n = \{ M_\mathsf{key}\}_{\mathsf{key} \in \mathcal{K}^\hi_n}$ be two ensembles of quantum Turing machines preparing public key pseudoentangled states across geometrically local cuts from Section~\ref{sec:lwetm}. For a $\mathsf{key} \in \mathcal{K}^\lo_n \cup \mathcal{K}^\hi_n$, we denote by $H_\mathsf{key}$ the local Hamiltonian from Section~\ref{sec:QTM-to-Hamiltonian}. We then choose $\mathcal{H}^\lo_n = \{H_\mathsf{key}\}_{\mathsf{key} \in \mathcal{K}^\lo_n}$ and $\mathcal{H}^\hi_n = \{H_\mathsf{key}\}_{\mathsf{key} \in \mathcal{K}^\hi_n}$.
It follows from our pseudoentanglement construction and the padded history state Hamiltonian that these are efficiently sampleable families of local Hamiltonians, noting that each term has a classical description (as a list of matrix entries) of size $\poly(n)$ and there are $\poly(n)$ terms, so the full Hamiltonians has an explicit classical description of size $\poly(n)$, too.
The statement about the spectral gap follows from Theorem~\ref{thm:ground-state-form-quantum}. The first property claiming that $H \leftarrow \mathcal{H}^\lo_n$ and $H \leftarrow \mathcal{H}^\hi_n$ are computationally indistinguishable follows from Proposition~\ref{prop:lwe-pseudoentangled}.

To prove the second property, first consider a Hamiltonian $H \leftarrow \mathcal{H}^\lo_n$.
Let $|\psi_{\mathrm{ground}}\rangle$ be its ground state. By Theorem~\ref{thm:ground-state-form-quantum}, we have 
\begin{align}
|\psi_{\mathrm{ground}}\rangle=\frac{1}{\mathscr{T}_{\mathrm{total}}}\sum_{t=1}^{\mathscr{T}_{\mathrm{total}}}\ket{\mathsf{Clock}(t)}\otimes\ket{\mathsf{Track}(t,w^*)_{(k+1,k+2)}},
\end{align}
where we have $\mathscr{T}_{\mathrm{total}}=\Theta(n^{k+1})$. Moreover, by Lemma~\ref{lem:DCRA-state-low}, for any $t\geq 2n\mathscr{T}^*$, we have
\begin{align}
\ket{\mathsf{Track}(t,w^*)_{(k+1,k+2)}}=\ket{\psi_{\lo}}\otimes \ket{\mathsf{Auc}(t)},
\end{align}
where $\ket{\psi_{\lo}}$  has near area-law entanglement, and the qudits consisting $\ket{\psi_{\lo}}$ is evenly distributed across the line, as described in Section~\ref{sec:TM-public-key-pseudoentanglement}. We define
\begin{align}
|\psi_{\mathrm{ground}}'\rangle
&\coloneqq\frac{1}{\sqrt{\mathscr{T}_{\mathrm{total}}-2n\mathscr{T}^*}}\sum_{t=2n\mathscr{T}^*+1}^{\mathscr{T}_{\mathrm{total}}}\ket{\mathsf{Clock}(t)}\otimes\ket{\mathsf{Track}(t,w^*)_{(k+1,k+2)}}\nonumber\\
&=\ket{\psi_{\lo}}\otimes \frac{1}{\sqrt{\mathscr{T}_{\mathrm{total}}-\mathscr{T}^*}}\sum_{t=2n\mathscr{T}^*+1}^{\mathscr{T}_{\mathrm{total}}}\ket{\mathsf{Clock}(t)}\otimes\ket{\mathsf{Auc}(t)},
\end{align}
which satisfies
\begin{align}
\||\psi_{\mathrm{ground}}'\rangle-|\psi_{\mathrm{ground}}\rangle\|\leq \sqrt{\frac{2n\mathscr{T}^*}{\mathscr{T}_{\mathrm{total}}}}\leq\frac{1}{n^2}.
\end{align}
Moreover, for any cut $(A,B)$, we have
\begin{align}
S((\psi_{\mathrm{ground}}')_A)=O(S(\psi_{\lo})_A)=\tilde{O}(\mathrm{Area}(A))
\end{align}
by Lemma~\ref{lem:2D-entanglement-low}. Then by Lemma~\ref{lem:continuity_vonneumannentropy}, we can derive that
\begin{align}
S((\psi_\mathrm{ground})_A)\leq S((\psi_{\mathrm{ground}}')_A)+2n\||\psi_{\mathrm{ground}}'\rangle-|\psi_{\mathrm{ground}}\rangle\|=\tilde{O}(\mathrm{Area}(A))=\tilde{O}(\mathrm{Area}(B)).
\end{align}
By applying Lemma~\ref{lem:2D-entanglement-high} and a similar argument, when we consider $H \leftarrow \mathcal{H}_n^\hi$ and a bipartition $(r, n-r)$,
\begin{equation*}
S((\psi_{\mathrm{ground}})_A) = \tilde{\Omega}(\text{min}\{|A|,|B|\}).
\end{equation*}
This completes the proof.
\end{proof}

\section*{Acknowledgments}
We thank Anurag Anshu, Bill Fefferman, Soumik Ghosh, Jeongwan Haah, Sandy Irani, Zeph Landau, Tony Metger, Wilson Nguyen, and Umesh Vazirani for helpful discussions.
A.B., C.Z., and Z.Z.~were supported in part by the DOE QuantISED grant DE-SC0020360.  
A.B. was supported in part by the U.S. DOE Office of Science under Award Number DE-SC0020377 and by the AFOSR under grants FA9550-21-1-0392 and FA9550-24-1-0089.
C.Z.~was supported in part by the Shoucheng Zhang graduate fellowship.
Z.Z. was supported in part by the Stanford interdisciplinary graduate fellowship.
This work was done in part while the authors were visiting the Simons Institute for the Theory of Computing, supported by DOE QSA grant \#FP00010905

\newcommand{\etalchar}[1]{$^{#1}$}

\end{document}